\pdfoutput=1
\documentclass[10pt,acmsmall,screen]{acmart}
\setcopyright{none}

\usepackage{booktabs}
\usepackage{wrapfig}
\usepackage{stmaryrd}
\usepackage{scalerel}
\usepackage{amssymb}
\usepackage{bbding}
\usepackage{amsmath}
\usepackage{amsthm}
\usepackage{array}
\usepackage{multirow}
\usepackage{xparse}
\usepackage{pbox}
\usepackage{color}
\usepackage{xcolor,colortbl}
\usepackage{mathtools}
\usepackage[noend]{algpseudocode}
\usepackage{soul}
\usepackage{fancyvrb}
\usepackage{cancel}
\usepackage{placeins}
\usepackage{mathpartir}
\usepackage{subcaption}
\usepackage{relsize,enumitem}
\usepackage{colortbl}

\newcommand{\systemprefix}{\textsc{CG}~Calculus}
\newcommand{\parlang}{\textsc{CG$^{||}$}~Calculus}

\newcommand{\eagersys}{\mathtt{E}}

\newcommand{\nkey}{k}

\newcommand{\KS}[1]{\mathbf{KL}\tuple{#1}}
\newcommand{\K}{\resizebox{0.5\width}{\height}{$K$}}
\newcommand{\seqK}{K}

\newcommand{\ND}[1]{\mathbf{N}\tuple{#1}}

\newcommand{\kssetminus}{\pmb{\setminus}}
\newcommand{\kssubseteq}{\pmb{\subseteq}}

\newcommand{\config}{C}

\newcommand{\data}{B}

\newcommand{\station}{S}

\newcommand{\strictprocessor}{\omega}
\newcommand{\processor}{o}
\newcommand{\batchgroup}{U}

\newcommand{\ostream}{O}
\newcommand{\rstream}{R}

\newcommand{\fvs}{\ell}

\newcommand{\inode}{N}

\newcommand{\ep}{e}
\newcommand{\val}{v}

\newcommand{\configreduce}{\rightarrow}
\newcommand{\eagerreduce}{\xrightarrow{\eagersys}}
\newcommand{\eagerreducestar}{\mathrel{\xrightarrow{\eagersys}\!\!^*}}

\newcommand{\parreduce}{\twoheadrightarrow}

\newcommand{\target}{\odot}

\newcommand{\htitle}[1]{\thetitle{TLO}{#1}}
\newcommand{\ttitle}[1]{\thetitle{T}{#1}}
\newcommand{\rttitle}[1]{\thetitle{RT}{#1}}
\newcommand{\ptitle}[1]{\thetitle{P}{#1}}

\newcommand{\map}{\mathtt{map}}
\newcommand{\fold}{\mathtt{fold}}
\newcommand{\starmap}{\mathtt{mapAll}}
\newcommand{\starfold}{\mathtt{foldAll}}

\newcommand{\kwlet}{\mathtt{let}}
\newcommand{\kwin}{\mathtt{in}}

\newcommand{\df}{\stackrel{\triangle}{=}}

\newcommand{\fif}{\mathrm{if}\ }

\newcommand{\defassign}{\!\!::=\!\!}

\newcommand{\VAR}{\mathit{x}}
\newcommand{\VARY}{\mathit{y}}
\newcommand{\VARZ}{\mathit{z}}

\newcommand{\lb}{\langle}
\newcommand{\rb}{\rangle}

\newcommand{\emptyseq}{[]}
\renewcommand{\emptyset}{\{\}}

\newcommand{\fundef}[1]{\mathit{#1}}

\newcommand{\dom}{\fundef{dom}}
\newcommand{\ran}{\fundef{ran}}

\newcommand{\plaintitle}[1]{\textsc{#1}}
\newcommand{\thetitle}[2]{\textsc{#1-#2}}

\definecolor{myblue}{rgb}{100,.82,.86}
\sethlcolor{myblue}

\newcommand{\exgraphscale}{0.7}

\newcommand{\fresh}{\mathrm{fresh}}

\newcommand{\tuple}[1]{\lb #1 \rb}

\newcommand{\addn}{\mathtt{add}}

\makeatletter
\newcommand*\bigbullet{\mathpalette\bigbullet@{1.5}}
\newcommand*\bigbullet@[2]{\mathbin{\vcenter{\hbox{\scalebox{#2}{$\m@th#1\bullet$}}}}}
\makeatother

\makeatletter
\renewcommand*\bigcirc{\mathpalette\bigcirc@{1.5}}
\newcommand*\bigcirc@[2]{\mathbin{\vcenter{\hbox{\scalebox{#2}{$\m@th#1\circ$}}}}}
\makeatother

\newcommand{\init}{\fundef{init}}

\newcounter{sarrow}

\makeatletter
\newcommand{\xdashrightarrow}[2][]{\ext@arrow 0359\rightarrowfill@@{#1}{#2}}
\def\rightarrowfill@@{\arrowfill@@\relax\relbar\rightarrow}
\def\arrowfill@@#1#2#3#4{%
  $\m@th\thickmuskip0mu\medmuskip\thickmuskip\thinmuskip\thickmuskip
   \relax#4#1
   \xleaders\hbox{$#4#2$}\hfill
   #3$%
}
\makeatother

\usepackage{tikz}

\usetikzlibrary{positioning}
\usetikzlibrary{backgrounds}
\usetikzlibrary{calc}
\usetikzlibrary{arrows}
\usetikzlibrary{arrows.meta}
\usetikzlibrary{fit}
\usetikzlibrary{shapes}
\usetikzlibrary{hobby}
\usetikzlibrary{decorations.markings}
\usetikzlibrary{decorations.text}
\usetikzlibrary{decorations.pathmorphing}
\usetikzlibrary{decorations.pathreplacing}
\tikzset{snake arrow/.style=
{
decorate,
decoration={snake,amplitude=.4mm,segment length=2mm,post length=1mm}},
}
\usetikzlibrary{patterns}
\tikzset{lnode/.style={circle,fill=white,draw,minimum size=.5cm,inner sep=0pt},}
\tikzset{nnode/.style={circle,draw,minimum size=1.2cm,inner sep=0pt},}
\tikzset{op/.style={draw,fill=white,inner sep=0.05cm},}
\tikzset{pnode/.style={dashed, circle, fill=white,draw,minimum size=.6cm, inner sep=0pt},}
\tikzset{ledge/.style={->,},}
\tikzset{pedge/.style={->,},}
\tikzset{delta list/.style={rectangle,draw},}
\tikzset{srstream/.style={align=left,fill=gray!25,rectangle,draw},}
\tikzset{ostream/.style={align=left,rectangle,draw},}
\tikzset{tl/.style={align=left,rectangle,draw,fill=green!15,anchor=west,text width=0.5cm},}
\tikzset{tlp/.style={rectangle,fill=gray!50,anchor=west, inner sep=0.25cm},}
\tikzset{sarrow/.style={-{Latex[length=7mm,width=7mm]},draw=gray!35,opacity=0.7,line width=4mm},}
\tikzset{sarrownotip/.style={draw=gray!35,line width=3mm},}

\makeatletter
\pgfkeys{%
  /tikz/on layer/.code={
    \def\tikz@path@do@at@end{\endpgfonlayer\endgroup\tikz@path@do@at@end}%
    \pgfonlayer{#1}\begingroup%
  }%
}
\makeatother

\definecolor{Gray}{gray}{0.95}
\newcolumntype{G}{>{\columncolor{Gray}}c}

\algnewcommand\algorithmiccase{\texttt{case}}
\algdef{SE}[CASE]{Case}{EndCase}[1]{\algorithmiccase\ #1\ $\mathtt{of}$}{\algorithmicend\ \algorithmiccase}
\algtext*{EndCase}
\algdef{SE}[VARIANT]{Variant}{EndVariant}[1]{#1\ $\Rightarrow$\ }{\algorithmicend\ \algorithmicvariant}
\algtext*{EndVariant}
\algdef{SE}[SUBALG]{Indent}{EndIndent}{}{\algorithmicend\ }%
\algtext*{Indent}
\algtext*{EndIndent}

\usepackage{mdframed}
\newmdenv[linecolor=red,frametitle=Philip]{philipbox}

\newcommand{\toserver}{{\pmb{\Uparrow}}}
\newcommand{\fromserver}{{\pmb{\Downarrow}}}

\newcommand{\cE}{\mathbb{E}}

\newcommand{\dgcode}[1]{{\color{blue}#1}}

\newcommand{\bE}{\mathbb{B}}
\newcommand{\lE}{\mathbb{L}}
\newcommand{\tE}{\mathbb{T}}
\newcommand{\fE}{\mathbb{F}}

\newcommand{\func}{f}

\newcommand{\feq}{\equiv}

\definecolor{RC1}{rgb}{1.0,1.0,1.0}
\definecolor{RC2}{rgb}{0.9,0.9,0.9}
\definecolor{RC4}{rgb}{0.8,0.8,0.8}
\definecolor{RC3}{rgb}{0.7,0.7,0.7}

\definecolor{CBOX}{rgb}{0.9,0.9,0.9}
\definecolor{BLACK}{rgb}{0,0,0}
\newcommand{\systembox}[1]{\fcolorbox{BLACK}{CBOX}{#1}}

\makeatletter
\newcommand*{\bigcdot}{}% Check if undefined
\DeclareRobustCommand*{\bigcdot}{%
  \mathbin{\mathpalette\bigcdot@{}}%
}
\newcommand*{\bigcdot@scalefactor}{.5}
\newcommand*{\bigcdot@widthfactor}{1.15}
\newcommand*{\bigcdot@}[2]{%
  \sbox0{$#1\vcenter{}$}% math axis
  \sbox2{$#1\cdot\m@th$}%
  \hbox to \bigcdot@widthfactor\wd2{%
    \hfil
    \raise\ht0\hbox{%
      \scalebox{\bigcdot@scalefactor}{%
        \lower\ht0\hbox{$#1\bullet\m@th$}%
      }%
    }%
    \hfil
  }%
}
\makeatother

\newcommand{\cons}{\mathbin{\pmb{:}\pmb{:}}}

\newcommand{\concat}{\mathbin{\pmb{+}\mkern-3mu\pmb{+}}}

\newcommand{\integer}{n}

\newcommand{\bs}{\mathbin{\setminus}}
\newcommand{\cleanexp}{\mathtt{F}}
\newcommand{\dirtyexp}{\mathtt{T}}
\newcommand{\cdstatus}{\mathtt{\varepsilon}}
\newcommand{\dirtyfunc}{\stackrel{\dirtyexp}{\rightarrow}}
\newcommand{\cleanfunc}{\stackrel{\cleanexp}{\rightarrow}}
\newcommand{\cdfunc}{\stackrel{\cdstatus}{\rightarrow}}

\makeatletter
\def\arcr{\@arraycr{}}
\makeatother

\newcommand{\writeresult}{\mathbin{\vcenter{\hbox{\scalebox{1.5}{$\triangleleft$}}}}}
\newcommand{\writeoptwo}{\mathbin{\underline{\blacktriangleleft}}}
\newcommand{\writeoperation}{\blacktriangleleft}

\newcommand{\eagerdata}{\resizebox{0.5\width}{\height}{$\data$}}
\newcommand{\eagerinode}{\resizebox{0.5\width}{\height}{$\inode$}}

\newcommand{\eagerbE}{\hat{\bE}}
\newcommand{\eagerlE}{\hat{\lE}}
\newcommand{\eagertE}{\hat{\tE}}
\newcommand{\eagerfE}{\hat{\fE}}

\newcommand{\neone}[1]{{}^{1}#1}
\newcommand{\netwo}[1]{{}^{2}#1}
\newcommand{\nethree}[1]{{}^{3}#1}
\newcommand{\neselect}[1]{{}^{\neI}#1}
\newcommand{\neI}{\pi}

\newcommand{\setadd}{\oplus}
\newcommand{\setsubtract}{\ominus}

\newcommand{\aset}[1]{\overline{#1}}
\newcommand{\aseq}[1]{\overrightarrow{#1}}

\newcommand{\scalemath}{0.8}
\newcommand{\scalecode}{0.8}
\newcommand{\scalegraph}{0.8}

\newcommand{\dcomp}[1]{\mathbin{\overset{#1}{\circ{\circ}}}}

\begin{document}

%\title{A Foundation of Lazy Streaming Graphs}
\title{Formal Foundations of Continuous Graph Processing}
%\title{A Foundation of Online Graph Processing}

\author{Philip Dexter}
\email{pdexter1@cs.binghamton.edu}
\author{Yu David Liu}
\email{davidl@binghamton.edu}
\author{Kenneth Chiu}
\email{kchiu@binghamton.edu}
\affiliation{State University of New York (SUNY) at Binghamton}

\author{}

% short version
\newcommand{\inlong}[1]{}
\newcommand{\inshort}[1]{#1}
% long version
%\renewcommand{\inlong}[1]{#1}
%\renewcommand{\inshort}[1]{}
% colored versions
%\renewcommand{\inlong}[1]{{\color{blue}#1}}
%\renewcommand{\inshort}[1]{{\color{red}#1}}

%\newcommand{\dnote}[1]{\textbf{(((((D: #1)))))}}
%\newcommand{\pnote}[1]{\textbf{[[P: #1]]}}
%\renewcommand{\dnote}[1]{}
%\renewcommand{\pnote}[1]{}

\setcopyright{none}

\bibliographystyle{ACM-Reference-Format}
%\citestyle{acmauthoryear}

\begin{abstract}

%\dnote{need updates}

With the growing need for online and iterative graph processing, software systems that \emph{continuously} process large-scale graphs become widely deployed. 
With optimizations inherent as part of their design, these systems are complex, and have unique features beyond conventional graph processing. 
%These systems and the applications built over them are faced with distinct design challenges different from ``one operations at a time'' graph processing, and bring in unique opportunities for optimization.    
%A streaming graph system continuously processes a stream of operations over a large graph. 
%In cloud-based big data processing, % is the rule not the exception, 
%streaming graph systems and applications emerge as a core part of computing infrastructure. 
This paper describes \systemprefix{}, the first semantic foundation for continuous graph processing.
The calculus captures the essential behavior of both the backend graph processing engine and the frontend application, with a focus on two essential features: \emph{temporal locality optimization} (TLO) and \emph{incremental operation processing} (IOP).
A key design insight is that the operations continuously applied to the graph can be captured by a semantics defined over the \emph{operation stream} flowing through the graph nodes. %, a feature we call \emph{in-graph operation streams}.
%\systemprefix{} features a general operational semantics for continuous graph processing, with the 2 features at its core and the stream-based model unifying them together.  
%and this semantic view provides a general and expressive model for supporting both lazy operation processing and parallelism, and ultimately, capturing the essence of streaming graphs. %lazy operation processing in streaming graphs and compose with common forms of temporal locality optimizations. %We further extend \systemprefix{} to \parlang{} to account for parallelism, an important goal of data processing. 
%As for parallelism, we show that a distinctive opportunity called \emph{in-stream task parallelism} exists. Last but not least, 
\systemprefix{} is a systematic study on the \emph{correctness} of building continuous graph processing systems and applications.
%Through an operational semantics design and a type system design, 
The most important result is \emph{result determinism}: despite significant non-deterministic executions introduced by TLO and IOP, the results produced by \systemprefix{} are the same as conventional graph processing without TLO or IOP. The metatheory of \systemprefix{} is mechanized in Coq.

\end{abstract}

\maketitle

\section{Introduction}
\label{sec:intro}

%\dnote{david need to think: 1) the impact of non-terminition on proof; 2) whether bisimulation is still needed} 

%\dnote{key words: scalable, perform real-time graph analytics on rapidly evolving graphs, online, incremnetal graph processing, } 

Large graphs serve as the bedrock for numerous state-of-the-art data-intensive applications, from social networks, bioinformatics, artificial intelligence, to large-scale simulation. To support these applications, there is significant interest in building graph processing systems, such as graph databases~\cite{10.1145/1322432.1322433} and graph analytics and mining systems~\cite{6360146}. What emerges as the essential \emph{modus operandi} in this rapidly developing domain is \emph{continuous} graph processing: a long-running application continuously applies a large number of operations to the continuously evolving graphs. Continuous graph processing is the converging point of two directions of graph processing in active pursuit. First, as data-intensive applications are increasingly deployed on cloud computing platforms and data centers, \emph{online graph processing}~\cite{ediger2012stinger,Cheng:2012:KTP:2168836.2168846,10.1145/2567948.2580051,Vora:2017:KFA:3037697.3037748,10.1145/3267809.3267811,dhulipala2019low,10.1145/3364180,han2014chronos,shi2016tornado,196274,10.1145/2168836.2168854} is highly relevant, where the graph processing engine must process operations at a rapid rate. Second, many graph analytic and mining tasks can only be expressed through processing a large number of operations in iterations, i.e.,  \emph{iterative graph processing}~\cite{kang2009pegasus,Malewicz:2010:PSL:1807167.1807184,low2012distributed,kyrola2012graphchi,10.1145/2517349.2522739,180251,190488}.

\subsection{The Essence of Continuous Graph Processing}

Whereas early graph processing systems focused on how novel graph processing algorithms and designs~\cite{even2011graph} can help effectively process \emph{one} operation, continuous graph processing systems focus on how graph processing can benefit from optimizations that take a \emph{multitude} of operations into account. In this spirit, a broad family of optimizations at the core of continuous graph processing systems can be summarized as \emph{temporal locality optimization (TLO)}: temporally consecutive operations applied to the graph may be manipulated for more effective graph processing \emph{before or during their processing}. In the simple case of two consecutive operations $\processor_1$ and $\processor_2$, where $\processor_1$ is to be applied to the graph before $\processor_2$, four forms of TLO are well known in continuous graph processing: 

\begin{itemize}
    \item \emph{Batching}: processing $\processor_1$ and $\processor_2$ ``in tandem,'' so that only one graph traversal is needed for processing both, as opposed to two traversals if $\processor_1$ and $\processor_2$ are processed one by one. 

    \item \emph{Reordering}: applying $\processor_2$ first and $\processor_1$ later to the graph, on the assumption that this reversed order of application can produce the same result as processing in the original order. Reordering is useful in use scenarios such as $\processor_2$ has a higher priority or a closer deadline.
    
    \item \emph{Fusing}: composing $\processor_1$ and $\processor_2$ into one operation $\processor$, on the assumption that $\processor$ can produce the same result as processing both $\processor_1$ and $\processor_2$. Just like batching, fusing is useful in reducing the number of graph traversals.  
    
    \item \emph{Reusing}:  applying $\processor_1$ and $\processor_2'$ to the graph where $\processor_2'$ derives from $\processor_2$ but reuses the result of $\processor_1$ processing, so that the computation performed for $\processor_1$ processing is not redundantly done. This style of TLO is known as Multi-Query Optimization (MQO)~\cite{10.1145/42201.42203,10.1145/318898.318993,105474} in database systems.
\end{itemize}

A second distinct form of optimization in continuous graph processing is \emph{incremental operation processing} (IOP). Operations often arrive at a high rate in continuous graph processing systems. The need for scalable --- and sometimes real-time --- data processing dictates that ``one operation at a time'' processing often cannot meet the demand. 
%Indeed, if a continuous graph processing system were to process ``one operation at a time'', such a system from a design point would degenerate to a non-continuous graph processing system after all. In practice, a
A common theme of online continuous graph processing systems %, especially those designed for online processing, 
is to defer some operations for incremental processing later.
%, and incrementally process them as needed or when system resources become more available. 

Indeed, the two forms of optimizations essential for continuous graph processing go hand in hand: it is often the delay resulted from IOP that multiple operations can co-exist to participate in the optimizations enabled by TLO.

The large body of experimental research demonstrates why continuous graph processing matters: graph databases, graph analytics, and graph mining systems are becoming an indispensable part of state-of-the-art computing infrastructure. At the same time, %s shown in existing research, 
the online and iterative nature of these systems poses distinct challenges 
%different from ``one operation at a time'' graph processing. These challenges 
that call for complex solutions in graph processing engine design spanning TLO and IOP, as well as the programming model for the data-intensive application at the frontend. 
%parallelism, and the subtle interaction of these features. 
%in addition, the important question on how the data-intensive application interacts with the graph processing engine must also be addressed. 
In a nutshell, continuous graph processing systems and applications form a landscape with \emph{wide deployment, unique challenges, and complex and diverse solutions}. Surprisingly, in contrast with the rapid advances in experimental research, no prior formal foundations exist for this important family of software systems.

\begin{figure}[t]
\centering

\begin{tabular}{@{\hspace{-0.8cm}}c@{\hspace{0.5cm}}c}
\scalebox{\scalegraph}{
\begin{tikzpicture}

% server/client
\node[] at (0, 1) (server) {\hspace{-0.9cm}Backend};
\node[] at (4.5, 1) (client) {\hspace{1.3cm}Frontend};

% graph
\node[nnode] at ($(server) + (0.5,-1.55)$) (eve) {$\mathtt{eve}$};
\node[nnode] at ($(eve) + (-2.0,0)$) (deb) {$\mathtt{deb}$};
\node[nnode] at ($(deb) + (1.0,-1.0)$) (cam) {$\mathtt{cam}$};
\node[nnode] at ($(cam) + (-1.0,-1.0)$) (bob) {$\mathtt{bob}$};
\node[nnode] at ($(bob) + (2.0,0)$) (amy) {$\mathtt{amy}$};

\draw[ledge, out=135, in=135, looseness=2] (eve) to (bob);
\draw[ledge, bend right=20] (eve) to (amy);
\draw[ledge, bend left] (deb) to (cam);
\draw[ledge, bend left] (cam) to (bob);
\draw[ledge, bend right=20] (amy) to (eve);
    
% client box
\begin{pgfonlayer}{background}
\node[draw, fill=black!75, inner sep=4pt, below left=0.3cm and -0.6cm of client, inner ysep=2.0cm, inner xsep=1.5cm, anchor=north west] (clientbox) {};
\end{pgfonlayer}
\node[text=white, anchor=north west] (t1) at (clientbox.north west) {\tiny $\kwlet $};
\node[text=white] (t2) at (t1.south) {\tiny $\kwlet $};
\draw[-,draw=white,decorate,decoration={snake,amplitude=0.4mm}] (t1) -- (t1 -| clientbox.east);
\draw[-,draw=white,decorate,decoration={snake,amplitude=0.4mm}] (t2) -- (t2 -| clientbox.east);
\coordinate (y) at ([shift={(-0.15,-0.25)}]t2);
\draw[-,draw=white,decorate,decoration={snake,amplitude=0.4mm}] (y) -- (y -| clientbox.east);
\coordinate (y) at ([shift={(-0.15,-0.5)}]t2);
\draw[-,draw=white,decorate,decoration={snake,amplitude=0.4mm}] (y) -- (y -| clientbox.east);
\coordinate (y) at ([shift={(-0.15,-0.75)}]t2);
\draw[-,draw=white,decorate,decoration={snake,amplitude=0.4mm}] (y) -- (y -| clientbox.east);
\coordinate (y) at ([shift={(-0.15,-1.0)}]t2);
\draw[-,draw=white,decorate,decoration={snake,amplitude=0.4mm}] (y) -- (y -| clientbox.east);
\coordinate (y) at ([shift={(-0.15,-1.25)}]t2);
\draw[-,draw=white,decorate,decoration={snake,amplitude=0.4mm}] (y) -- (y -| clientbox.east);
\coordinate (y) at ([shift={(-0.15,-1.5)}]t2);
\draw[-,draw=white,decorate,decoration={snake,amplitude=0.4mm}] (y) -- (y -| clientbox.east);
\coordinate (y) at ([shift={(-0.15,-1.75)}]t2);
\draw[-,draw=white,decorate,decoration={snake,amplitude=0.4mm}] (y) -- (y -| clientbox.east);
\coordinate (y) at ([shift={(-0.15,-2.0)}]t2);
\draw[-,draw=white,decorate,decoration={snake,amplitude=0.4mm}] (y) -- (y -| clientbox.east);
\coordinate (y) at ([shift={(-0.15,-2.25)}]t2);
\draw[-,draw=white,decorate,decoration={snake,amplitude=0.4mm}] (y) -- (y -| clientbox.east);
\coordinate (y) at ([shift={(-0.15,-2.50)}]t2);
\draw[-,draw=white,decorate,decoration={snake,amplitude=0.4mm}] (y) -- (y -| clientbox.east);
\coordinate (y) at ([shift={(-0.15,-2.75)}]t2);
\draw[-,draw=white,decorate,decoration={snake,amplitude=0.4mm}] (y) -- (y -| clientbox.east);
\coordinate (y) at ([shift={(-0.15,-3.0)}]t2);
\draw[-,draw=white,decorate,decoration={snake,amplitude=0.4mm}] (y) -- (y -| clientbox.east);
\coordinate (y) at ([shift={(-0.15,-3.25)}]t2);
\draw[-,draw=white,decorate,decoration={snake,amplitude=0.4mm}] (y) -- (y -| clientbox.east);
\coordinate (y) at ([shift={(-0.15,-3.5)}]t2);
\draw[-,draw=white,decorate,decoration={snake,amplitude=0.4mm}] (y) -- (y -| clientbox.east);
\coordinate (y) at ([shift={(-0.15,-3.75)}]t2);
\draw[-,draw=white,decorate,decoration={snake,amplitude=0.4mm}] (y) -- (y -| clientbox.east);
\node[fill=black!75,anchor=east, inner xsep=0.06cm, inner ysep=2.0cm] at (clientbox.east) {};

% graph box
\node[draw,fit=(eve)(deb)(bob)(amy),inner sep=0.6cm] (graphbox) {};

% stream arrows
\begin{pgfonlayer}{background}
\coordinate (evex) at ($(eve) + (0,0.3)$);
\coordinate (amyx) at ($(amy) + (0,0.3)$);
\draw[sarrow, text opacity=1, rounded corners] (clientbox.west |- evex) to node {\hspace{-2.0cm}$\processor_1, \ldots, \processor_i$} (graphbox.east |- evex) to ($(eve) + (0,0.3)$) to ($(deb) + (0,0.3)$) to ($(cam) + (0,0.3)$) to ($(bob) + (0,0.3)$) to ($(amy) + (0,0.3)$) to (graphbox.east |- amyx) to node {\hspace{-2.0cm}$\val_1, \ldots, \val_j$} (clientbox.west |- amyx);
\end{pgfonlayer}

\end{tikzpicture}
}
&
\raisebox{1cm}{
\begin{tikzpicture}[scale=\exgraphscale]
  \begin{scope}[every node/.style={scale=\exgraphscale,font=\fontsize{12pt}{\baselineskip}\selectfont}]
    \node[draw,dashed,fill=olive!25,align=left] at (1.9,0.5) (legend) {Legends: \\ \hspace{2em} operation stream \\ \hspace{0.5em} $\processor$ \hspace{0.5em} operation \\ \hspace{0.5em} $\val$ \hspace{0.5em} result \\ \hspace{2em} graph processing engine \\ \hspace{2em} data-intensive application \\[1ex] \hspace{2em} graph node and edge};
    \draw[{Latex[length=1.2mm,width=3.2mm]}-,draw=black,line width=1.2mm] ($(legend.west) + (0.13,0.80)$) -- ($(legend.west) + (1.02,0.80)$);
    \draw[{Latex[length=1mm,width=3mm]}-,draw=gray!35,line width=1mm] ($(legend.west) + (0.15,0.80)$) -- ($(legend.west) + (1.00,0.80)$);
    %%
%    \draw[-{Latex[length=1.2mm,width=3.2mm]},draw=black,line width=1.2mm] ($(legend.west) + (0.13,0.30)$) -- ($(legend.west) + (1.02,0.30)$);
%    \draw[-{Latex[length=1mm,width=3mm]},draw=gray!35,line width=1mm] ($(legend.west) + (0.15,0.30)$) -- ($(legend.west) + (1.00,0.30)$);
    %%
    \node[draw,fill=white] at ($(legend.west) + (0.55,-0.25)$) {$\hspace{1em}$};
    \node[draw,fill=black!75] at ($(legend.west) + (0.55,-0.65)$) {$\hspace{1em}$};
    \node[lnode] (x) at ($(legend.west) + (0.4,-1.2)$) {};
    \coordinate (z) at ($(x) + (0.5,0)$);
    \draw[ledge] (x) -- (z);
  \end{scope}
\end{tikzpicture}
}
\end{tabular}

\caption{The Frontend and Backend of Continuous Graph Processing}
\label{fig:clientserver:simple}

\end{figure}
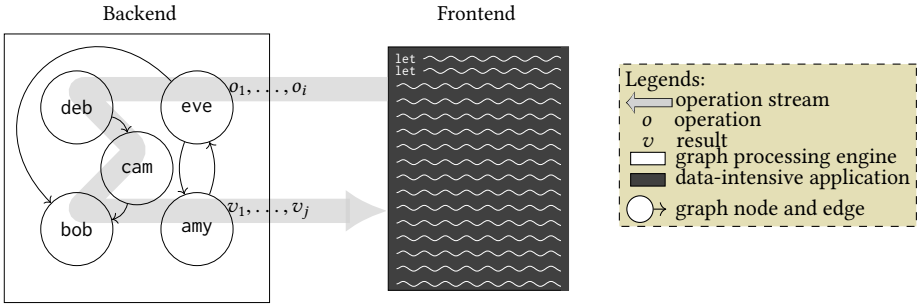

\subsection{\systemprefix{}}
\label{subsec:dcalculus}

We introduce \systemprefix{}, the first formal foundation for continuous graph processing. Our calculus has two high-level design goals:  \emph{illuminating the design space} of continuous graph processing systems and applications, and \emph{improving the assurance} of their construction. 
%On the highest level, our foundational approach aims at complementing existing experimental research with a rigorous understanding in continuous graph processing, on its expressiveness and correctness. 
The behavior captured by \systemprefix{} spans the graph processing application (the \emph{frontend}) and the graph processing system (the \emph{backend}), as illustrated in Fig.~\ref{fig:clientserver:simple}. The frontend program continuously produces data processing operations such as $\processor_1, \ldots, \processor_i$ in the Figure, and delivers them to the backend that maintains a potentially large and evolving graph. As operations are processed and results become available, the backend delivers the latter back to the frontend, such as $\val_1, \ldots, \val_j$ in the Figure. 
%The backend-frontend division is aligned with our intuitive notions of \emph{continuous graph processing systems} and \emph{continuous graph processing applications} respectively. 
\systemprefix{} unifies both the frontend and the backend in one calculus.

%to study the behavior of backend stream processing, the frontend programming interface, and the interaction between the two. 

\paragraph{A Semantic Model with Operation Streams}

A key insight of \systemprefix{} is that the spirit of \emph{continuous} graph processing can be embodied by viewing the operations as a \emph{stream}, which we call the \emph{operation stream}; more importantly, the operation stream does not only exist at the frontend-backend boundary, but also ``flows through'' the graph nodes. For example, Fig.~\ref{fig:clientserver:simple} shows an operation stream forms from the frontend to backend (which we call the \emph{top-level operation stream} for convenience), and then continues to flow into nodes \texttt{eve}, \texttt{deb}, \texttt{cam}, \texttt{bob}, \texttt{amy}, in that order (which we call the \emph{in-graph operation stream}). In other words, we conceptually align continuous graph processing with \emph{stream processing}. This view is not only aligned with our intuition of what ``continuous'' means, but also does it provide a unified foundation to model the essential features of continuous graph processing: TLO is modeled as stream rewriting, and IOP is modeled as operation propagation in the stream.
%and parallelism is naturally supported thanks to stream processing. 
We will elaborate these features shortly. To place this novel view in context, % understand the novel view taken by operation streams, 
observe that there is a fundamental difference between \systemprefix{} and \emph{data streaming} (e.g.,~\cite{lucid,lustre,murray2011ciel,Zaharia:2016:ASU:3013530.2934664,10.1145/2517349.2522737,murray2013naiad,thies2002streamit,10.1145/1297027.1297043,Meyerovich:2009:FPL:1640089.1640091,10.1007/978-3-662-44202-9_15}.
%and their precursor, data flow systems (e.g.,~\cite{lustre,lucid}). %, despite the common ``streaming'' nature. 
In \systemprefix{},
a stream is formed by operations, to be passed through structured data.
In data streaming systems, 
a stream is formed by data, to be passed through structured operations.

%This view is aligned with our intuitive notion of stream processing, and it serves as a natural extension to the ``top-level'' operation stream delivered from the frontend to the backend.

%Second, despite ``streaming'' being a key feature,
%\systemprefix{} should not be confused with 
%We will detail related work in \S~\ref{sec:related}.

%Modern streaming graph applications are \emph{continuous} in nature, \emph{i.e.}, the data processing program in the Figure is long-running. Throughout its lifetime, the program continuously \emph{sends} operations into the operation stream, and continuously \emph{receives} results from the result store. 

\paragraph{A Unified Design Space Spanning Backend and Frontend}
\label{sec:unified}

The centerpiece of \systemprefix{} is an operational semantics to account for the backend behavior, i.e., the graph processing engine. It unifies the design space of TLO and that of IOP into one core calculus. Furthermore, \systemprefix{} consists of a simple programming model, and a type system to reason about the frontend-backend interaction. %, especially its interaction with the backend.

\systemprefix{} captures the essence of TLOs as a reduction relation for rewriting the operation stream. This relation can define all 4 forms of TLOs introduced earlier. In addition to defining \emph{how} TLOs are supported, \systemprefix{} highlights \emph{where} and \emph{when} TLOs may happen. Thanks to in-graph operation streams, TLOs may happen at an arbitrary graph node that the in-graph operation stream flows through, leading to in-graph batching, in-graph reordering, in-graph fusing, and in-graph reusing.  

%As operations in the in-graph operation stream flow through the nodes in a non-deterministic manner as we described earlier, operations that happen to reach the same graph node may subject to TLOs in an opportunistic manner. %As we shall see, this flexible non-deterministic design is beneficial for parallelism support.

In \systemprefix{}, in-graph operation streams provide a natural and general way to support IOP. %If one takes a simple serial execution model where the processing of one operation implies the delay of that of another, in-graph operation streams 
%They \emph{de facto} define a general model that 
In the in-graph operation stream, the operation can incrementally move through the graph nodes, and be deferred at any arbitrary node and resumed later. 
%Indeed, this is the essence of \emph{incremental} data processing. Furthermore, t
A more basic design widely adopted in graph processing systems --- the operation may be deferred at a ``top-level buffer'' at the boundary between the frontend and the backend~\cite{10.1145/3267809.3267811,Cheng:2012:KTP:2168836.2168846,Cheng:2012:KTP:2168836.2168846,Vora:2017:KFA:3037697.3037748} --- is a special case of IOP design in \systemprefix{}. 

% Last but not least, operation streams expose significant opportunities for parallelism. The insight here is that in-graph operation streams naturally allow for \emph{benign non-determinism}, i.e., non-deterministic executions with deterministic results. %For example, while two operations in the in-graph operation stream are both going through the traversal but have reached different nodes, it is a non-deterministic choice over which operation propagates next and which defers further. Earlier, we have also mentioned that TLOs may happen in an opportunistic manner, another source of non-determinism in executions.  
% Benign non-determinism is known to indicate the potential for safe parallelism~\cite{10.1145/2509136.2509533}. Through a principled design, \systemprefix{} systematically turns each form of non-determinism latent in the in-graph operation stream into parallelism. % leading to a parallel variant of \systemprefix{} we call \parlang{}. 
% %As we shall see in \S~\ref{sec:parallel}, in-graph operation streams enable several forms of parallelism in
% With parallelism support, operations may be propagated and realized in parallel as long as they are not performed over the same graph nodes; in a similar vein, multiple TLOs may also happen in parallel.

\inlong{%%%%%%%%%%%%%%%%%%%%%%%%% begin inlong %%%%%%%%%%%%%%%%%%%%%%%%%
On the front end, \systemprefix{} is endowed with a simple programming model, an extension to $\lambda$ calculus with data processing operations. 
The interaction between the frontend program and the backend graph processing engine is modeled through a pair of expressions: \emph{emitting} the operation from the frontend to the backend, and \emph{claiming} the operation result from the backend to the frontend. 
The frontend and the backend interact asynchronously~\cite{multilisp}, decoupling the frontend from the backend. 
}%%%%%%%%%%%%%%%%%%%%%%%%% end inlong %%%%%%%%%%%%%%%%%%%%%%%%%

\paragraph{Correctness}

With expressive forms of TLO and IOP  unified in one calculus, 
establishing the \emph{correctness} of continuous graph processing
%systems and applications 
is a non-trivial problem. With TLO, the operations in the stream are altered. With IOP, significant non-deterministic executions are introduced. On the high level, both TLO and IOP can be viewed as optimizations over data processing, and the very premise of \emph{sound} optimization is that graph processing should not produce an unexpected result. % a correctness concern. 

The main property enjoyed by
\systemprefix{} 
%and its \parlang{} variant 
is \emph{result determinism}: despite significant non-deterministic executions introduced by TLO and IOP (see \S~\ref{sec:highlevelsemantics}), all terminating executions of the same program produce the same result. 
This important result subsumes the \emph{observable equivalence} of \systemprefix{} with the baseline \emph{eager data
processing} --- the processing style where operations in the operation stream are processed one by one, without TLO or IOP. 
%We establish observable 
%equivalence through bisimulation, detailed in \S~\ref{sec:properties}. \dnote{update this when we know what confluence really means, and positioned in the context of ob equiv. } 

\inlong{%%%%%%%%%%%%%%%%%%%%%%%%% begin inlong %%%%%%%%%%%%%%%%%%%%%%%%%
The result subsumes important properties of
continuous graph processing: 
\begin{itemize}
    \item TLO is \emph{sound}, i.e., it does not alter the final result
of computation despite its manipulation in the operation stream before the processing of operations are completed. 
\item The non-determinism introduced through in-graph operation streams is indeed \emph{benign}, i.e., 
a model with non-deterministic executions but deterministic results.
%\item the parallelism is \emph{safe}, i.e., the parallelism support still preserves the determinism in results.
\end{itemize} 
}%%%%%%%%%%%%%%%%%%%%%%%%% end inlong %%%%%%%%%%%%%%%%%%%%%%%%%

\systemprefix{} is also endowed with an effect type system with the primary goal of enforcing \emph{phase distinction}: while the computation at the frontend can freely issue new operations for backend processing, the backend computation should not issue new operations for processing. If phase distinction were ignored, the non-determinism latent in in-graph operation streams would lead to non-determinism in results. Intuitively, this is analogous to a high-level data race that our type system eliminates.

\subsection{Contributions}

To the best of our knowledge, \systemprefix{} is the first rigorous study on continuous graph processing. Its scope spans the backend graph processing engine design, the frontend programming model, and the interaction between the two. 
%The frontend and  and backend corresponds to what people call streaming graph systems and streaming graph applications.  
%Our foudational approach complements active experimental research by unifying essential features in the design space of continuous graph processing under one calculus, and improving the assurance of continuous graph processing systems and applications. Together, 
%\systemprefix{} serves as a first-step for verifying data-intensive systems and applications.
Technically, this paper makes the following contributions:

\begin{itemize}

\item a novel operational semantics based on operation streams to capture the essence of the backend of continuous graph processing, where TLO and IOP are unified in one system;

%\item It defines a calculus to elucidate the interaction between TLO and IOP at the backend. It supports common forms of TLO. It features in-graph operation streams to generalize lazy operation processing, empowering the serial semantics with significant opportunities for non-deterministic executions. 

\item a frontend asynchronous programming model with a type system to enforce phase distinction;

%\item It describes a parallel calculus to turn non-determinism in the serial calculus to parallelism. The calculus supports rich and fine-grained forms of parallelism in operation and data processing.

\item a metatheory establishing result determinism in the presence of non-deterministic executions, %entailing observable equivalence with eager data processing, 
as well as type soundness; %The former subsumes important correctness properties in continuous graph processing design, such as sound temporal locality optimization and benign non-determinism.

\item mechanized proofs in Coq for the metatheory, available online~\footnote{anonymous at \url{https://github.com/anonymous10012/lang.v}}.

%lays a foundation for streaming graphs via \eagerlang{}, a calculus to
%intuitively capture the essence of (eager) streaming graph processing.

% \item It elucidates lazy processing of streaming graphs via \ourlang{}, a
% calculus with expressive support for lazy propagation (decentralization,
% locality, and non-determinism) and in-data optimization (batching, fusion, and
% splicing).

% \item It presents a bisimulation result to establish the soundness of lazy data
% processing through observable equivalence between the two modes of data
% processing.

\end{itemize}

\section{Motivating Examples}
\label{sec:background}

In this section, we informally highlight the essential features of \systemprefix{} through examples. The two examples --- one on graph databases and the other on graph analytics --- also serve as concrete usage contexts that motivate the rigorous development of \systemprefix{}. In other words, the expressiveness and correctness of graph databases and graph analytic systems matter. 

%into the application context whose expressiveness and correctness motivate our formal foundation.

%development of \systemprefix{}, through two examples. The section also serves as an informal account for the essential features 
%\systemprefix{} formally captures, especially the notions of eager vs. lazy graph processing. 

\subsection{Motivating Usage Context 1: Continuous Graph Processing Databases}
\label{sec:coresocial}

\begin{figure}[t]
\centering

\begin{minipage}[t]{0.38\textwidth}
\scalebox{\scalecode}{
\begin{minipage}[t]{\linewidth}
\begin{algorithmic}[1]\small
% Client
% \Statex $\texttt{(*\ \ \ \ Client\ \ \ \ *)}$
\State \texttt{// node payload values}
\State $\kwlet\ \mathtt{amy, bob, cam, deb, eve, fred} =$
\Statex $\hspace{0.55cm}\mathtt{n_{amy}, n_{bob}, n_{cam}, n_{deb}, n_{eve}, n_{fred}}\ \kwin$
\State \texttt{// graph construction}
\State $\kwlet\ \mathtt{a} = \mathtt{\color{blue}\mathtt{add\ amy}}\ \kwin$
\State $\kwlet\ \mathtt{b} = \mathtt{\color{blue}\mathtt{add\ bob}}\ \kwin$
\State $\kwlet\ \mathtt{c} = \mathtt{\color{blue}\mathtt{add\ cam}}\ \kwin$
\State $\kwlet\ \mathtt{d} = \mathtt{\color{blue}\mathtt{add\ deb}}\ \kwin$
\State $\kwlet\ \mathtt{e} = \mathtt{\color{blue}\mathtt{add\ eve}}\ \kwin$
\State ${\color{blue}\mathtt{addRelationship}\ \mathtt{c}\ \mathtt{b}};$
\State ${\color{blue}\mathtt{addRelationship}\ \mathtt{d}\ \mathtt{c}};$
\State $\mathtt{\color{blue}\mathtt{addRelationship}\ \mathtt{e}\ \mathtt{b}};$
\State $\mathtt{\color{blue}\mathtt{addRelationship}\ \mathtt{e}\ \mathtt{a}};$
\State $\mathtt{\color{blue}\mathtt{addRelationship}\ \mathtt{a}\ \mathtt{e}};$
\State \texttt{// dynamic queries and updates}
\State $\kwlet\ \mathtt{nb} = \mathtt{\color{blue}\mathtt{queryNode}\ \mathtt{b}}\ \kwin$
\State $\mathtt{\color{blue}\mathtt{updatePayload}\ \mathtt{a}\ \mathtt{nb}};$
\State $\kwlet\ \mathtt{nb2} = \mathtt{\color{blue}\mathtt{queryNode}\ \mathtt{b}}\ \kwin$
\State $\kwlet\ \mathtt{f} = \mathtt{\color{blue}\mathtt{add}\ \mathtt{fred}}\ \kwin$
\State $\mathtt{\color{blue}\mathtt{addRelationship}\ \mathtt{b}\ \mathtt{f}};$
\State $\mathtt{\color{blue}\mathtt{deleteRelationship}\ \mathtt{b}\ \mathtt{f}};$
\State $\ldots$
\end{algorithmic}
\end{minipage}
}
\caption{The \textsc{CoreSocial} Application in Pseudocode}
\label{fig:serverclient}
\end{minipage}\hfill\vline\hfill
\hspace{-1em}\begin{minipage}[t]{0.58\textwidth}
\scalebox{\scalecode}{
\begin{minipage}[t]{\textwidth}
\begin{algorithmic}[1]\small
\State $\kwlet\ \mathtt{numSuperSteps = 30}\ \kwin$
% \State $\kwlet\ \mathtt{fK = \lambda \tuple{key;\ \_;\ \_} . \lambda kset . \{ key \} \cup kset}\ \kwin$
% \State $\kwlet\ \mathtt{\tuple{\_;\ \_;\ keys}} = \ \dgcode{\mathtt{foldSet}\ \mathtt{fK}\ \emptyset\ \mathtt{\KSstar}}\ \kwin$
\State \texttt{// keys of interest}
\State $\kwlet\ \mathtt{keys} = \ldots \ \kwin$
\State $\kwlet\ \mathtt{numNodes} = \mathtt{length\ keys}\ \kwin$
\State $\kwlet\ \mathtt{fPInit = \lambda \tuple{\_; \_; \_} . \frac{1}{\mathtt{numNodes}}}\ \kwin$
% \State $\kwlet\ \mathtt{fA = \lambda nk . \lambda np . \lambda na . na}\ \kwin$
\State $\dgcode{\mathtt{mapVal\ fPInit\ keys}};$
\State $\mathtt{foreach\ 1..numSuperSteps}$ \label{l:loop}
\Indent
\State $\kwlet\ \mathtt{neighborPSums} =$
\Indent
\State $[\tuple{\mathtt{nk}; \dgcode{\mathtt{foldVal\ fPSum\ 0\ keys}}}$ \label{l:listc}
\State $\hspace{0.5em}\mid \mathtt{nk \ in\ keys},$
\State $\hspace{1.1em}\kwlet\ \mathtt{fPSum = \lambda \tuple{\_;\ payload;\ adjlist} . \lambda sum . if\ nk \ in\ adjlist}$
\Statex $\hspace{7.0cm}\mathtt{then\ payload + sum}$
\Statex $\hspace{7.0cm}\mathtt{else\ sum}]\ \kwin$
\EndIndent
\State $\mathtt{foreach\ \tuple{nk; \tuple{\_;\ neighborsSum;\ \_}}\ in\ neighborPSums}$
\Indent
\State $\kwlet\ \mathtt{fPG = \lambda \tuple{\_;\ \_;\ adjlist} . \frac{0.15}{numNodes} + 0.85 * \frac{neighborsSum}{length\ adjlist}}\ \kwin$
\State $\dgcode{\mathtt{mapVal\ fPG\ [ nk ]}}$
\EndIndent
\EndIndent
\end{algorithmic}
\end{minipage}
}\vskip0.2in
\caption{The \textsc{CorePR} Application in Pseudocode (Expressions encodable by $\lambda$ calculus are liberally used, such as loop at Line~\ref{l:loop} and list comprehension at Line~\ref{l:listc}. A summary of encoded expressions can be found in \S~\ref{sec:syntax}. Notation $\_$ represents a name that does not occur in the function body.  )
%\dnote{at line 11, should sum be the first argument, and the tuple as the second?}
%\pnote{currently it's correct based on our \plaintitle{Fold} rule}
}
\label{fig:pagerank} 
\end{minipage}

\end{figure}

Graph databases~\cite{papakonstantinou1995object,buneman1996query,Venkataramani:2012:TFS:2213836.2213957}
are an important family of databases that rely on structured graphs for data storage. 
\inlong{%%%%%%%%%%%%%%%%%%%%%%%%% begin inlong %%%%%%%%%%%%%%%%%%%%%%%%%
They are known for their intuitive and expressive representation for complex relationships among data. In real-world deployment settings, a graph database often serves as the backend for continuous processing of queries and updates, defined through application programming interface (API) functions defined by the database. 
}%%%%%%%%%%%%%%%%%%%%%%%%% end inlong %%%%%%%%%%%%%%%%%%%%%%%%%

%To connect with the rest of the paper, the
%running example in Fig.~\ref{fig:serverclient} is written in the
%\systemprefix{} grammar extended with standard programming features.

\begin{example}[\textsc{CoreSocial} in \systemprefix{} Sugared Syntax]
\label{ex:serverclient}
Fig.~\ref{fig:serverclient} shows a minimalistic program for maintaining a social network in the form of 
a graph database. In this sugared syntax, 
\inlong{%%%%%%%%%%%%%%%%%%%%%%%%% begin inlong %%%%%%%%%%%%%%%%%%%%%%%%%
akin to ML/Haskell functional programming, the sequenced expressions (through the $\kwlet\dots \kwin$ expression and the ; expression) mimic the continuous submission of database queries and updates.  
%It issues a number
%of graph processing operations to the database.
%
}%%%%%%%%%%%%%%%%%%%%%%%%% end inlong %%%%%%%%%%%%%%%%%%%%%%%%%
Lines 1-8 are node additions and
Lines 9-13 are relationship additions.
The remaining lines further consist of a mixture of queries (Lines 15 and 17) and updates (Lines 16, 18, 19, 20). Each graph processing operation --- highlighted in {\color{blue}blue} --- is analogous to an API function in the graph database. 
The (logical) graph after the program reaches Line 13
is shown as the backend of Fig.~\ref{fig:clientserver:simple}.
\end{example}

%The program in Fig.~\ref{fig:serverclient} is written in \systemprefix{}, with syntactical sugars that will be made clear in \S~\ref{sec:baseline}. 

%As shown in this example,
The programmer syntax assumed by our calculus is conventional: 
it consists of standard features encodable by $\lambda$ calculus, together with graph processing expressions.
%The latter has been shown in {\color{blue}blue} in the example.
%
%The functionality of each graph processing expression should be self-explanatory through their names.
%
%We are generous with the latter in this example, %number of graph processing expressions here,
%because as we shall see in \S~\ref{sec:baseline},
%they will be encoded into a much smaller set.
The runtime representation of the graph is also conventional. A directed graph consists of a sequence of \emph{nodes}, each consisting of a unique identifier (\emph{key}), a value the node carries (\emph{payload}), and a list of keys for its adjacent nodes connected through out-edges (\emph{adjacency list}). In the \textsc{CoreSocial} example, the $\mathtt{add}$ expression at Line 4 adds a node whose key is freshly generated, whose payload is $\integer_\mathtt{amy}$, and whose adjacent list is initialized as an empty sequence. The generated key is returned and bound to name $\mathtt{a}$.
At Line 9, the $\mathtt{addRelationship}$ expression updates the adjacent list of node whose payload is $\integer_\mathtt{cam}$ with the single sequence that contains the key of the node whose payload is $\integer_\mathtt{bob}$.
The functionality of other graph processing expressions should be self-explanatory through their names.   

%The behavior of processing an operation at the backend can be intuitively viewed as a \emph{graph traversal}, where the operation reaches the nodes one by one following the traversal order (we say the operation is \emph{propagated}) and perform computations when reaching the node(s) the operation is intended for (we say the operation is \emph{realized}). 

The simple \textsc{CoreSocial} example demonstrates that \systemprefix{}
support \emph{dynamic graphs}, i.e., it allows for continuous graph evolution in terms of both node and topology changes. 
\inlong{%%%%%%%%%%%%%%%%%%%%%%%%% begin inlong %%%%%%%%%%%%%%%%%%%%%%%%%
This is a basic requirement for graph databases, but note that in the general context of continuous graph processing, our model is capable of modeling the more expressive subset of solutions in support of graph evolution. Indeed, in \systemprefix{}, operations that may lead to graph changes (\emph{updates}) and those that do not (\emph{queries}) can intermingle in an arbitrary order without restriction. %of evaluation. 
}%%%%%%%%%%%%%%%%%%%%%%%%% end inlong %%%%%%%%%%%%%%%%%%%%%%%%%

%\begin{example}[Dynamism]
%The \textsc{CoreSocial} graph is continuously updated, both in terms of nodes (Lines 4-8, 16, and 18) and topology (Lines 9-12 and 18-19). Update operations can also be intermingled with query operations (Lines 15 and 17).
%\end{example}

%In the usage domain of databases, \systemprefix{} is primarily concerned with \emph{streaming}  databases: 
%As the backend of a long-running application, the database engine may continuously receive queries and updates, process them, and return 
%the results. 

Beneath the conventional programmer syntax, \emph{asynchronous} semantics is designed for graph operations: the evaluation of a graph operation at the frontend does \emph{not} need to 
block until the backend database returns the result. Instead, the evaluation places the operation of concern into the operation stream destined for the backend, which we say the operation is \emph{emitted} from now on.

%As a whole, the multitude of asynchronous operations from
%the frontend naturally form a stream. 

\begin{example}[Graph Database Processing as Operation Streams]
\label{ex:streaming}
%The evaluation of $\addn\ \mathtt{amy}$ at Line 4
%does not require the database to complete the node addition.
%Instead, the operation is submitted to the database
%and the application program can continue with Line 5.
The execution through Line 4-8 emit 5 operations forming an operation stream as %follows:
%\[
$[\addn\ \mathtt{n}_\mathtt{amy}, \addn\ \mathtt{n}_\mathtt{bob}, \addn\ \mathtt{n}_\mathtt{cam}, \addn\ \mathtt{n}_\mathtt{deb}, \addn\ \mathtt{n}_\mathtt{eve}]$. 
%\]
\end{example}

Modeling the frontend-backend interaction through asynchronous semantics is not new. Indeed, asynchronous data processing~\cite{elteir2010enhancing,wang2013asynchronous,bertsekas1989parallel} is a classic data processing system design, often justified by the performance impedance mismatch between the application and the underlying system. How to design a programming language with asynchronous
semantics is also well understood. \systemprefix{} follows the same route of futures~\cite{multilisp}. For example, the emission of $\mathtt{add\ amy}$ at Line 4 generates a future value, which is subsequently \emph{claimed} at Line 12 \emph{a la} future semantics. In this backdrop, \systemprefix{} is minimalistic in its support on the frontend, but with a fresh motivation on why asynchronous data processing matters: it enables operation streams and brings in their benefits, e.g., on TLO.

%At the frontend, the stream nature of operations in \systemprefix{} is enabled by the

%Each asynchronously processed operation
%eventually produces a result --- for example,
%$\addn\ \mathtt{amy}$ is evaluated to the freshly generated key of the added node after it is processed by the backend. This result value can later be \emph{claimed} at the frontend. 
%For example, $\mathtt{a}$ appears at Line 12, 13, and 16.

The \textsc{CoreSocial} example further shows that \systemprefix{} supports \emph{dependent operations}: an operation may have an argument referring to the result of an earlier emitted operation. 

\begin{example}[Dependent Operations]
\label{ex:depend}
Line 15 queries the node $\mathtt{b}$ through the
$\mathtt{queryNode}$ expression. The resulting
value is used to update the payload
of the node $\mathtt{a}$ at Line 16 through 
the $\mathtt{updatePayload}$ expression. 
\end{example}

\inlong{%%%%%%%%%%%%%%%%%%%%%%%%% begin inlong %%%%%%%%%%%%%%%%%%%%%%%%%
If operation processing were synchronous, dependent operations would be a non-feature: the frontend could simply eliminate the dependency through a $\beta$ reduction. 
}%%%%%%%%%%%%%%%%%%%%%%%%% end inlong %%%%%%%%%%%%%%%%%%%%%%%%%

With asynchronous operation processing, dependent operations entail that the dependency is carried into the operation stream and subsequently resolved at the \emph{backend}. 
%\systemprefix{} supports dependent operations by allowing the claim of a future to happen both on the backend as well as the frontend. 
To revisit Example~\ref{ex:depend}, the execution of Line 15-16 emits both operations into the operation stream. At the \emph{backend}, the argument of the $\mathtt{updatePayload}$ expression, a future value, can be claimed upon the completion of processing $\mathtt{queryNode}$, without any interaction with the frontend. 

%The two features of dependent operations and asynchronous operation processing turn out interacting with each other. It implies that operations in the operation stream may dependent on each other. 

\subsection{Motivating Usage Context 2: Iterative Graph Analytics}
\label{sec:corepr}

Graph analytics and mining is another important workload for 
graph processing. These programs have an algorithm-centric core for computing often graph-theoretic properties. 
%Simple graph analytics may
%only involve one traversal of the graph,
%such as finding the highest degree of a graph.
Most realistic algorithms involve
multiple \emph{iterations}, each of which involves
non-trivial computations based on graph payload and topological information.
In this section, we use PageRank~\cite{pagerank} as an example to describe how
\systemprefix{} can capture the essential behavior of iterative graph analytics. 
For iterative algorithms such as PageRank, each iteration is often referred to as a \emph{super-step}.
\inlong{%%%%%%%%%%%%%%%%%%%%%%%%% begin inlong %%%%%%%%%%%%%%%%%%%%%%%%%
At the first super-step,
each node is initialized with an identical payload value of
$\frac{1}{N}$ where $N$ is the number of nodes in the graph.
For each ensuing super-step, each node updates its payload with a new
value that aggregates the payload values of its in-degree 
adjacent nodes.
}%%%%%%%%%%%%%%%%%%%%%%%%% end inlong %%%%%%%%%%%%%%%%%%%%%%%%%

\begin{example}[\textsc{CorePR} in \systemprefix{} Sugared Syntax]
Fig.~\ref{fig:pagerank} presents a 30-super step PageRank algorithm in \systemprefix{} sugared syntax.
Lines 2-4 compute the number of nodes.
Lines 5-6 initialize the node payloads.
Line 7 iterates over super-steps.
Each super-step has two sub-steps.
The first sub-step, at Lines 9-11,
computes the sum of payload values for each node's in-degree adjacent nodes.
%The list $\mathtt{neighborPSums}$ maps node keys to
%this sum.
The second sub-step, a loop at Lines 12-14,
updates each node with a new payload value,
by utilizing $\mathtt{fPG}$, the core PageRank aggregation function. 
\end{example}

This program uses an extensive set of language features (encodable by $\lambda$ calculus),
but the \systemprefix{}-specific expressions are only two graph processing operations,
shown in {\color{blue}blue}. The $\mathtt{mapVal}$-$\mathtt{foldVal}$ pair should be familiar to
MapReduce~\cite{mapreduce} users and functional programmers. %, a popular data processing framework, and functional
%programmers in general. 
Here, \emph{selective} map/fold is supported: the last argument for the $\mathtt{mapVal}$ or $\mathtt{foldVal}$ operation is a list of keys,
and the operation will only be applied to nodes whose keys appear in the list.
%
%When $\mathtt{mapVal}/\mathtt{foldVal}$ is non-selective, we use special value $\KSstar$ as the last argument, i.e., all node keys in the graph. 
%
Both operations can inspect the node, and hence the mapping function and (the second argument of) the folding function take in a triple as an argument (recall that a graph node consists of a key, a payload, and its out-edge adjacency list). The $\mathtt{mapVal}$ operation updates the payload of the node, and the $\mathtt{foldVal}$ operation folds into a value in the same type of the payload.

Through the lens of \systemprefix{}, an insight drawn from this example is that an iterative graph analytics
program consists of numerous graph processing operations --- within a
super-step and across super-steps --- continuously applied to
the graph being analyzed. Just as the \textsc{CoreSocial} example in \S\ref{sec:coresocial},  
the evaluation of a graph processing operation such as $\mathtt{mapVal}$ or $\mathtt{foldVal}$ in \textsc{CorePR} 
does not require the backend graph to complete its processing. 
Instead, each evaluation at Line 3, Line 6, Line 9, and Line 14, results in emitting a graph processing operation to the operation stream, in the same behavior we described in Example~\ref{ex:streaming}. 

\begin{example}[Iterative Graph Analytics as Continuous Graph Processing]
\label{ex:itera}
Assume \textsc{CorePR} is applied to a graph with 2 nodes whose keys are $\mathtt{k}_1$ and $\mathtt{k}_2$. Its execution through the initial step and the first two super-steps produces the following operation stream:
\begin{center}
\scalebox{\scalemath}{
$\begin{array}{l}
[\mathtt{mapVal}\ \mathtt{fPInit\ keys},
\\
\mathtt{foldVal}\ \mathtt{fPSum}_{1,1}\ 0\ \mathtt{keys},
\mathtt{foldVal}\ \mathtt{fPSum}_{1,2}\ 0\ \mathtt{keys},
\mathtt{mapVal}\ \mathtt{fPG}_{1,1}\ [ \mathtt{k}_1 ],
\mathtt{mapVal}\ \mathtt{fPG}_{1,2}\ [ \mathtt{k}_2 ],
\\
\mathtt{foldVal}\ \mathtt{fPSum}_{2,1}\ 0\ \mathtt{keys},
\mathtt{foldVal}\ \mathtt{fPSum}_{2,2}\ 0\ \mathtt{keys},
\mathtt{mapVal}\ \mathtt{fPG}_{2,1}\ [ \mathtt{k}_1 ],
\mathtt{mapVal}\ \mathtt{fPG}_{2,2}\ [ \mathtt{k}_2 ]
]
\end{array}$
}
\end{center}
where $\mathtt{fPSum}_{i,j}$ and $\mathtt{fPG}_{i,j}$ are the closures of $\mathtt{fPSum}$ and $\mathtt{fPG}$ functions
in the $i^{th}$ iteration of the loop starting at Line 9
and the $j^{th}$ iteration of the loop starting at Line 12.
\end{example}

%%%%%%%%%%%%%%%%%%%%%%%%%%%%
%%%%%%%%%%%%%%%%%%%%%%%%%%%%
%%%%%%%%%%%%%%%%%%%%%%%%%%%%
%%%%%%%%%%%%%%%%%%%%%%%%%%%%
%%%%%%%%%%%%%%%%%%%%%%%%%%%%
%%%%%%%%%%%%%%%%%%%%%%%%%%%%

\subsection{IOP in Continuous Graph Processing}
%\subsection{In-Graph Operation Streams and Lazy Operation Processing}
\label{subsec:ingraph}

%To see the connection between our stream view and how graph processing is conventionally perceived, observe that passing an operation through the in-graph operation stream in \systemprefix{} is analogous to a \emph{graph traversal} standard in processing an operation: the graph nodes are reached one by one following the traversal order (i.e., the operation is \emph{propagated}), and along the way, the in-graph computation is performed each time the operation reaches a node it is intended for (i.e., the operation is \emph{realized}). What makes in-graph operation streams an appealing model is that it naturally supports ``many operations at a time'' processing: it allows the traversals of all operations in the in-graph operation stream to \emph{co-exist}, independent of the underlying execution model being serial or parallel. 

Taking a per-operation view, graph processing can be viewed as a process that reaches the graph nodes one by one following the traversal order (\emph{propagation}), and along the way, computation is performed when the operation reaches the node(s) it is intended for (\emph{realization}). 
The default ``baseline'' behavior in graph processing is \emph{eager processing}: the processing of an operation must be completed before another is started. 
\inlong{%%%%%%%%%%%%%%%%%%%%%%%%% begin inlong %%%%%%%%%%%%%%%%%%%%%%%%%
This is independent of the frontend behavior: even if the frontend produces an operation stream, the backend may still be designed by processing each operation in this ``top-level'' operation stream
\emph{one at a time}.
}%%%%%%%%%%%%%%%%%%%%%%%%% end inlong %%%%%%%%%%%%%%%%%%%%%%%%%

\begin{example}[Eager Processing]
\label{ex:eager}
%Consider Lines 9-10 in Fig.~\ref{fig:serverclient} again. 
%Let us continue with the asynchronous semantics at the front end and assume both operations are now emitted to the operation stream. 
If one were to apply eager processing for executing Lines 9-10 in Fig.~\ref{fig:serverclient}, the backend would process $\processor_1$ first, traversing through $\mathtt{eve}$ and $\mathtt{deb}$, and finally realizing at the latter. After the completion of $\processor_1$, the traversal of the graph may start for $\processor_2$, 
through nodes $\mathtt{eve}$, $\mathtt{deb}$,
and $\mathtt{cam}$.
\end{example}

In contrast, \systemprefix{} naturally supports ``many operations at a time'' processing: %it allows the traversals of all operations in the in-graph operation stream to \emph{co-exist}. % independent of the underlying execution model being serial or parallel. 

% A focus of our calculus is to study the behavior of
% \emph{backend} graph processing, i.e., how the graph processes
% each operation after they have been (asynchronously) submitted
% to the stream.

% Against this baseline, the main contribution
% of \systemprefix{} is to rigorously define
% a novel notion of lazy data processing,
% which enables \emph{operation-carrying data}:
% the graph nodes may carry the operations to be processed
% over the graph itself.
% An operation may be
% \emph{partially} processed,
% \emph{i.e.,} an operation may
% traverse a portion of the graph
% and delay the traversal of the remaining graph.

%The delayed operation is carried by
%the graph node the traversal has so far reached,
%in an \emph{operation store} associated with that node.
% \pnote{top level $\ostream$, no root node}
% For uniform treatment, we assume each graph has a top-level (imagine) 
% root node, which is also associated with an operation store, intuitively 
% capturing the operation stream that has been issued by the program, but has not yet 
% been processed. 

%Eager and lazy graph processing exhibit
%different ordering characteristics
%for operation processing.

\newcommand{\basegraph}{%
    \node[nnode] at ($(server) + (0.5,-1.55)$) (eve) {$\mathtt{eve}$};
    \node[nnode] at ($(eve) + (-2.0,0)$) (deb) {$\mathtt{deb}$};
    \node[nnode] at ($(deb) + (1.0,-1.0)$) (cam) {$\mathtt{cam}$};
    \node[nnode] at ($(cam) + (-1.0,-1.0)$) (bob) {$\mathtt{bob}$};
    \node[nnode] at ($(bob) + (2.0,0)$) (amy) {$\mathtt{amy}$};
    \begin{pgfonlayer}{background}
    \coordinate (evex) at ($(eve) + (0,0.3)$);
    \coordinate (amyx) at ($(amy) + (0,0.3)$);
    \draw[sarrow, text opacity=1, rounded corners] ($(eve) + (1.8,0.3)$) to ($(eve) + (0,0.3)$) to ($(deb) + (0,0.3)$) to ($(cam) + (0,0.3)$) to ($(bob) + (0,0.3)$) to ($(amy) + (0,0.3)$) to ($(amy) + (1.5,0.3)$);
    \end{pgfonlayer}
}

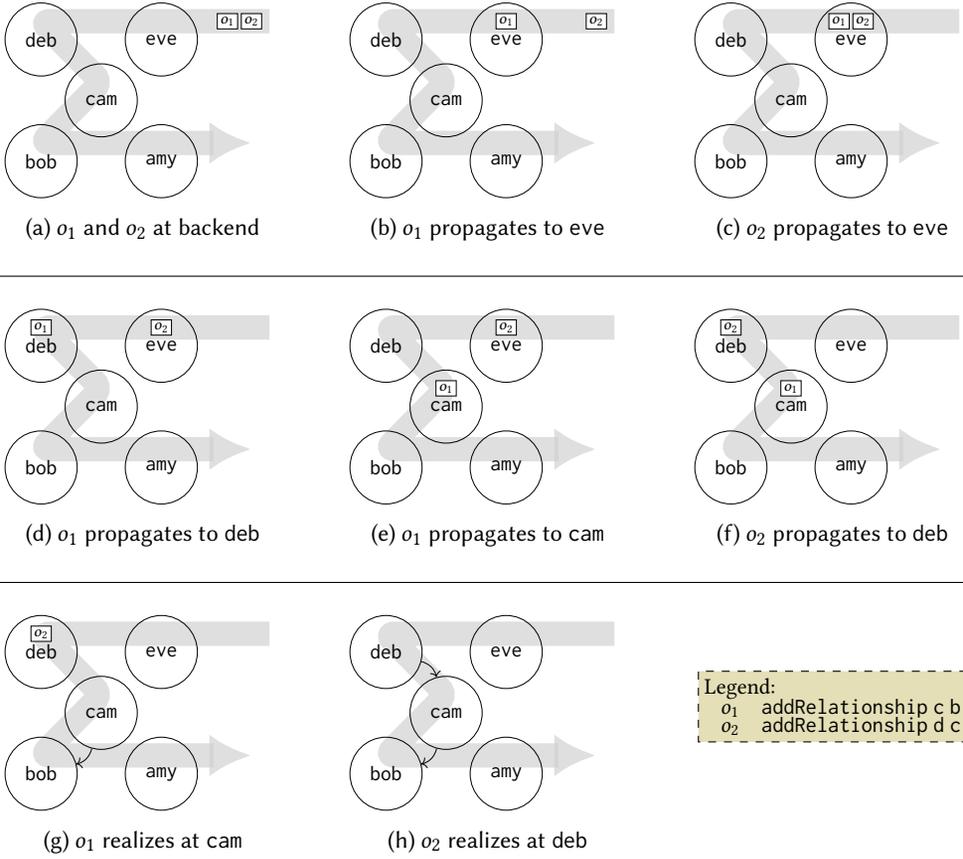
\begin{figure}[t]
\centering

\begin{tabular}{c@{}c@{}c}
\begin{subfigure}{0.33\linewidth}
\centering
  \scalebox{\scalegraph}{
  \begin{tikzpicture}
    
    \basegraph

    \node[op] at ($(eve) + (1.5,0.3)$) (op2) {\scriptsize$\processor_2$};
    \node[op] at ($(eve) + (1.1,0.3)$) (op1) {\scriptsize$\processor_1$};
    
  \end{tikzpicture}
  }
  \caption{$\processor_1$ and $\processor_2$ at backend}
\end{subfigure}
&
\begin{subfigure}{0.33\linewidth}
\centering
  \scalebox{\scalegraph}{
  \begin{tikzpicture}
  
    \basegraph
    
    \node[op] at ($(eve) + (1.5,0.3)$) (op2) {\scriptsize$\processor_2$};
    \node[op] at ($(eve) + (0,0.3)$) (op1) {\scriptsize$\processor_1$};
    
  \end{tikzpicture}
  }
  \caption{$\processor_1$ propagates to $\mathtt{eve}$}
\end{subfigure}
&
\begin{subfigure}{0.33\linewidth}
\centering
  \scalebox{\scalegraph}{
  \begin{tikzpicture}
  
    \basegraph
    
    \node[op] at ($(eve) + (0.2,0.3)$) (op2) {\scriptsize$\processor_2$};
    \node[op] at ($(eve) + (-0.2,0.3)$) (op1) {\scriptsize$\processor_1$};
    
  \end{tikzpicture}
  }
  \caption{$\processor_2$ propagates to $\mathtt{eve}$}
\end{subfigure}
\\\\\hline\\
\begin{subfigure}{0.33\linewidth}
\centering
  \scalebox{\scalegraph}{
  \begin{tikzpicture}
  
    \basegraph
    
    \node[op] at ($(eve) + (0,0.3)$) (op2) {\scriptsize$\processor_2$};
    \node[op] at ($(deb) + (0,0.3)$) (op1) {\scriptsize$\processor_1$};
    
  \end{tikzpicture}
  }
  \caption{$\processor_1$ propagates to $\mathtt{deb}$}
\end{subfigure}
&
\begin{subfigure}{0.33\linewidth}
\centering
  \scalebox{\scalegraph}{
  \begin{tikzpicture}
  
    \basegraph
    
    \node[op] at ($(eve) + (0,0.3)$) (op2) {\scriptsize$\processor_2$};
    \node[op] at ($(cam) + (0,0.3)$) (op1) {\scriptsize$\processor_1$};
    
  \end{tikzpicture}
  }
  \caption{$\processor_1$ propagates to $\mathtt{cam}$}
\end{subfigure}
&
\begin{subfigure}{0.33\linewidth}
\centering
  \scalebox{\scalegraph}{
  \begin{tikzpicture}
  
    \basegraph
    
    \node[op] at ($(deb) + (0,0.3)$) (op2) {\scriptsize$\processor_2$};
    \node[op] at ($(cam) + (0,0.3)$) (op1) {\scriptsize$\processor_1$};
    
  \end{tikzpicture}
  }
  \caption{$\processor_2$ propagates to $\mathtt{deb}$}
\end{subfigure}
\\\\\hline\\
\begin{subfigure}{0.33\linewidth}
\centering
  \scalebox{\scalegraph}{
  \begin{tikzpicture}
  
    \basegraph
    
    \node[op] at ($(deb) + (0,0.3)$) (op2) {\scriptsize$\processor_2$};
    %\node[op] at ($(cam) + (0,0.3)$) (op1) {\scriptsize$\processor_1$};
    \draw[ledge, bend left] (cam) to (bob);
    
  \end{tikzpicture}
  }
  \caption{$\processor_1$ realizes at $\mathtt{cam}$}
\end{subfigure}
&
\begin{subfigure}{0.33\linewidth}
\centering
  \scalebox{\scalegraph}{
  \begin{tikzpicture}
  
    \basegraph
    
    %\node[op] at ($(deb) + (0,0.3)$) (op2) {\scriptsize$\processor_2$};
    \draw[ledge, bend left] (deb) to (cam);
    %\node[op] at ($(cam) + (0,0.3)$) (op1) {\scriptsize$\processor_1$};
    \draw[ledge, bend left] (cam) to (bob);
    
  \end{tikzpicture}
  }
  \caption{$\processor_2$ realizes at $\mathtt{deb}$}
\end{subfigure}
&
\begin{tikzpicture}[scale=\exgraphscale]
  \begin{scope}[every node/.style={scale=\exgraphscale,font=\fontsize{12pt}{\baselineskip}\selectfont}]
    \node[draw,dashed,fill=olive!25,align=left] at (1.9,0.5) (legend) {Legend: \\ \hspace{0.5em} $\processor_1$ \hspace{0.5em} $\mathtt{addRelationship\ c\ b}$\\ \hspace{0.5em} $\processor_2$ \hspace{0.5em} $\mathtt{addRelationship\ d\ c}$};
  \end{scope}
\end{tikzpicture}
\end{tabular}

\caption{In-Graph Operation Streams for Fig.~\ref{fig:serverclient} Lines 9-10.}
\label{fig:datacentriclaziness}

\end{figure}

%Placing this flexible model at the core of \systemprefix{} is good news for generality, in that \systemprefix{} can support expressive flavors of TLO, LIP and parallelism, while at the same time subsuming more basic forms of graph processing, e.g., ``one operation at a time'', when simple restrictions are placed over the more general semantic model (see ~\S\~ref{?} for details). 

%As discussed in \S~\ref{sec:intro}, in-graph operation streams are a central feature in our operational semantics, playing a pivotal rule in the design space of lazy operation processing and parallelism. % is one of the two essential design dimensions of streaming graphs. A central feature in \systemprefix{} that contributes to the generality of our support for lazy operation processing is i
%LIP in \systemprefix{} is supported through in-graph operation streams, which we now describe through the \textsc{CoreSocial} example. 

\begin{example}[In-Graph Operation Streams]
\label{ex:opstream}
Fig~\ref{fig:datacentriclaziness} illustrates 8 runtime configurations of the backend graph for \textsc{CoreSocial} in a \systemprefix{} reduction sequence. The first one coincides with the moment when the processing of Line 1-8 in Fig.~\ref{fig:serverclient} is completed, and the operations in Line 9-10 have been emitted but not processed. These two operations $\processor_1$ and $\processor_2$ flow through the graph nodes following the traversal of $\mathtt{eve}$, $\mathtt{deb}$, $\mathtt{cam}$, $\mathtt{bob}$, $\mathtt{amy}$, in that order.
Intuitively, the in-graph stream view entails that the traversal of multiple operations may co-exist: for configurations (b)(c)(d)(e)(f), neither $\processor_1$ nor $\processor_2$ is completed. In addition, the propagation steps for different operations may intermingle, the first 3 transitions in Fig.~\ref{fig:datacentriclaziness} are propagation steps for $\processor_1$, $\processor_2$, and $\processor_1$, respectively. 

%for example.
%Assume both operations are in the backend operation store.
%Let us reconsider the two operations
%we previously considered for eager graph processing
%in Example~\ref{ex:eager}.
% a possible execution
%of lazy database processing.
%Observe that the processing of the second operation
%can start before the completion of the first operation,
%as shown in Fig.~\ref{fig:datacentriclaziness}(d).
%In Fig.~\ref{fig:datacentriclaziness}(b),
%the first operation is carried by $\mathtt{deb}$
%and the second is carried by $\mathtt{eve}$.
%These are examples of operation-carrying data
%for operations that have started being processed
%but not yet completed.
%In Fig.~\ref{fig:datacentriclaziness}(e),
%the first operation will be completed
%by adding an edge from deb to $\mathtt{cam}$.
%
\end{example}

%\dnote{maybe plug in the bit on the traversal order}

The behavior exhibited in Example~\ref{ex:opstream} is \emph{incremental propagation}:  %if we take a per-operation view: 
the processing of $\processor_1$ can be deferred without the need of ``rushing'' to its realization. When $\processor_1$ is deferred, the runtime can process (i.e., either propagate or realize) another operation, such as the later emitted $\processor_2$. 
\inlong{%%%%%%%%%%%%%%%%%%%%%%%%% begin inlong %%%%%%%%%%%%%%%%%%%%%%%%%
Fig.~\ref{fig:eagerlazydraw} provides a timeline to visualize the difference between eager and lazy processing. Whereas Fig.~\ref{fig:eagerlazydraw:eager} requires no overlapping in processing time between the two operations,  the timeline of Fig.~\ref{fig:eagerlazydraw:lazy} is interleaved. 
When an operation takes a propagation or realization step, another co-existing operation  in the in-graph operation stream can simply be held at the node it has flows to in the operation stream, which we informally refer to as the \emph{host node} of the operation. We also say the latter node is now in \emph{suspension}. 
}%%%%%%%%%%%%%%%%%%%%%%%%% end inlong %%%%%%%%%%%%%%%%%%%%%%%%%

%Two examples of suspension are Fig~\ref{fig:datacentriclaziness}(c) (where the first operation yields for the second operation), and Fig~\ref{fig:datacentriclaziness}(d) (where the second operation yields for the first operation). Taking a global view, the processing timeline of the 5 operations may overlap. 

%LIP has two important consequences on continuous graph processing. First, 
%recall that in \S~\ref{subsec:dcalculus}, lazy operation processing goes hand in hand with temporal locality optimization, 
%in-graph lazy operation processing thus enables \emph{in-graph} TLO, to be detailed shortly. 

As a general calculus, \systemprefix{} places no restriction on the ``schedule'' of operation stream processing: when multiple operations are processed, a \emph{non-deterministic} choice can be made as to which operation should take a step. For example, instead of transitioning from Fig.~\ref{fig:datacentriclaziness}(b) to Fig.~\ref{fig:datacentriclaziness}(c), the program runtime may choose to have $\processor_1$ take another propagation step to $\mathtt{deb}$. A data processing system that supports non-deterministic executions and deterministic results --- which \systemprefix{} enjoys --- is good news for adaptiveness support, which now illustrate through a straggler example, a classic problem in data processing~\cite{188988}.

%We defer the discussion of the former to \S~\ref{sec:parallel}, and now focus on the latter. 
%\pnote{defer to the technical report}
 
\inlong{%%%%%%%%%%%%%%%%%%%%%%%%% begin inlong %%%%%%%%%%%%%%%%%%%%%%%%%

\begin{figure}
\centering

\begin{tabular}{ccc}
\begin{subfigure}{0.33\textwidth}
\centering

\scalebox{\scalegraph}{
\begin{tikzpicture}

\draw[-Triangle] (-0.1,0) -- (4.7,0);
\draw (-0.1,-0.1) -- (-0.1,0.1);
\node at (2.4, -0.5) {Elapsed Time};

\node[tl, text width=2cm] at (0.1,1) (b1) {$\processor_1$};

\node[tl, text width=2cm] at (2.4,1.75) (b2) {$\processor_2$};

\end{tikzpicture}
}

\caption{Eager Operation Processing}
\label{fig:eagerlazydraw:eager}

\end{subfigure}
&
\begin{subfigure}{0.33\textwidth}
\centering

\scalebox{\scalegraph}{
\begin{tikzpicture}

\draw[-Triangle] (-0.1,0) -- (4.7,0);
\draw (-0.1,-0.1) -- (-0.1,0.1);
\node at (2.4, -0.5) {Elapsed Time};

\node[tl, text width=3.3cm] at (0.1,1) (b1) {$\processor_1$};
\node[tlp, inner xsep=0cm, text width=1.2cm] at (1.2,1) {};

\node[tl, text width=3.2cm] at (1.2,1.75) (b2) {$\processor_2$};
\node[tlp, inner xsep=0cm, text width=1.2cm] at (2.4,1.75) {};

\end{tikzpicture}
}

\caption{Lazy Operation Processing}
\label{fig:eagerlazydraw:lazy}

\end{subfigure}
&
\scalebox{\scalegraph}{
\begin{tikzpicture}
\node[draw,dashed,fill=olive!25,align=left] at (2.3,4.0) (legend) {Legend: \\ \hspace{2.5em} execution \\[1.5ex] \hspace{2.5em} suspension};
\node[tl, text width=0.50cm, inner ysep=0.19cm] at ($(legend.west) + (0.20,0.05)$) {};
\node[tlp, text width=0.20cm, inner ysep=0.19cm] at ($(legend.west) + (0.20,-0.4)$) {};
\end{tikzpicture}
}
\end{tabular}

\caption{A Timeline for The Operation Processing of 2 Operations Emitted at Lines 9-10 in
Fig.~\ref{fig:serverclient}
%\pnote{what is `eager'/`lazy' here}
}
\label{fig:eagerlazydraw}

\end{figure}
}%%%%%%%%%%%%%%%%%%%%%%%%% end inlong %%%%%%%%%%%%%%%%%%%%%%%%%

% \begin{example}[Interleaved Operation Processing in Lazy Graph Processing]
% Consider the same five operations as in
% the previous example.
% A processing timeline for lazy graph processing
% is presented in Fig.~\ref{fig:eagerlazydraw:lazy}.
% The green shading in the boxes represents
% time dedicated for processing of the operation.
% The gray shading represents a temporarily suspension
% of the operation's processing.
% %
% Operations in lazy graph processing
% in an interleaved fashion; any operation
% may be processed when another is suspended.
% \end{example}

\begin{figure}[t]
\centering

\begin{tabular}{@{}c@{\ \ \ \ }|@{\ \ \ }c@{}}
\begin{subfigure}{0.5\textwidth}
\centering

\scalebox{\scalegraph}{
\begin{tikzpicture}
%\begin{tikzpicture}[scale=\exgraphscale]
  %\begin{scope}[scale=\exgraphscale,every node/.style={scale=\exgraphscale,font=\fontsize{12pt}{\baselineskip}\selectfont}]
\node[draw,dashed,fill=olive!25,align=left] at (1.3,4.0) (legend) {Legend: \\ \hspace{2.5em} straggling \\ \hspace{0.5em} $\processor_{4\textrm{--}7}$ \hspace{0.5em} 4$^\textrm{th}$--7$^\textrm{th}$ operations in Example~\ref{ex:itera}};
\node[tlp, inner ysep=0.19cm, text width=0.20cm] at ($(legend.west) + (0.20,0.0)$) {};
\draw[thick, dotted] ($(legend.west) + (0.2,0.0)$) -- ($(legend.west) + (0.9,0.0)$);
\draw[thick, dotted] ($(legend.west) + (0.2,0.12)$) -- ($(legend.west) + (0.9,0.12)$);
\draw[thick, dotted] ($(legend.west) + (0.2,-0.12)$) -- ($(legend.west) + (0.9,-0.12)$);
  %\end{scope}
%\end{tikzpicture}

\draw (-0.1,0) -- (1.6,0);
\draw[dotted] (1.6,0) -- (3.8,0);
\draw (3.8,0) -- (6.3,0);
\draw[dotted] (-0.1,0) -- (-0.6,0);
\draw[dotted,-Triangle] (6.3,0) -- (6.9,0);
\node at (2.7, -0.75) {Elapsed Time};
\node[text=black!75] at (-0.75, 1.75) (ss1) {Super Step 1};
\node[text=black!75] at (-0.75, 2.5) (ss2) {Super Step 2};

\node[tl] at (0.1,1) (b1) {$\processor_4$};

\node[tl, text width=0.5cm] at (0.9,1.75) (b2) {$\processor_5$};
\node[tl, text opacity=0, text width=0.5cm] at (3.7, 1.75) {$\nkey_1$};
\node[tlp, text width=1.65cm] at (1.6,1.75) {};
\draw[thick, dotted] (1.7, 1.75) -- (3.7, 1.75);
\draw[thick, dotted] (1.7, 1.95) -- (3.7, 1.95);
\draw[thick, dotted] (1.7, 1.55) -- (3.7, 1.55);

\node[tl] at (4.5,2.5) (b3) {$\processor_6$};

\node[tl] at (5.3,3.25) (b4) {$\processor_7$};

\node[fit=(b1)(b2)(ss1)] (sb1) {};
\node[fit=(b3)(b4)(ss2)] (sb2) {};
\coordinate (end) at (6.9,0);
\draw[draw=black!75, decorate, decoration={zigzag, amplitude=0.3mm}] (sb1.north west) -- (sb1.north east -| end);
\end{tikzpicture}
}

\caption{Eager Processing}
\label{fig:timeline:eager}

\end{subfigure}
&
\begin{subfigure}{0.5\textwidth}
\centering

\scalebox{\scalegraph}{
\begin{tikzpicture}

\node at (1.3,4.6) {};

\draw (-0.1,0) -- (3.6,0);
\draw[dotted,-Triangle] (3.6,0) -- (6.1,0);
\draw[dotted] (-0.1,0) -- (-0.6,0);
\node at (2.7, -0.75) {Elapsed Time};
\node[text=black!75] at (-0.75, 1.75) (ss1) {Super Step 1};
\node[text=black!75] at (-0.75, 2.5) (ss2) {Super Step 2};

\node[tl, text width=1.75cm] at (0.1,1) (b1) {$\processor_4$};
\node[tlp, text width=0.5cm] at (0.6,1) {};

\node[tl, text width=3.7cm] at (0.6,1.75) (b2) {$\processor_5$};
\node[tlp, text width=3.9cm] at (1.6,1.75) {};
\node[tlp, inner ysep=0.25cm, text width=0.1cm] at (9.0, 1.75) {};
\draw[thick, dotted] (3.6, 1.75) -- (6.0, 1.75);
\draw[thick, dotted] (3.6, 1.95) -- (6.0, 1.95);
\draw[thick, dotted] (3.6, 1.55) -- (6.0, 1.55);

\node[tl, text width=1.25cm] at (2.1,2.5) (b3) {$\processor_6$};
\node[tlp, text width=0.25cm] at (2.6, 2.5) {};

\node[tl, text width=0.50cm] at (2.6,3.25) (b4) {$\processor_7$};
%\node[tlp, text width=0.19cm, inner xsep=0.0cm] at (3.9, 3.25) {};

\node[fit=(b1)(b2)(ss1)] (sb1) {};
\node[fit=(b3)(b4)(ss2)] (sb2) {};
\coordinate (end) at (6.1,0);
\draw[draw=black!75, decorate, decoration={zigzag, amplitude=0.3mm}] (sb1.north west) -- (sb1.north east -| end);

\end{tikzpicture}
}

\caption{Adaptive Processing}
\label{fig:timeline:lazy}

\end{subfigure}
\end{tabular}

\caption{PageRank Execution with Stragglers}
\label{fig:timeline}

\end{figure}
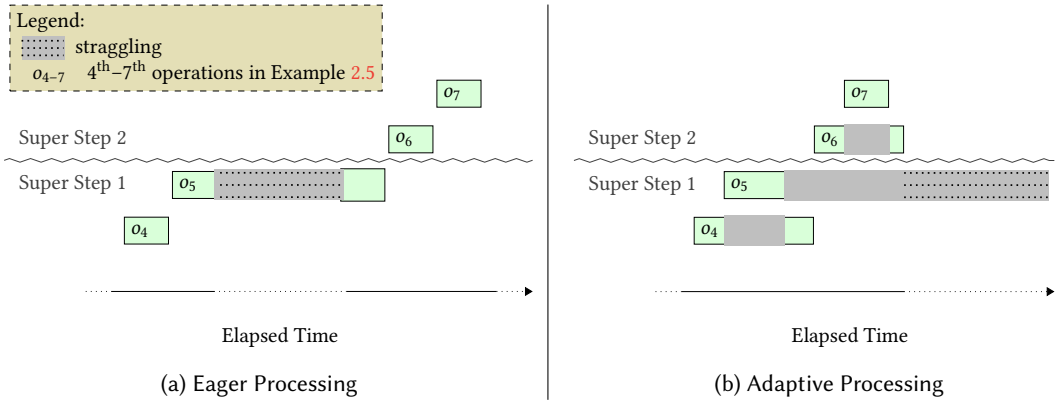

% Due to system resource fluctuations and transcient failures, a data processing 
% operation --- such as a $\mathtt{mapVal}$ operation or a $\mathtt{foldVal}$ operation in
% \textsc{CorePR} --- may be suspended during processing, and overall take a significantly 
% longer time to complete than it 
% would normally do. An advantage of non-determinism latent in the in-graph operation 
% stream is that it provides an opportunity to mitigate the impact of stragglers. 
% In the following example, we show how this may be achieved.

\begin{example}[Super-Step Blending for Straggler Mitigation]
\label{ex:superstepblend}
Fig.~\ref{fig:timeline} illustrates two timelines of execution
of \textsc{CorePR} in the same setting as Example~\ref{ex:itera}. We only plot each timeline to start when Line 14 is reached for the first time (the fourth operation in Example~\ref{ex:itera}).
Due to system resource fluctuations and transient failures, let us assume the processing of operation $\processor_5$ may be suspended, becoming a straggler.
In Fig.~\ref{fig:timeline:eager}, the slowdown by the straggler delays the beginning of the next super-step. 
In Fig.~\ref{fig:timeline:lazy} however, while the straggler is suspended, operation $\processor_6$ in
the second super-step may start, intuitively interleaving the two super-steps. 
The straggling operation will indeed have to eventually
complete, but the program is not delayed by its straggling.
\end{example}

% \inlong{%%%%%%%%%%%%%%%%%%%%%%%%% begin inlong %%%%%%%%%%%%%%%%%%%%%%%%%

% \systemprefix{} captures the essence of in-graph operation streams, including the latent non-deterministic executions. The decision on how a choice is made in the presence of non-determinism is out of the scope of our calculus.
% %, so long as any choice must preserve correctness, i.e., the deterministic results. 
% A client calculus of \systemprefix{}can make decisions based on a variety of factors, e.g., systems utilization, user-defined priorities, or deadlines.
% }%%%%%%%%%%%%%%%%%%%%%%%%% end inlong %%%%%%%%%%%%%%%%%%%%%%%%%

%\dnote{need to think a bit what it means for the straggler to in the calculus. Does this mean the node-level computation is non-terminating? How do we deal with that? } 

Another dimension of IOP support in \systemprefix{} is \emph{incremental load update}: 

\begin{example}[Incremental Load Update]
Suppose the operation at Line 14 is processed at the backend and the node indicated by $\mathtt{nk}$ is reached whose payload value is 5. Our calculus will update the payload of the node with expression $\mathtt{fPG\ 5}$. The realization step is completed without the need to evaluate $\mathtt{fPG\ 5}$ to a value. The intra-node evaluation of the payload expression can be deferred and resumed later non-deterministically. From now on, we informally refer to any expression inside a runtime graph node as a \emph{load}. 
\end{example} 

\inlong{%%%%%%%%%%%%%%%%%%%%%%%%% begin inlong %%%%%%%%%%%%%%%%%%%%%%%%%

Furthermore, the evaluation of the payload expression above can also be performed non-deterministically. A client calculus of \systemprefix{} can impose the evaluation schedule based on a variety of factors, e.g., systems utilization, user-defined priorities, or deadlines. 

}%%%%%%%%%%%%%%%%%%%%%%%%% end inlong %%%%%%%%%%%%%%%%%%%%%%%%%

%The evaluation of the payload expression above can be deferred. The reduction steps of its evaluation (which we call \emph{load reductions}) can interleave with those of processing operations in the in-graph operation stream (which we call \emph{task reductions}) such as propagating an operation in a non-deterministic manner. Load reductions over different host nodes may also interleave. 

It should be made clear that non-deterministic executions are a feature only when they can introduce deterministic results. In other words, a non-deterministic propagation is \emph{not} an arbitrary propagation. In particular, note that the operations in the operation stream form a \emph{chronological order} that indicates the order of emission. It must be preserved (unless TLO rules allows for reordering). This requirement can be understood through an example.

\begin{example}[Chronological Order Preservation]
Let us assume the payload value in \texttt{amy} is initially 1, i.e., $\mathtt{n_{amy}} = 1$. Operation $\processor_1$ is an operation to double the payload of \texttt{amy} while $\processor_2$ is an operation to add the payload of \texttt{amy} by 10. After the two operations are completed, \texttt{amy} should have a payload of 12. Should we allow $\processor_2$ to ``swap'' with $\processor_1$, the payload of \texttt{amy} would be 22.
\end{example}

\inlong{%%%%%%%%%%%%%%%%%%%%%%%%% begin inlong %%%%%%%%%%%%%%%%%%%%%%%%%

\begin{figure}[t]
\centering

\begin{tabular}{ccc}
\begin{subfigure}{0.33\linewidth}
\centering
\scalebox{\scalegraph}{
  \begin{tikzpicture}
  
    \basegraph
    
    \node[op] at ($(eve) + (0.2,0.3)$) (op2) {\scriptsize$\processor_2$};
    \node[op] at ($(eve) + (-0.2,0.3)$) (op1) {\scriptsize$\processor_1$};
    
  \end{tikzpicture}
  }
  \caption{$\processor_1$ and $\processor_2$ at $\mathtt{eve}$}
\end{subfigure}
&
\begin{subfigure}{0.33\linewidth}
\centering
\scalebox{\scalegraph}{
  \begin{tikzpicture}
  
    \basegraph
    
    \node[op] at ($(deb) + (0.2,0.3)$) (op2) {\scriptsize$\processor_2$};
    \node[op] at ($(deb) + (-0.2,0.3)$) (op1) {\scriptsize$\processor_1$};
    
  \end{tikzpicture}
  }
  \caption{$\processor_1$ and $\processor_2$ propagate to $\mathtt{deb}$}
\end{subfigure}
&
\scalebox{\scalegraph}{
\begin{tikzpicture}
    \node[draw,dashed,fill=olive!25,align=left] at (1.9,0.5) (legend) {Legend: \\ \hspace{0.5em} $\processor_1$ \hspace{0.5em} $\mathtt{addRelationship\ c\ b}$\\ \hspace{0.5em} $\processor_2$ \hspace{0.5em} $\mathtt{addRelationship\ d\ c}$};
\end{tikzpicture}
}
\end{tabular}

\caption{In-Graph Operation Batching}
\label{fig:batching}

\end{figure}
}%%%%%%%%%%%%%%%%%%%%%%%%% end inlong %%%%%%%%%%%%%%%%%%%%%%%%%

\subsection{TLO in Continuous Graph Processing}
\label{subsec:tlo}

\inlong{%%%%%%%%%%%%%%%%%%%%%%%%% begin inlong %%%%%%%%%%%%%%%%%%%%%%%%%
With in-graph operation streams, the processing of individual operations may be deferred at an arbitrary node, and the opportunity of optimizing temporally adjacent operations arises when they happen to be deferred at the same host node. 
}%%%%%%%%%%%%%%%%%%%%%%%%% end inlong %%%%%%%%%%%%%%%%%%%%%%%%%

We now revisit the \textsc{SocialCore} example to illustrate the common forms of TLO that \systemprefix{} supports. 

%Lazy database processing is important because it enables important families
%of database optimizations.
%First, operations reaching the same node can be propagated together
%if they do not realize at their current node.
%This is called \emph{batching}.

\begin{example}[In-Graph Operation Batching]
\label{ex:batch}

%For data-intensive applications, operations may be received by the database at a rate higher than the latter can process. Thanks to lazy processing, the processing lifecycles of multiple operations may overlap. 
%When multiple operations propagate to the same node, an opportunity arises to batch together the remainder of their propagations. 
Consider Fig.~\ref{fig:datacentriclaziness}(c). Since neither $\mathtt{addRelationship}$ operation
realizes at $\mathtt{eve}$,
both may propagate in a ``batch'' to $\mathtt{deb}$ in one reduction step. 
\inlong{%%%%%%%%%%%%%%%%%%%%%%%%% begin inlong %%%%%%%%%%%%%%%%%%%%%%%%%
This is illustrated in Fig.~\ref{fig:batching}.
}%%%%%%%%%%%%%%%%%%%%%%%%% end inlong %%%%%%%%%%%%%%%%%%%%%%%%%
% illustrates the database processing of Fig.~\ref{fig:datacentriclaziness}(c) onward,
%enabled by batching.
\end{example}

\inlong{%%%%%%%%%%%%%%%%%%%%%%%%% begin inlong %%%%%%%%%%%%%%%%%%%%%%%%%

\begin{figure}[t]
\centering

\begin{tabular}{ccc}
\begin{subfigure}{0.33\linewidth}
\centering
  \hspace{-1.5cm}
\scalebox{\scalegraph}{
  \begin{tikzpicture}
  
    \basegraph

    \draw[ledge, out=135, in=135, looseness=2] (eve) to (bob);
    \draw[ledge, bend right=20] (eve) to (amy);
    \draw[ledge, bend left] (deb) to (cam);
    \draw[ledge, bend left] (cam) to (bob);
    \draw[ledge, bend right=20] (amy) to (eve);

    \node[op] at ($(deb) + (0.4,0.3)$) (op3) {\scriptsize$\processor_3$};
    \node[op] at ($(deb) + (0.0,0.3)$) (op2) {\scriptsize$\processor_2$};
    \node[op] at ($(deb) + (-0.4,0.3)$) (op1) {\scriptsize$\processor_1$};

  \end{tikzpicture}
  }
  \caption{$\processor_1$, $\processor_2$, and $\processor_3$ are at $\mathtt{deb}$}
\end{subfigure}
&
\begin{subfigure}{0.33\linewidth}
\centering
\hspace{-1.5cm}
\scalebox{\scalegraph}{
  \begin{tikzpicture}
  
    \basegraph

    \draw[ledge, out=135, in=135, looseness=2] (eve) to (bob);
    \draw[ledge, bend right=20] (eve) to (amy);
    \draw[ledge, bend left] (deb) to (cam);
    \draw[ledge, bend left] (cam) to (bob);
    \draw[ledge, bend right=20] (amy) to (eve);

    \node[op] at ($(deb) + (-0.4,0.3)$) (op3) {\scriptsize$\processor_3$};
    \node[op] at ($(deb) + (-0.0,0.3)$) (op2) {\scriptsize$\processor_2$};
    \node[op] at ($(deb) + (0.4,0.3)$) (op1) {\scriptsize$\processor_1$};

  \end{tikzpicture}
  }
  \caption{$\processor_3$ is reordered ahead of $\processor_1$ and $\processor_2$}
\end{subfigure}
&
\raisebox{-0.5cm}{
\scalebox{\scalegraph}{
\begin{tikzpicture}
    \node[draw,dashed,fill=olive!25,align=left] at (1.9,0.5) (legend) {Legend: \\ \hspace{0.5em} $\processor_1$ \hspace{0.5em} $\mathtt{queryNode\ b}$ \\ \hspace{0.5em} $\processor_2$ \hspace{0.5em} $\mathtt{updatePayload\ a\ nb}$ \\ \hspace{0.5em} $\processor_3$ \hspace{0.5em} $\mathtt{queryNode\ b}$};
\end{tikzpicture}
}
}
\end{tabular}

\caption{In-Graph Operation Reordering
%\dnote{this example needs update because you have changed line 17}
}
\label{fig:reordering}

\end{figure}
}
%%%%%%%%%%%%%%%%%%%%%%%%% end inlong %%%%%%%%%%%%%%%%%%%%%%%%%

%When multiple operations reaching the same graph node 
%are ``related'' according to some predefined rewriting rules, additional optimization opportunities may arise: 

\inlong{%%%%%%%%%%%%%%%%%%%%%%%%% begin inlong %%%%%%%%%%%%%%%%%%%%%%%%%

In \systemprefix{}, batching is supported through a pair of reduction rules over the in-graph operation stream --- one for grouping multiple operations into a unit of propagation (\emph{batching}) and the other for ungrouping in the opposite manner (\emph{unbatching}). Our support is \emph{general} in the sense that the selection of temporally adjacent operations --- how often they should be applied, where in the operation stream batching/unbatching should happen, and how many operations should form one batch --- is supported in a non-deterministic manner, and can be further refined by client calculi built on top of \systemprefix{}. In practice, batching is particularly useful when the graph processing backend is ``overloaded,'' i.e., the frontend places operations to the operation stream at a faster rate than the backend can process it. %Batching reduces the number of traversals needed for operation processing. 

Similarly, the other 3 forms of temporal locality optimization are also supported through a set of reduction rules that conservatively define how an operation stream can be rewritten without altering the computation result. The principle behind is simple: rule-based pattern matching and term rewriting over consecutive (and hence, temporally local) operations. We now show these optimizations through examples. 
}%%%%%%%%%%%%%%%%%%%%%%%%% end inlong %%%%%%%%%%%%%%%%%%%%%%%%%

\begin{example}[In-Graph Operation Reordering]
\label{ex:reorder}
Consider a configuration where three operations at Lines 15-17 in Fig.~\ref{fig:serverclient}
reach node $\mathtt{deb}$. 
The third operation, $\mathtt{queryNode\ b}$,
reads from the node $\mathtt{b}$ while
the second operation writes to node $\mathtt{a}$.
\systemprefix{} allows these last two operations to ``swap'' since they do not operate on the same node.
\inlong{%%%%%%%%%%%%%%%%%%%%%%%%% begin inlong %%%%%%%%%%%%%%%%%%%%%%%%%
This is illustrated in Fig.~\ref{fig:reordering}. 
}%%%%%%%%%%%%%%%%%%%%%%%%% end inlong %%%%%%%%%%%%%%%%%%%%%%%%%

\end{example}

\begin{example}[In-Graph Operation Fusing]
\label{ex:cancel}
%In social networks, a recently updated node is often updated again.
%\pnote{try to find citation}
Consider a configuration where two operations at Lines 19-20
in Fig.~\ref{fig:serverclient}
both reach node $\mathtt{eve}$.
%as illustrated in Fig.~\ref{fig:fusionexample}(a).
\systemprefix{} allows the $\mathtt{addRelationship}$ and $\mathtt{deleteRelationship}$ operations to ``cancel out'' so that further processing of both is avoided. 
\end{example}

\begin{example}[In-Graph Operation Reusing]
Let us follow up on Example~\ref{ex:reorder}. After swapping, 
two $\mathtt{queryNode\ b}$ operations are adjacent in the operation stream at node $\mathtt{deb}$.
\systemprefix{} allows the second instance to immediately return, referencing the return value of the first instance.
\end{example}

%The key take-away message of \systemprefix{} is to 
%demonstrate the relationship between
%lazy operation processing and temporal locality optimization:
%it is the former that enables the latter,
%leading to important optimizations for graph processing.

%Rewriting is well studied in programming languages.
%A non-goal of our calculus is to come up with an exhaustive list of 
%rewriting rules that cover all possibilities of reordering, fusing, or reusing.

An expressive feature of \systemprefix{} is that TLO optimizations 
%--- operation batching, operation reordering, and operation fusion --- 
may happen \emph{in-graph}. %\emph{at an arbitrary node}. 
%as long as the involved operations reside in the same operation store of that node. 
For example, batching in Example~\ref{ex:batch} happens at node $\mathtt{eve}$; reordering and reusing in Example~\ref{ex:reorder} at node $\mathtt{deb}$; and fusing in Example~\ref{ex:cancel} at node $\mathtt{eve}$. 

%This general design subsumes the more naive optimization design where the entire backend graph is viewed as a blackbox, and batching/reordering/fusing may only happen in an ``operation buffer'' at the boundary between the frontend and the backend~\cite{?}.   

In \systemprefix{}, TLOs are supported through rewriting rules over the operation stream. The principle behind is well known in programming languages: rule-based pattern matching and term rewriting over consecutive (and hence, temporally local) operations. In our calculus, TLOs are applied dynamically. 
%The alternative design would to view it as a compiler optimization problem: applying the rewriting rules at compile time. This latter approach is not feasible if we are able to 
This is aligned with our ``open-world'' assumption on the usage scenarios in practice: when the program is compiled, the operations may not be statically known yet. 
\inlong{%%%%%%%%%%%%%%%%%%%%%%%%% begin inlong %%%%%%%%%%%%%%%%%%%%%%%%%
In other words, the program we showed in Fig.~\ref{fig:serverclient} may well be a textual \emph{a posteriori} representation of an interactive program, where each line of graph processing operation is filled, say, into some graphical interface, and executed by the graph.  
}%%%%%%%%%%%%%%%%%%%%%%%%% end inlong %%%%%%%%%%%%%%%%%%%%%%%%%

\inlong{%%%%%%%%%%%%%%%%%%%%%%%%% begin inlong %%%%%%%%%%%%%%%%%%%%%%%%%
\subsection{Correctness in Continuous Graph Processing Design}
\label{subsec:types}

An important design goal of \systemprefix{} and \parlang{} is to establish the correctness of continuous graph processing where expressive features in the design space are unified under one system.

\paragraph{Result Determinism in the Presence of Benign Non-Determinism} 
 
The discussion in \S~\ref{subsec:ingraph}
%and \S~\ref{subsec:par}
has revealed that significant non-determinism are latent in the behavior of in-graph operation streams. 
%: which operation to propagation and realize, which operations will be optimized, and which nodes will be valuated are all non-deterministically determined. 
Recall we also said that these non-determinism are ``benign'': they do not change the ``end-to-end'' behavior of continuous graph processing systems. 

We show \systemprefix{} produces the observably equivalent results as eager processing, i.e., the one-at-a-time semantics we described in \S~\ref{subsec:ingraph}. %We further establish that \systemprefix{} is sound. 
With different forms of non-determinism in the calculus, the proof of this theorem is non-trivial, which we use bisimulation. An interesting insight in defining our baseline eager processing is that it shares exactly the same reduction system as \systemprefix{}, except that the evaluation context of the former must be restricted to enforce a deterministic evaluation order. % to appear. %Building on this insight, we prove our soundness through bisimulation. 

\paragraph{Sound Temporal Locality Optimization} 

The essence of an ``optimization'' is that it should not change the result of computation. The result of observable equivalence above entails sound temporal locality optimization: all rewriting rules of \systemprefix{} for temporal locality optimization --- batching, reordering, fusing, and reusing --- preserve the result of computation. 

%Recall in \S~\ref{subsec:tlo}, we discussed that our goal is not provide an exhaustive list of all possible rewriting rules. To allow future designers to extend our calculus with their specific rewriting rules, we introduce an invariant that captures the general \emph{soundness} of a temporal locality optimization. Any rewriting rule that satisfy this invariant can be introduced to our calculus while preserving the soundness of the streaming graph design itself. Finally, we also show the rewriting rules we have defined are indeed sound.
%\dnote{I don't know if this is something we can show here. It would be good if you already have a lemma like this. We did this in the other abstraction paper. If this is too hard, we don't have to get this far. }
%\pnote{At a minimum, we're indirectly saying that the rewriting rules are sound in the observable equivalence theorem}

% \paragraph{Sound Parallelism}

% Finally, we establish that different forms of parallelism introduced in \parlang{} preserve the observable equivalence of continuous graph processing w.r.t. \systemprefix{}. The development of the meta-theory and its proof takes advantage of one fact: benign non-determinism entails parallelism. In other words, there is a close correspondence between each form of parallelism in \parlang{} and each form of non-determinism in \systemprefix{}. As a result, safe parallelism nearly comes for free. %The close correspondence also shows that our parallelism is principled: it mirrors well with the non-determinism features in the serial semantics. 

}%%%%%%%%%%%%%%%%%%%%%%%%% end inlong %%%%%%%%%%%%%%%%%%%%%%%%%

\subsection{A Type System for Phase Distinction}

The primary goal of \systemprefix{} type system is to enforce \emph{phase distinction} of operation emission: the backend should not emit an operation for processing while processing another operation. To see why this restriction is important, let us start with a counterexample. 

\begin{example}[Backend Operation Emission]
\label{ex:emit}
Consider the following program:
\begin{center}
\scalebox{\scalecode}{
$\begin{array}{l}
\kwlet\ \mathtt{k = \ldots}\ \texttt{// key of interest} \\
\kwlet\ \mathtt{f = \lambda \tuple{\_;\ payload;\ \_} . payload * 2\ \kwin} \\
\kwlet\ \mathtt{g = \lambda \tuple{\_;\ payload;\ \_} . (mapVal\ f\ [k]; payload)}\ \kwin \\
\mathtt{mapVal\ g\ [k] ;} \\
\kwlet\ \mathtt{h = \lambda \tuple{\_;\ payload;\ \_} . payload + 1\ \kwin} \\
\mathtt{mapVal\ h\ [k]} \\
\end{array}$
}
\end{center}
If the operation $\mathtt{mapVal\ f\ [k]}$ inside the body of $\mathtt{g}$ is emitted
\emph{before} the operation $\mathtt{mapVal\ h\ [k]}$ is emitted, 
the node with key $\mathtt{k}$
will have its payload value multiplied by $2$ and then incremented by $1$.
If the order is reversed, the payload value will be incremented by $1$ and then multiplied by $2$. 

\end{example}

The root problem is that the evaluation order between a backend-emitted $\mathtt{mapVal}$ and a frontend-emitted is not always decided upon, a symptom analogous to a race condition. %, and different results may be produced from different evaluation orders. 
Note that this is fundamentally different from the non-determinism executions (with deterministic results) that \systemprefix{} does support. 
Our type system disallows backend operation emission through effect types: for every operation that is emitted from the frontend, we guarantee that its processing does not have the effect of operation emission. 
%We term this property \emph{phase distinction} for operations. %: the effect of evaluating an operation expression \emph{on the frontend} is unrestricted, but its effect \emph{on the backend} must not include any operation emission. 
%On the high level, it is analogous to designing a programming model where race conditions are impossible. 

% \begin{figure}[t]
% \centering
% 
% \begin{tabular}{cc}
% \begin{tikzpicture}
% 
% \basegraph
% 
% \node[op] at ($(eve) + (0,0.3)$) (op2) {\scriptsize$\processor_2$};
% \node[op] at ($(deb) + (0,0.3)$) (op1) {\scriptsize$\processor_1$};
% 
% \begin{pgfonlayer}{background}
%   \node[fit=(deb),ultra thick,draw,loosely dotted,rounded corners=0.6cm,opacity=0.4,fill=blue!5,inner sep=0.1cm] {};
%   \node[fit=(eve),ultra thick,draw,loosely dotted,rounded corners=0.6cm,opacity=0.4,fill=blue!5,inner sep=0.1cm] {};
% \end{pgfonlayer}
% 
% \end{tikzpicture}
% &
% \hspace{1cm}
% \raisebox{2cm}{
% \begin{tikzpicture}[scale=\exgraphscale]
%   \begin{scope}[every node/.style={scale=\exgraphscale,font=\fontsize{12pt}{\baselineskip}\selectfont}]
%     \node[draw,dashed,fill=olive!25,align=left] at (1.9,0.5) (legend) {Legend: \\ \hspace{2em} parallel unit};
%     \node[fill=blue!5,draw,thick,loosely dotted,rounded corners=0.2cm,inner sep=0.2cm] at ($(legend.west) + (0.55,-0.2)$) {$\hspace{1em}$};
%   \end{scope}
% \end{tikzpicture}
% }
% \end{tabular}
% 
% \caption{The Graph in Fig.~\ref{fig:taskparallel} with Parallel Units}
% \label{fig:parallel}
%   
% \end{figure}

\inlong{%%%%%%%%%%%%%%%%%%%%%%%%% begin inlong %%%%%%%%%%%%%%%%%%%%%%%%%

\begin{table}[t]
\centering

\scalebox{\scalemath}{
\begin{tabular}{rll}
\toprule
Notation & Meaning & Definition
%\\\midrule
\\\cmidrule(lr){1-3}
\rowcolor{RC2}
\multicolumn{3}{c}{
sequences
}
\\
% %\rowcolor{RC1}
% $\emptyseq$
% & empty sequence
% &
% \\
%\rowcolor{RC1}
$[\sigma_1, \sigma_2, \ldots, \sigma_m]$
& sequence
&
\\
%\rowcolor{RC1}
$\aseq{\sigma}^m$
& sequence $[\sigma_1, \sigma_2, \ldots, \sigma_m]$ if $m \geq 0$
%or $\emptyseq$ if $m = 0$
\\
%\rowcolor{RC1}
$\aseq{\sigma}$
& sequence $[\sigma_1, \sigma_2, \ldots, \sigma_m]$ for some $m$
\\
%\rowcolor{RC1}
$|\Sigma|$
& length of sequence $\Sigma$
& $m$ if $\Sigma = \aseq{\sigma}^m$
%or $0$ if $S = \emptyseq$
\\
%\rowcolor{RC1}
$\sigma \cons \Sigma$
& sequence with head $\sigma$ and tail sequence $\Sigma$
& $[\sigma, \sigma_1, \ldots, \sigma_m]$ if $\Sigma = \aseq{\sigma}^m$
\\
%\rowcolor{RC1}
$\Sigma \concat \Sigma'$
& concatenation of sequences $\Sigma$ and $\Sigma'$
& $[\sigma_1, \ldots, \sigma_m, \sigma'_1, \ldots, \sigma'_{m'}]$
if $\Sigma = \aseq{\sigma}^m, \Sigma' = \aseq{\sigma'}^{m'}$
\\
% \rowcolor{RC1}
% $s \in S$
% & membership of $s$ in $S$
% & when $S = S' \concat s \cons S'' $
% \\
% \rowcolor{RC1}
% $S \cap S'$
% & $S$ intersected with $S'$
% & $s :: (S'' \cap S')$ where $S = s :: S'' \land s \in S'$ or
% \\ &&
% $S'' \cap S'$ where $S = s :: S'' \land s \notin \S'$ or
% \\ &&
% or $\emptyseq$ when $S = \emptyseq$
% \\
% \rowcolor{RC1}
% $S \subseteq S'$
% & $S$ is a subset of $S'$
% & when $S = S \cap S' = S$
% \\
% % \rowcolor{RC1}
% % $[f(s) \mid s \in S]$
% % & new sequence based on $S$
% % & $[f(s_1), f(s_2), \ldots, f(s_u)]$ where $S = \aseq{s}^m$
% % \\
% \rowcolor{RC1}
% $S \setminus s$
% & subtract $s$ from $S$
% & $s' \cons (S' \setminus s)$ where $S = s' \cons S' \land s' \neq s$ or
% \\ && $S' \setminus s$ where $S = s' \cons S' \land s' = s$ or
% \\ && $\emptyseq$ when $S = \emptyseq$
% \\
\rowcolor{RC2}
\multicolumn{3}{c}{
sets
}
\\
%$\emptyset$
%& empty set
%\\
$\{ \sigma_1, \sigma_2, \ldots, \sigma_m \}$
& set
\\
$\aset{\sigma}^m$
& set $\{ \sigma_1, \sigma_2, \ldots, \sigma_m \}$ if $m \geq 0$
%or $\emptyset$ if $m = 0$
\\
$\aset{\sigma}$
& set $[\sigma_1, \sigma_2, \ldots, \sigma_m]$ for some $m$
\\
$|\Sigma|$
& size of set $\Sigma$
& $m$ if $\Sigma = \aset{\sigma}^m$
%or $0$ if $S = \emptyset$
\\
%$S \cup S'$, $s \in S$, \pnote{etc}
%\\
\rowcolor{RC2}
\multicolumn{3}{c}{
mappings
}
\\
%\rowcolor{RC2}
$\aseq{\sigma \mapsto \sigma'}^m$
or
$\aset{\sigma \mapsto \sigma'}^m$
& mapping (when $\sigma_1$, \ldots, $\sigma_m$ are distinct)
&
\\
%\rowcolor{RC2}
$M(\sigma)$
& mapping lookup
& $\sigma'$ where $\sigma \mapsto \sigma' \in M$
\\
%\rowcolor{RC2}
$\dom(M)$
& mapping domain
& $\aset{\sigma}^m$ where
$M = \aseq{\sigma \mapsto \sigma'}^m$
or
$M = \aset{\sigma \mapsto \sigma'}^m$
\\
%\rowcolor{RC2}
$\ran(M)$
& mapping range
& $\aset{\sigma'}^m$ where
$M = \aseq{\sigma \mapsto \sigma'}^m$
or
$M = \aset{\sigma \mapsto \sigma'}^m$
\\
% \rowcolor{RC2}
% \multicolumn{3}{c}{
% key lists
% }
% \\
% %\rowcolor{RC4}
% $\nkey \ksin \KS{S}$
% & key list key inclusion
% &
% $s \in \aseq{S}$
% %if $\K = \KS{S}$
% \\
% %\rowcolor{RC4}
% $\KS{S} \kscap \KS{S'}$
% & key list intersection
% & $\KS{S \cap S'}$
% \\
% %\rowcolor{RC4}
% $\KS{S} \kssubseteq \KS{S'}$
% & key list containment
% &  $S \subseteq S'$
% \\
% %\rowcolor{RC4}
% $\KS{S} \kssetminus \KS{S'}$
% & key list key removal
% &
% $\KS{S \setminus S'}$
% %if $\K = \KS{S}$
% \\
\rowcolor{RC2}
\multicolumn{3}{c}{
$\lambda$ calculus
}
\\
%\rowcolor{RC3}
$\ep[ \val / \VAR ]$
& variable substitution in expression
&
\\
%\rowcolor{RC3}
$\feq$
& term equivalence
&
\\
%\rowcolor{RC3}
$\mathtt{Id}$
& $\lambda \VAR. \VAR$ 
&
\\
%\rowcolor{RC3}
%$\epsilon(\inode)$
%& node $\inode$ with empty operation store
%& $\tuple{\inode; \emptyseq}$
%\\
%\rowcolor{RC3}
%$\ksin$, $\kssetminus$, $\kscap$, $\kssubseteq$, $\ksrestore$
%%$\kscons, \ksin, \kssetminus, \kscap, \kscup, \kssubseteq$
%& key selector operators
%& Definition~\ref{def:ksops}
%\\
% \rowcolor{RC3}
% $\target(\processor)$
% & targets of operation $\processor$
% & Definition~\ref{def:target}
% \\
% \rowcolor{RC3}
% $\dess \programsummary \tuple{\ostream; \ep; \rstream}$
% & summary update
% & Definition~\ref{def:ssupdate}
% \\
%\rowcolor{RC3}
$(\lambda x . \ep) \circ (\lambda y . \ep')$ % / $\func \mathbin{\circ{\circ}} \func'$
& function composition
& $\lambda z . (\lambda x . \ep)\ ((\lambda y . \ep')\ z)$
\\
% \rowcolor{RC3}
% $\configsummary \config$
% & configuration summary
% & Definition~\ref{def:ourlang:sum}
% \\
% \rowcolor{RC3}
% $\config \myendosim{} \config'$
% & configuration endo-bisimilarity
% & Definition~\ref{def:endobisim}
% \\
%\rowcolor{RC3}
%$\config \parsim{\eagersys}{\lazysys} \config'$
%& configuration $\eagersys$--$\lazysys$ bisimilararity
%& Definition~\ref{def:bisim}
%\\
\bottomrule
\end{tabular}
}

\caption{Notations
% \pnote{can we get rid of $\kssetminus$, $\ksin$, etc.}
% (sequence related notations are shaded in white,
% mapping related notations are shaded in \colorbox{RC2}{light gray},
% target related notations are shaded in \colorbox{RC4}{gray},
% and standard $\lambda$ calculus functions are shaded in \colorbox{RC3}{dark gray})
%\dnote{"when" should be something like $\exists.$ variable 1, vairable 2. } 
%\pnote{which line? for example I don't see how $\exists$ fits for the $S \subseteq S'$ condition}
%\dnote{put eqiv here, also make sure Id font is right; also give names to the key selector definitions} 
}
\label{table:notations}

\end{table}

%\dnote{mental note: may need some name to refer to this N O pair; it's awkward to call it a node, because of O. IT's more or less like a node augumented with a streamlet. I find I talk about this thing a lot but don't know what to call it. Same has to do with the name "neighborhood" which I often use but somewhat sketchy. I start to the word "host node" to refer to the N where an operation has reached, i.e., in that N's associated streamlet. } 

%\dnote{I have made a "raw cut" pass that updated writing has touched all sections, except the theorems section of meta. A ton of loose ends left though. I also left a lot of notes, mostly in the opsem section}

%\pnote{note: metavariable symbol/name $\data$}

}%%%%%%%%%%%%%%%%%%%%%%%%% end inlong %%%%%%%%%%%%%%%%%%%%%%%%%

\section{\systemprefix{} Syntax and Runtime Definitions}
\label{sec:baseline}

In this section, we provide definitions for \systemprefix{}, including %introduce
%basic definitions for constructing a streaming graph runtime.
%
%We define the frontend syntax,
%provide a transitional guide from programmer syntax to formal syntax,
%and define graph processing operations
abstract syntax in \S~\ref{sec:syntax}
%backend graph representation in \S~\ref{sec:abstractconcrete},
and runtime configuration in \S~\ref{sec:opandresultstreams}. 
\inlong{%%%%%%%%%%%%%%%%%%%%%%%%% begin inlong %%%%%%%%%%%%%%%%%%%%%%%%%
Common notations and definitions 
used throughout the paper are summarized in Table~\ref{table:notations}.
}%%%%%%%%%%%%%%%%%%%%%%%%% end inlong %%%%%%%%%%%%%%%%%%%%%%%%%

\inshort{
\paragraph{Notations}
We summarize 3 common structures used in this paper: sequence, set, and mapping. We use notation $[\sigma_1, \sigma_2, \ldots, \sigma_m]$ to represent a sequence of $\sigma_1$, $\dots$, $\sigma_m$ (in that order) for some $m \geq 0$; we shorthand it as $\aseq{\sigma}^m$, or $\aseq{\sigma}$ when its length does not matter. We further call $\sigma_1$ as the \emph{head element} and $\sigma_m$ as the \emph{last element}. When $m = 0$, we further represent an empty sequence as $\emptyseq$. Binary operator $\sigma \cons \Sigma$ prepend $\sigma$ to sequence $\Sigma$ as the head, and binary operator $\Sigma \concat \Sigma'$ concatenate $\Sigma$ and $\Sigma'$ together. We elide their definitions here. We use notation $\{ \sigma_1, \sigma_2, \ldots, \sigma_m \}$ to represent a set with elements $\sigma_1$, $\dots$, $\sigma_m$ for some $m \geq 0$; we shorthand it as $\aset{\sigma}^m$, or $\aset{\sigma}$ when its length does not matter. When $m = 0$, we further represent an empty set as $\emptyset$. Common set operators $\in$, $\subseteq$, and $\cap$ %, and $\setminus$ 
apply. %When order does not matter for a sequence, we apply the same 4 set operator over sequences \dnote{check}. 
We overload operator $|\Sigma|$ to compute the length of the sequence $\Sigma$, or the size of set $\Sigma$. 

When a sequence takes the form of $\aseq{\sigma \mapsto \sigma'}^m$ or when a set takes the form of $\aset{\sigma \mapsto \sigma'}^m$, we call it a \emph{mapping} when $\sigma_1$, $\dots$, $\sigma_m$ are distinct. Given $M$ as the aforementioned mapping, we further define $M(\sigma_i)$ as $\sigma_i'$ for some $1 \leq i \leq m$; $\dom(M)$ as $\aset{\sigma}^m$; and $\ran(M)$ as $\aset{\sigma'}^m$.

We use some common functions of $\lambda$ calculus, whose definitions we omit. We use $\ep[ \val / \VAR ]$ to represent substitution of name $\VAR$ with value $\val$ for expression $\ep$.  We use $\feq$ for term equivalence. 
We use $\mathtt{Id}$ to represent identity function, and $\circ$ for function composition.
}%inshort

\subsection{Syntax}
\label{sec:syntax}

\begin{figure}[t]
\centering

\begin{tabular}{r|l}
\scalebox{\scalemath}{
$\begin{array}{@{}l@{\ \ }l@{\ \ }lr}

\multicolumn{4}{@{}l}{\textbf{Expressions, Operations, Values}}\\

\ep
& \defassign &
\val
\mid \ep\ \ep
\mid \VAR
\mid \mathbf{fix}\ \ep
& \textit{expression}
\\ & \mid &
\seqK
\mid \ep \setadd \ep
\mid \ep \setsubtract \ep
\\ & \mid &
\inode
\mid \neselect{\ep}
\\ & \mid &
\toserver\ \processor
\mid \fromserver\ \ep
\\

\processor
& \defassign &
\addn\ \ep
\mid \map\ \ep\ \ep
\mid \fold\ \ep\ \ep\ \ep
&
\textit{operation}
\\

\val
& \defassign &
\func
\mid \nkey
\mid \integer
\mid \K
\mid \eagerinode
\mid \fvs
% \mid \fvs
% \mid \bot
&
\textit{value}
\\

\func
& \defassign &
\lambda \VAR : \tau . \ep
&
\textit{function}
\\

\fvs
& & %\in &
%\mathbb{LABEL}
&
\textit{future value/label}
\\

\tau
& &
&
\textit{type (see \S~\ref{sec:safety})}
\\

%\hline
\end{array}$
}
&
\raisebox{0.125cm}{
\scalebox{\scalemath}{
$\begin{array}{@{}l@{\ \ }l@{\ \ }lr}
\multicolumn{4}{@{}l}{\textbf{Keys, Nodes, Integers, Names}}\\

\nkey
& &
&
\textit{key}
\\

\seqK
& \defassign &
\KS{\aseq{\ep}}
& \textit{key list}
\\

\K
& \defassign &
\KS{\aseq{\nkey}}
& \textit{key list value}
\\

\inode
& \defassign &
\ND{\ep; \ep; \ep}
& \textit{node}
\\

\eagerinode
& \defassign &
\ND{\nkey; \integer; \K}
& \textit{node value}
\\

\integer
& &
&
\textit{integer}
\\

\multicolumn{3}{@{}l}{\VAR, \VARY, \VARZ}
&
\textit{name}
\\

\neI
& \in &
\{ 1, 2, 3 \}
&
\textit{projection index} 
\\
\end{array}$
}
}
\end{tabular}

\caption{Abstract Syntax}
\label{fig:clientdefs}

\end{figure}
%%%%

\begin{table}[t]
\centering

\scalebox{\scalemath}{
\begin{tabular}{lll}
\toprule
Programmer Syntax
&
Formal Syntax
&
$\mathtt{f}$
\\
\midrule
$\tuple{\ep; \ep'; \ep''}$
& $\ND{\ep; \ep'; \ep''}$
\\
$[\ep_1, \ldots, \ep_n]$
& $\KS{[\ep_1, \ldots, \ep_n]}$
\\
$\mathtt{addRelationship}\ \ep\ \ep'$
&
$\map\ \mathtt{f}\ \KS{[\ep]}$
&
$\lambda x . \ND{\neone{x}; \netwo{x}; \nethree{x} \setadd \KS{[\ep']}}$
\\
$\mathtt{deleteRelationship}\ \ep\ \ep'$
&
$\map\ \mathtt{f}\ \KS{[\ep]}$
&
$\lambda x . \ND{\neone{x}; \netwo{x}; \nethree{x} \setsubtract \KS{[\ep']}}$
\\
$\mathtt{updatePayload}\ \ep\ \ep'$
&
$\map\ \mathtt{f}\ \KS{[\ep]}$
&
$
\lambda x . \ND{\neone{x}; \netwo{\ep'}; \nethree{x}}
$
\\
$\mathtt{queryNode}\ \ep$
&
$\fold\ \mathtt{f}\ \ND{\_; 0; \KS{\emptyseq}}\ \KS{[\ep]}$
&
$\lambda x . \lambda y . x$
\\
% $\mathtt{queryPayload}\ \ep$
% &
% $\fold\ \mathtt{f}\ \tuple{\nkey\textrm{\pnote{?}}; 0; \KS{\emptyseq}}\ \KS{[\ep]}$
% &
% $\lambda x . \lambda y . x$
% \\
% $\mathtt{queryRelationships}\ \ep$
% &
% $\fold\ \mathtt{f}\ \tuple{\nkey\textrm{\pnote{?}}; 0; \KS{\emptyseq}}\ \KS{[\ep]}$
% &
% $\lambda x . \lambda y . x$
% \\
$\mathtt{mapVal}\ \ep\ \ep'$
&
$\map\ \mathtt{f}\ \ep'$
&
$\lambda x . \ND{\neone{x}; \ep\ x; \nethree{x}}$
\\
$\mathtt{foldVal}\ \ep\ \ep'\ \ep''$
&
$\fold\ \mathtt{f}\ \ND{\_; \ep'; \KS{\emptyseq}}\ \ep''$
&
$\lambda x . \lambda y . \ND{\neone{y}; \ep\ x\ \netwo{y}; \nethree{y}}$
\\
% $\mathtt{foldSet\ \ep\ \ep'\ \ep''}$
% &
% $\fold\ \mathtt{f}\ \ND{\_; 0; \ep'}\ \ep''$
% &
% $\lambda x . \lambda y . \ND{\neone{y}; \netwo{y}; \ep\ x\ \nethree{y}}$
% \\
%$\mathtt{key}\ \ep$
%&
%$\neone{\ep}$
%\\
%$\mathtt{payload}\ \ep$
%&
%$\netwo{\ep}$
%\\
%$\mathtt{adjlist}\ \ep$
%&
%$\nethree{\ep}$
%\\
%$\mathtt{queryRelationships}\ \ep$
%&
%$\fold\ \mathtt{f}\ \tuple{?; 0; \KS{\emptyseq}}\ \KS{[\ep]}$
%&
%$\lambda x . \lambda y . x$
%\\
\bottomrule
\end{tabular}
}

\caption{Graph Operations Encodings
%\dnote{I flipped the two subsections, so try to move the key list and node encoding first in the table}
%\dnote{In the "updatePayload" encoding, that middle guy with a superscript 2. Is that correct? }
%\pnote{yes, since the argument is a node, not a payload}
%\dnote{really? then I don't understand Line 16 of the example. I thought nb is like an integer} 
%\pnote{no, it's a node (see queryNode)}
%\dnote{I see. Can we do it the other way where quernode returns the payload only? I just think it is confusing that a querynode returns a node that later is used for other guys to upload their payload with. I understand this is just a matter of consistency, but I would rather querynode/updatepayload to be more intuitive to end users. }
%\pnote{i think we tried that before but realized it removed laziness from the examples}
%\dnote{For querynode, don't we need to get a superscript 2 for whatever is returned? Also, is it x or y as the body? }
%\pnote{I think that would remove laziness. and I believe it's an x (the current node is passed as the first argument to fold, as \plaintitle{Fold} is written now)}
%\dnote{also, there are two cases where emptyset is used in KL; it should be empty sequence}
%\dnote{for "foldVal" I thought the initial value is a node anyways, so I don't understand why the second argument in the "fold" encoding has to wrap it up again.}
%\pnote{how we use foldVal in the PG example, we do not pass a node as the initial argument.}
} 
\label{fig:encodings}

\end{table}

Fig.~\ref{fig:clientdefs} defines the abstract syntax of \systemprefix{}. It consists of conventional 
$\lambda$ calculus features, such as name, abstraction, application, and the fixpoint expression. Features that appeared in the earlier examples but can be encoded by $\lambda$ calculus, 
including $\mathtt{if-then-else}$, list comprehension, $\mathtt{let-in}$, the $;$ expression, $\mathtt{foreach}$, and arithmetic, are omitted. We choose to include integers explicitly, so that node payload values
can be intuitively represented. 
The additional expressions come in two categories.

\paragraph{Expressions for Graph Structural Support}

Both the key list and the node are first-class citizens in our calculus. In the programmer syntax, the former is represented as a sequence and the latter as a triple. To differentiate programming abstractions from meta-level structures, we associate the key list with an explicit constructor \textbf{KL} and the node with constructor \textbf{N} in the formal syntax, as shown in Table~\ref{fig:encodings}.   %\dnote{here please make sure we actually use sequence syntax instead of set syntax for the bare key lists in the examples} 

Key lists in our calculus play two roles: defining the (ordered) adjacency list of a node, providing as argument for selective mapping and folding. 
%For example, at Line 15 in Fig.~\ref{fig:pagerank}, the last argument of operation $\mathtt{mapVal}$ --- $\{\mathtt{nk}\}$ in sugared syntax --- is formally written as $\KS{[\mathtt{nk}]}$. 
The $\setadd$ and $\setsubtract$ expressions are binary operators over key lists for their concatenation and subtraction respectively.
To assist key list subtraction,
we define operator $\aseq{\nkey} \kssetminus \aseq{\nkey'}$ as
identical to $\aseq{\nkey}$ except that every element that appears in $\aseq{\nkey'}$ is removed.
%\dnote{Since we have introduced "target" operator here. We might as well rename "key selector" to "target" for all.} 
%%%%%%   We define 
%%%%%%   $\nkey \ksin \KS{S}$
%%%%%%   as
%%%%%%   $s \in S$;
%%%%%%   $\KS{S} \kscap \KS{S'}$
%%%%%%   as $\KS{S \cap S'}$; 
%%%%%%   $\KS{S} \kssubseteq \KS{S'}$
%%%%%%   as $S \subseteq S'$; 
%%%%%%   $\KS{S} \kssetminus \KS{S'}$
%%%%%%   as $\KS{S \setminus S'}$. \dnote{update}
%%%%%%   %if $\K = \KS{S}$
%Nodes are also first-class citizens either in the expression form $\inode$ or the value form $\eagerinode$. 
As we have seen in the \textsc{CorePR} example, each graph node ($\inode$) is a triple: a key, a payload expression, and an adjacency list expression. %When the second and third components are in the value form, the payload expression is an integer, and the adjacency expression is a key list %(\S\ref{sec:syntax}).
%that indicates what nodes the current node forms an out-edge with.
Projection expression $\neselect{\ep}$ computes the $\neI$-th component of the triple.
%In our calculus, both $\map$ and $\fold$ operations are defined over first-class nodes. For example, in Table~\ref{fig:encodings},
%the encoding of $\mathtt{mapVal}\ \ep\ \ep'$ says that the $\map$ operation function takes a node and returns a node, and the $\ep$ function operates on the second component of the node.
\inlong{%%%%%%%%%%%%%%%%%%%%%%%%% begin inlong %%%%%%%%%%%%%%%%%%%%%%%%%
%The encoding of $\mathtt{foldVal\ \ep\ \ep'\ \ep''}$ says that the $\fold$ operation function uses a node as an accumulating value, and the function $\ep$ operates on the second component of the accumulating node.
}%%%%%%%%%%%%%%%%%%%%%%%%% end inlong %%%%%%%%%%%%%%%%%%%%%%%%%

\inlong{%%%%%%%%%%%%%%%%%%%%%%%%% begin inlong %%%%%%%%%%%%%%%%%%%%%%%%%
Indeed, key lists and nodes can be easily encoded by $\lambda$ calculus through sets and tuples. We choose to represent them explicitly for two reasons. First, In a typed calculus like hours, an explicit representation of these entities has the benefit of ensuring the type safety of their use. Second, by allowing them to be first-class citizens, we highlight the central role of key lists (graph edges and operation targets) and graph nodes play in continuous graph processing. 
}%%%%%%%%%%%%%%%%%%%%%%%%% end inlong %%%%%%%%%%%%%%%%%%%%%%%%%

\paragraph{Expressions for Graph Operation Lifecycle Support}

\begin{wrapfigure}{R}{0.33\textwidth}
%\begin{figure}
\centering
\scalebox{\scalecode}{
\begin{minipage}{0.5\textwidth}
\begin{algorithmic}[ht]\small
\State \texttt{// dynamic queries and updates}
\State $\kwlet\ \mathtt{nb} = \mathtt{\color{blue}\toserver\ \mathtt{queryNode}\ \fromserver\ \mathtt{b}}\ \kwin$
\State $\mathtt{\color{blue}\toserver\ \mathtt{updatePayload}\ \fromserver\ \mathtt{a}\ \fromserver\ \mathtt{nb}};$
\State $\kwlet\ \mathtt{nb2} = \mathtt{\color{blue}\toserver\ \mathtt{queryNode}\ \fromserver\ \mathtt{b}}\ \kwin$
\State $\kwlet\ \mathtt{f} = \mathtt{\color{blue}\toserver\ \mathtt{add}\ \fromserver\ \mathtt{fred}}\ \kwin$
\State $\mathtt{\color{blue}\toserver\ \mathtt{addRelationship}\ \fromserver\ \mathtt{b}\ \fromserver\ \mathtt{f}};$
\State $\mathtt{\color{blue}\toserver\ \mathtt{deleteRelationship}\ \fromserver\ \mathtt{b}\ \fromserver\ \mathtt{f}};$
\State $\ldots$
\end{algorithmic}
\end{minipage}
}
\caption{From Programmer Syntax to Formal Syntax: $\toserver$ and $\fromserver$ (Fig.~\ref{fig:serverclient} Lines 14-21)}
\label{fig:explicitarrows}
\end{wrapfigure}
%\end{figure}
Two new expressions handle the operation stream at the frontend: operation 
emission ($\toserver\ \processor$) and result claim ($\fromserver\ \ep$). To highlight the asynchronous
nature of operation processing, each program point of operation emission (or result claim) in the 
programmer syntax is annotated with a $\toserver$ symbol (or $\fromserver$ symbol) explicitly.
For example,
Fig.~\ref{fig:explicitarrows}
%the pseudocode on the right
shows how the Lines 14-21 of Fig.~\ref{fig:serverclient} can be explicitly
annotated with $\toserver$ and $\fromserver$.

\systemprefix{} supports 3 core operations: $\addn$, $\map$, and $\fold$. The first operation has been used in the \textsc{CoreSocial} and \textsc{CorePR} examples. The second and third operations are similar to $\mathtt{mapVal}$ and $\mathtt{foldVal}$ in \textsc{CorePR}, except that mapping function in $\map$ returns a node, and the folding function $\fold$ is a binary function over nodes.
The choice of these 3 operations is a balance between simplicity and expressiveness, with several considerations. First, extensive programming experience with MapReduce has shown that the $\map-\fold$ pair is capable of programming a large number of graph analytics algorithms~\cite{kang2009pegasus,lin2010design}. %Pagerank is only one of the many examples. 
Second, $\map$ and $\fold$ can encode graph database operations, with the encoding of those appearing in the \textsc{CoreSocial} example shown in Table~\ref{fig:encodings}. Third, $\addn$ is useful for supporting dynamic graphs. 
%We do not explicitly support node deletion because nodes are generally garbage collected in real world systems, where marking can be supported through $\map$.
%Here we choose not to support node deletion because this operation is nearly universally supported in real-world graph processing systems through a process similar to garbage collection: the to-be-deleted node is marked with a special tag, and physically removed later in an independent and often offline process. Marking a node can be encoded by updating the special payload of the node, through $\map$.  
Overall, the choice of graph-specific operations may impact on the specific rules for TLO, but is orthogonal to the development of IOP in our calculus. 

For operations, we introduce
a convenience function $\target$ that computes the keys of nodes where the operation is intended for realization: %, which we call the \emph{target} of the operation:
\begin{definition}[Operation Target]
\label{def:target}
The function $\target(\processor)$
computes the \emph{target} of the operation $\processor$, defined as $\aset{\nkey}$ if $\processor = \map\ \func\ \KS{\aseq{\nkey}}$ or $ \processor = \fold\ \func\ \ep\ \KS{\aseq{\nkey}}$. 
The operator is undefined for $\addn$.
%\dnote{please double check. I turned $\processor$ to $\strictprocessor$}
%\pnote{reverted for now, until we figure out $\processor$ vs $\strictprocessor$}
\end{definition}

\paragraph{Values}

The values of our languages are functions, node keys, node payloads (integers), key list values, node values, and futures. 
A future value $\fvs$ is generated when an operation is emitted (\S~\ref{sec:coresocial}), and as we shall see soon, it also serves as the unique label for identifying the operation and its result in backend graph processing.
Except for futures, all forms of values are also programs, including keys. To be consistent with real-world data processing, we allow programmers to name a key in their program. 
\inlong{%%%%%%%%%%%%%%%%%%%%%%%%% begin inlong %%%%%%%%%%%%%%%%%%%%%%%%%
The implication is that \systemprefix{} allows a programmer to explicitly name a key in their program, in addition to allowing them to use keys returned from the $\addn$ operation (such as Line 4 in Fig.~\ref{fig:serverclient}). Indeed, it would be tempting to disallow the former, because we would end up having a more ``elegant'' calculus: in that case, a key value would be analogous to an object reference in object-oriented languages, and we could design a type system so that ``key not found'' would never occur, just as Java's type system would eliminate ``object not found'' errors. This artificial elegance however results from an oversimplification: in real-world data processing, keys are routinely named and used by users and ``key not found'' does happen in all practical databases.
}%%%%%%%%%%%%%%%%%%%%%%%%% end inlong %%%%%%%%%%%%%%%%%%%%%%%%%

\begin{figure}[t]
\centering

\begin{minipage}[t]{0.45\textwidth}
\scalebox{\scalemath}{
$\begin{array}{@{}l@{\ \ }l@{\ \ }lr}

%\multicolumn{4}{@{}l}{\textbf{Runtime}}\\

\config
& \defassign &
\tuple{\data; \ostream; \rstream; \ep}
&
\textit{configuration}
\\

\data
& \defassign &
\aseq{\station}
& \textit{backend}
\\

\station
& \defassign &
\tuple{\inode; \ostream}
& \textit{station}
\\

% \hline
% 
% \multicolumn{4}{@{}l}{\textbf{Backend}}\\

\ostream
& \defassign &
\aseq{\batchgroup}
&
\textit{operation stream/streamlet}
\\

\batchgroup
& \defassign &
\aseq{\fvs \mapsto \processor}
&
\textit{stream unit}
\\

\rstream
& \defassign &
\aset{\fvs \xmapsto{\K} \val}
&
\textit{result store}
\\

% \rval
% & \defassign &
% \val \mid \fvs
% &
% \textit{runtime value}
% \\

% \strictprocessor
% & \defassign &
% \addn\ \integer
% \mid \map\ \func\ \K
% \mid \fold\ \func\ \val\ \K
% &
% \textit{emittable operation}
% \\

\end{array}$
}

\caption{\systemprefix{} Runtime Definitions
%\pnote{do we have to extend $\ep$ here to include $\fvs$/$\rval$?}
%\pnote{use $\station$ in rest of paper where appropriate}
}
\label{fig:pgraph:grammar}
\end{minipage} \hfill\vline\hfill
%\hspace{1em}
%\vskip -10ex
\begin{minipage}[t]{0.45\textwidth}
\raisebox{-0.25cm}{
\scalebox{\scalecode}{
\begin{minipage}[t]{\linewidth}
\vskip -15ex
\[
\begin{array}{l}
[\tuple{\ND{\nkey_{\mathtt{eve}}; \integer_\mathtt{eve}; \KS{[\nkey_{\mathtt{bob}}, \nkey_{\mathtt{amy}}]}}; [[\fvs_2 \mapsto \processor_2]},
\\ \tuple{\ND{\nkey_{\mathtt{deb}}; \integer_\mathtt{deb}; \KS{[\nkey_{\mathtt{cam}}]}}; [[\fvs_1 \mapsto \processor_1]]},
\\ \tuple{\ND{\nkey_{\mathtt{cam}}; \integer_\mathtt{cam}; \KS{[\nkey_{\mathtt{bob}}]}}; \emptyseq},
\\ \tuple{\ND{\nkey_{\mathtt{bob}}; \integer_\mathtt{bob}; \KS{\emptyseq}}; \emptyseq},
\\ \tuple{\ND{\nkey_{\mathtt{amy}}; \integer_\mathtt{amy}; \KS{[\nkey_\mathtt{eve}]}}; \emptyseq}]
\end{array}
\]
\end{minipage}
}
}
%\vskip -5ex
\caption{A Backend Example of Fig.~\ref{fig:datacentriclaziness}(d) (
$\nkey_\mathtt{amy}$,
$\nkey_\mathtt{bob}$,
$\nkey_\mathtt{cam}$,
$\nkey_\mathtt{deb}$,
$\nkey_\mathtt{eve}$
are the generated keys corresponding to the
nodes with payload values
$\integer_\mathtt{amy}$,
$\integer_\mathtt{bob}$,
$\integer_\mathtt{cam}$,
$\integer_\mathtt{deb}$,
$\integer_\mathtt{eve}$
and
$\fvs_1$, $\fvs_2$
are generated labels corresponding to the operations
$\processor_1$, $\processor_2$)}
\label{ex:repr} 
\end{minipage} 
\end{figure}

\inlong{%%%%%%%%%%%%%%%%%%%%%%%%% begin inlong %%%%%%%%%%%%%%%%%%%%%%%%%

\begin{figure}[t]
\centering

%\begin{tabular}{cc}
\scalebox{\scalegraph}{
\begin{tikzpicture}

% server/client
\node[] at (0, 1) (server) {\hspace{-0.9cm}Backend ($\data$)};
\node[] at (4.5, 1) (client) {\hspace{1.1cm}Frontend ($\ep$)};

% graph
\node[nnode] at ($(server) + (0.5,-1.55)$) (eve) {$\mathtt{eve}$};
\node[nnode] at ($(eve) + (-2.0,0)$) (deb) {$\mathtt{deb}$};
\node[nnode] at ($(deb) + (1.0,-1.0)$) (cam) {$\mathtt{cam}$};
\node[nnode] at ($(cam) + (-1.0,-1.0)$) (bob) {$\mathtt{bob}$};
\node[nnode] at ($(bob) + (2.0,0)$) (amy) {$\mathtt{amy}$};

\draw[ledge, out=135, in=135, looseness=2] (eve) to (bob);
\draw[ledge, bend right=20] (eve) to (amy);
\draw[ledge, bend left] (deb) to (cam);
\draw[ledge, bend left] (cam) to (bob);
\draw[ledge, bend right=20] (amy) to (eve);
    
% client box
\begin{pgfonlayer}{background}
\node[draw, fill=black!75, inner sep=4pt, below left=0.3cm and -0.6cm of client, inner ysep=2.0cm, inner xsep=1.5cm, anchor=north west] (clientbox) {};
\end{pgfonlayer}
\node[text=white, anchor=north west] (t1) at (clientbox.north west) {\footnotesize $\kwlet\ \mathtt{x = \toserver\ \processor_1\ \kwin}$};
\node[text=white] (t3) at ($(t1.south) + (0,-0.05)$) {\footnotesize $\kwlet\ \mathtt{y = \toserver\ \processor_2\ \kwin}$};
\coordinate (t2) at ($(t3) + (-0.7,0)$);
\coordinate (y) at ([shift={(-0.15,-0.25)}]t2);
\draw[-,draw=white,decorate,decoration={snake,amplitude=0.4mm}] (y) -- (y -| clientbox.east);
\coordinate (y) at ([shift={(-0.15,-0.5)}]t2);
\draw[-,draw=white,decorate,decoration={snake,amplitude=0.4mm}] (y) -- (y -| clientbox.east);
\coordinate (y) at ([shift={(-0.15,-0.75)}]t2);
\draw[-,draw=white,decorate,decoration={snake,amplitude=0.4mm}] (y) -- (y -| clientbox.east);
\coordinate (y) at ([shift={(-0.15,-1.0)}]t2);
\draw[-,draw=white,decorate,decoration={snake,amplitude=0.4mm}] (y) -- (y -| clientbox.east);
\coordinate (y) at ([shift={(-0.15,-1.25)}]t2);
\draw[-,draw=white,decorate,decoration={snake,amplitude=0.4mm}] (y) -- (y -| clientbox.east);
\coordinate (y) at ([shift={(-0.15,-1.5)}]t2);
\draw[-,draw=white,decorate,decoration={snake,amplitude=0.4mm}] (y) -- (y -| clientbox.east);
\coordinate (y) at ([shift={(-0.15,-1.75)}]t2);
\draw[-,draw=white,decorate,decoration={snake,amplitude=0.4mm}] (y) -- (y -| clientbox.east);
\coordinate (y) at ([shift={(-0.15,-2.0)}]t2);
\draw[-,draw=white,decorate,decoration={snake,amplitude=0.4mm}] (y) -- (y -| clientbox.east);
\coordinate (y) at ([shift={(-0.15,-2.25)}]t2);
\draw[-,draw=white,decorate,decoration={snake,amplitude=0.4mm}] (y) -- (y -| clientbox.east);
\coordinate (y) at ([shift={(-0.15,-2.50)}]t2);
\draw[-,draw=white,decorate,decoration={snake,amplitude=0.4mm}] (y) -- (y -| clientbox.east);
\coordinate (y) at ([shift={(-0.15,-2.75)}]t2);
\draw[-,draw=white,decorate,decoration={snake,amplitude=0.4mm}] (y) -- (y -| clientbox.east);

\coordinate (y) at ([shift={(-0.25,-2.75)}]t2);
\node[text=white,anchor=north west] (t4) at (y) {\footnotesize $\mathtt{\fromserver\ x;}$};
\node[text=white] (t5) at ($(t4.south) + (0,-0.05)$) {\!\footnotesize $\mathtt{\fromserver\ y}$};

\node[fill=black!75,anchor=east, inner xsep=0.06cm, inner ysep=2.0cm] at (clientbox.east) {};

% graph box
\node[draw,fit=(eve)(deb)(bob)(amy),inner sep=0.6cm] (graphbox) {};

% stream arrows
\begin{pgfonlayer}{background}
\coordinate (evex) at ($(eve) + (0,0.3)$);
\coordinate (amyx) at ($(amy) + (0,0.3)$);
\draw[sarrow, text opacity=1, rounded corners] (clientbox.west |- evex) to node {\hspace{-1.9cm}$\processor_1, \ldots, \processor_i$} (graphbox.east |- evex) to ($(eve) + (0,0.3)$) to ($(deb) + (0,0.3)$) to ($(cam) + (0,0.3)$) to ($(bob) + (0,0.3)$) to ($(amy) + (0,0.3)$) to (graphbox.east |- amyx) to node {\hspace{-1.9cm}$\val_1, \ldots, \val_j$} (clientbox.west |- amyx);
\end{pgfonlayer}

\node[] at ($(eve) + (1.25,0.8)$) {Operation Stream ($\ostream$)};
\node[] at ($(amy) + (1.5,0.8)$) {Result Store ($\rstream$)};

\end{tikzpicture}
}

\caption{A Continuous Graph Processing Runtime}
\label{fig:clientserver:indepth}

\end{figure}
}%%%%%%%%%%%%%%%%%%%%%%%%% end inlong %%%%%%%%%%%%%%%%%%%%%%%%%

\subsection{The Structure of the Runtime}
\label{sec:opandresultstreams}

%\dnote{polishing pass skipped this section; too many hanging bits to recalibrate}

Fig.~\ref{fig:pgraph:grammar} provides core definitions of the program runtime. A configuration $\config$ consists of 4 components: the backend $\data$, the frontend expression $\ep$, and two structures that bridge them: the (top-level) \emph{operation stream} $\ostream$ 
and the \emph{result store} $\rstream$.
%\pnote{We treat $\rstream$ as a mapping set and use the set union and mapping lookup notations on it.}
%\pnote{A little awkward since we define mappings as sequences, not sets.}
\inlong{%%%%%%%%%%%%%%%%%%%%%%%%% begin inlong %%%%%%%%%%%%%%%%%%%%%%%%%
Fig.~\ref{fig:clientserver:indepth} is a more refined illustration of Fig.~\ref{fig:clientserver:simple}. 
}%%%%%%%%%%%%%%%%%%%%%%%%% end inlong %%%%%%%%%%%%%%%%%%%%%%%%%

The backend is the runtime graph, represented as a sequence of \emph{stations}, each of which consists of a graph node ($\inode$) and the operations ($\ostream$) that have so far propagated to that node. We also call the latter as a \emph{streamlet}.  
%The reader may note that the operation stream and the result stream defined here share the same structure as the $\ostream$ store and $\rstream$ store carried within graph nodes, respectively. This correspondence is not a coincidence. Indeed, one may think of each $\ostream$ store or each $\rstream$ store --- within each graph node --- as forming a stream of its own. 
In other words, we represent the in-graph operation stream as the combination of per-station streamlets. An example of the backend in its formal form can be found in Fig.~\ref{ex:repr}. 
%This view further unifies in-graph operation stream with the top-level operation stream, in that the latter represents the operations before they become the streamlet of the first node. 
This representation reflects \emph{fine-grained} nature of our support for incremental processing: the operation can propagate to and be deferred at \emph{any node}. Client calculi to \systemprefix{} can further restrict this most general treatment, e.g., a more implementation-oriented choice where nodes form partitions and streamlets can only be associated with (the first node of) partitions.

Our sequence-based representation of the runtime graph is aligned with existing graph processing (database or analytics) systems, where nodes are generally stored as a sequence, edges are explicitly represented as adjacency information within each node, and graph traversal follows the linear order of nodes (a process often called \emph{scanning} in existing literature). By embracing this representation, our calculus still applies to other more logical representations, e.g., an object graph on the heap, because (depth-first or breadth-first) traversal algorithms can place graph nodes into a sequence in the traversal order. For example, the graph in Fig.~\ref{fig:clientserver:simple} could indeed be implemented as a pointer-based graph, but its traversal order still implies a sequence-based representation, following \texttt{eve}, \texttt{deb}, \texttt{cam}, \texttt{bob}, \texttt{amy}.

We formally represent an operation stream/streamlet as a sequence of \emph{stream units} ($\batchgroup$), each of which is a sequence of operations. This 2-dimensional representation --- instead of a 1-dimensional one --- results from batching (\S~\ref{subsec:ingraph}), so that each stream unit can be viewed as a ``batch.''
%When an operation stream is initially produced by the frontend, it is a sequence of singleton sequences, as shown in Fig.~\ref{fig:clientserver:indepth}. 
Observe that we also associate each operation with a unique label ($\fvs$), and 
%associated with each operation is unique, and the association does not change while the operation flows though the backend, and 
the result of processing the operation is associated with the same label. The label is not only a convenience for correlating the operation and its result, but also serves as the future value of the operation as we explained earlier.
%in asynchronous operation processing (\S~\ref{sec:coresocial}). 

%Syntactically, we represent the label-operation association (the second dimension of sequence in the operation stream) as a mapping, so that common operators on mapping in Table~\ref{table:notations} apply. 

Each element in the result store takes the form of $\fvs \xmapsto{\K} \val$, associating result value $\val$ with label $\fvs$. The additional $\K$ is called a \emph{residual target}. 
%Recall that each operation is intended for nodes whose keys are specified in its target key list. 
If any key in the target key list of an operation cannot be found during processing, it will be kept as the residual target in the result store.  

%  in the operation stream and results in the result stream are \emph{labeled}. Observe that in Fig.~\ref{fig:clientserver:indepth}
% \emph{i.e.}, each operation $\processor$ is associated with a distinct label $\fvs \in \mathbb{LABEL}$.
% As we discussed in \S~\ref{sec:background},
% labels provide a simple name binding mechanism for dependent operations.
% %
% In our calculus, the binding between the label
% and its corresponding operation
% does not change.
%

The following definitions highlight the different access patterns of the operation stream and the result store: whereas order does not matter for the latter, it clearly matters for the former  (recall \S~\ref{subsec:ingraph}): 

\begin{definition}[Operation Stream Addition and Result Store Addition]
\label{def:osa}
The $\writeoperation$ operator appends a stream unit to the configuration, the $\writeoptwo$ operator appends a stream unit to a non-empty backend, and the $\writeresult$ operator adds results to the configuration: 
%, and the $\readoperation$ operator removes a stream item from a stream.
%The operator is overloaded for the top-level operation stream and the in-graph operation stream.
%
%\begin{center}
%\scalebox{\scalemath}{
%$
\[
\begin{array}{r@{\ }c@{\ }l}
\tuple{\data; \ostream; \rstream; \ep} \writeoperation \batchgroup & \df &
  \tuple{\data; \ostream \concat [\batchgroup]; \rstream; \ep}
  \\
\tuple{\inode; \ostream} \cons \data \writeoptwo \batchgroup & \df &
  \tuple{\inode; \ostream \concat [\batchgroup]} \cons \data
  \\
  \tuple{\data; \ostream; \rstream; \ep} \writeresult \rstream' & \df &
  \tuple{\data; \ostream; \rstream' \cup \rstream; \ep}
  \\

% \readoperation \tuple{\data; \batchgroup \cons \ostream; \rstream; \ep} & \df &
%   \tuple{\data; \ostream; \rstream; \ep}, \batchgroup
%   \\
% \readoperation [\tuple{\inode; \batchgroup \cons \ostream}] & \df &
%   [\tuple{\inode; \ostream}], \batchgroup
%   \\
%\fundef{reset}(\batchgroup) & \df & [ \fvs \mapsto \fundef{r}(\processor) \mid \fvs \mapsto \processor \in \batchgroup ]\\
%\fundef{r}(\processor) & \df & \begin{cases}
%\map\ \ep\ \ksrestore \K & \fif \processor = \map\ \ep\ \K \arcr
%\fold\ \ep\ \ep'\ \ksrestore \K & \fif \processor = \fold\ \ep\ \ep'\ \K
%\end{cases} \\
%\ksrestore \KS{\aseq{\nkey}} & \df & \KS{\aseq{\nkey}} \\
%\ksrestore \KSstar & \df & \KSstar \\
%\ksrestore \KShollowstar & \df & \KSstar \\
\end{array}
\]
%$
%}
%\end{center}
%
%Note that one of the cases is partial. A stream unit cannot be appended to an empty backend.
%\dnote{need better this.}
\end{definition}

The head element in the operation stream represents the earliest-emitted element in the stream, and the last element represents the latest-emitted element. The definition above says that any addition to an operation stream --- be it a top-level operation stream or a streamlet --- must be \emph{appended}. As we shall see in the operational semantics, any \emph{removal} from the operation stream will be from the head. It is through this consistent access pattern that the chronological order of the operations is preserved in our semantics. 

\inlong{%%%%%%%%%%%%%%%%%%%%%%%%% begin inlong %%%%%%%%%%%%%%%%%%%%%%%%%

The access of the operation stream follows a distinct FIFO pattern: \emph{appending} to the stream always applies to its end, and \emph{removing} from the stream is always from its beginning. For the top-level operation stream, recall that appending happens at the frontend, and removing happens at the backend, the FIFO access pattern thus coincides with the intuition that the operations emitted earlier at the frontend will also be processed earlier. The same pattern also applies to each streamlet inside the in-graph operation stream, which intuitively says that the operations propagated to a node earlier will also be propagated to the next node earlier.

The access pattern to the result store is more relaxed: even though its appending follows the same access pattern as the operation stream, we do not need to support removal, and instead can just support random access for reading from the result stream. The difference here results from two facts: 1) logically, there is no incentive for a continuously processed graph to remove results; Indeed, by allowing a result to hold in the result store indefinitely, newly emitted operations can be dependent (see \S~\ref{sec:coresocial}) on any earlier results. 2) the order of result claim depends on the program control flow, and it has no correspondence with the order in which the results are produced.

}%%%%%%%%%%%%%%%%%%%%%%%%% end inlong %%%%%%%%%%%%%%%%%%%%%%%%%

% \begin{definition}[Result Store Access]
% %
% \[
% \begin{array}{r@{\ }c@{\ }l}
% \tuple{\data; \ostream; \rstream; \ep} \writeresult \rstream' & \df &
%   \tuple{\data; \ostream; \rstream' \cup \rstream; \ep}
%   \\
% %\tuple{\data; \ostream; \rstream; \ep} \readresult \fvs & \df &
% %  \rstream(\fvs)
% %  \\
% %
% \end{array}
% \]
% %
% %The second case is partial: requiring $\fvs \in \dom(\rstream)$.
% \dnote{need better this.}
% \end{definition}

% The order of stream items in the operation stream matters: intuitively, it reflects the order of graph processing.
% For example,
% $
% [\fvs \mapsto \map\ (\lambda x . \ND{\neone{x}; \netwo{x} + 1; \nethree{x}})\ \K, \fvs \mapsto \map\ (\lambda x . \ND{\neone{x}; \netwo{x} * 2; \nethree{x}})\ \K']
% $.
% The \plaintitle{Emit} places a new operation to the end of the operation stream, indicating it is the newest item in the operation stream. The result store is monotonically increasing, and order does not matter. 
% %\dnote{if so, consider turn it into some set?}
% %\pnote{is it worth adding more notation?}

\inlong{%%%%%%%%%%%%%%%%%%%%%%%%% begin inlong %%%%%%%%%%%%%%%%%%%%%%%%%

By now, it should be easy to understand the encodings in Table~\ref{fig:encodings}.
%\begin{example}[Operation Encoding]
For example, consider the operation
$\mathtt{updatePayload\ a\ nb}$
at Line 16 in Fig.~\ref{fig:serverclient}.
Table~\ref{fig:encodings} says that
the operation is encoded as
$\map\ \mathtt{f}\ \KS{[a]}$
where
$\mathtt{f} = \lambda x . \ND{\neone{x}; \netwo{\ep'}; \nethree{x}}$
\emph{i.e.,} a map operation which updates
the node payload value to the payload of $\mathtt{nb}$
for the node $\mathtt{a}$.
%\end{example}

}%%%%%%%%%%%%%%%%%%%%%%%%% end inlong %%%%%%%%%%%%%%%%%%%%%%%%%

\section{\systemprefix{} Operational Semantics}
\label{sec:highlevelsemantics}

\begin{figure}[t]
\centering

\begin{tabular}{l@{\hspace{1.0cm}}r}
\scalebox{\scalemath}{
$\begin{array}{@{}l@{\ \ }l@{\ \ }lr}

\fE
& \defassign &
\tuple{\data; \ostream; \rstream; \cE}
& \textit{frontend context}
\\

\bE
& \defassign &
\tuple{\bullet; \ostream; \rstream; \ep}
& \textit{backend context}
\\

\tE
& \defassign &
\bE[\data \concat \bullet \concat \data]
& \textit{task context}
\\

\lE
& \defassign &
\tE[[\tuple{\cE; \ostream}]]
& \textit{load context}
\\ & \mid &
\multicolumn{2}{@{}l}{
\tE[[\tuple{\inode; \ostream \concat [\fvs \mapsto \fold\ \func\ \cE\ \K] \cons \ostream}]]
}
\\
\end{array}$
}
&
\scalebox{\scalemath}{
$\begin{array}{@{}l@{\ \ }l@{\ \ }lr}

\cE
& \defassign &
\bullet
\mid \cE\ \ep
\mid \val\ \cE
\mid \neselect{\cE}
\mid \KS{\aseq{\nkey} \concat [ \cE ] \concat \aseq{\ep}}
& \textit{expression context}
\\ & \mid &
\ND{\cE; \ep; \ep}
\mid \ND{\nkey; \cE; \ep}
\mid \ND{\nkey; \integer; \cE}
\\ & \mid &
\cE \setadd \ep \mid \K \setadd \cE
\mid \cE \setsubtract \ep \mid \K \setsubtract \cE
\\ & \mid &
\toserver\ \addn\ \cE
\mid \toserver\ \map\ \cE\ \ep
\mid \toserver\ \map\ \func\ \cE\
\\ & \mid &
\multicolumn{2}{@{}l}{
\toserver\ \fold\ \cE\ \ep\ \ep
\mid \toserver\ \fold\ \func\ \cE\ \ep
\mid \toserver\ \fold\ \func\ \val\ \cE
\mid \fromserver\ \cE
}
\\

% \mathbb{END}
% & \defassign &
% \tuple{\emptyseq; \bullet \cons \ostream; \rstream; \ep}
% \\ & \mid &
% \tuple{\data \concat [\tuple{\nkey; \bullet \cons \ostream}]; \ostream; \rstream; \ep}
% \\

\end{array}$
}
\end{tabular}

\caption{Evaluation Contexts
% \pnote{fix and use $\mathbb{END}$}
% \pnote{define $\mathbb{END}$ for eager if we keep it}
%\pnote{problem, order does matter in the $\cE$ context when evaluating members of $\KS{}$}
}
\label{fig:contexts}

\end{figure}

\begin{figure}[t]
\centering
%\linepenalty=10000

% \begin{tabular}{lr}
% \multicolumn{2}{@{}c@{}}{
% \raisebox{-0.14cm}{
% \begin{tikzpicture}
% \coordinate (x) at (0,0);
% \coordinate (y) at (\linewidth-\pgflinewidth,0);
% \draw (x) -- (y);
% \end{tikzpicture}
% }
% }
% \\
% \framebox{Frontend}&
% \systembox{$\config \configreduce \config$}
% \end{tabular}
%\scalebox{\scalemath}{
\begin{mathpar}
\scalebox{\scalemath}{$
\inferrule*[left=\plaintitle{Emit}]{
   \fvs\ \fresh
}{
   \fE[\toserver\ \strictprocessor]
   \configreduce
   \fE[\fvs] \writeoperation [\fvs \mapsto \strictprocessor]
}
$}

\scalebox{\scalemath}{$
\inferrule*[left=\plaintitle{Claim}]{
   \fE[\fromserver\ \fvs] = \tuple{\data; \ostream; \rstream; \ep}
%\\ \fvs \in \dom(\rstream)
}{
   \fE[\fromserver\ \fvs]
   \configreduce
   \fE[\rstream(\fvs)]
}
$}

\scalebox{\scalemath}{$
\inferrule*[left=\plaintitle{Beta}]{}{
   \fE[(\lambda \VAR . \ep)\ \val]
   \configreduce
   \fE[\ep[\val / \VAR]]
}
$}

\scalebox{\scalemath}{$
\inferrule*[left=\plaintitle{Node}]{
   \inode = \ND{\val_1; \val_2; \val_3}
}{
   \fE[\neselect{\inode}]
   \configreduce
   \fE[\val_{\neI}]
}
$}

\scalebox{\scalemath}{$
\inferrule*[left=\plaintitle{KSA}]{}{
   \fE[\KS{\aseq{\nkey}} \setadd \KS{\aseq{\nkey'}}]
   \configreduce
   \fE[\KS{\aseq{\nkey} \concat \aseq{\nkey'}}]
}
$}

\scalebox{\scalemath}{$
\inferrule*[left=\plaintitle{KSS}]{}{
   \fE[\KS{\aseq{\nkey}} \setsubtract \KS{\aseq{\nkey'}}]
   \configreduce
   \fE[\KS{\aseq{\nkey} \kssetminus \aseq{\nkey'}}]
}
$}

\scalebox{\scalemath}{$
\inferrule*[left=\plaintitle{Map}]{
   \nkey \in \aset{\nkey}_1
\\ \processor_i = \map\ \func\ \KS{\aseq{\nkey}_i}\ \textrm{for}\ i = 1,2
\\ \inode_1 = \ND{\nkey; \ep; \ep'}
\\ \inode_2 = \ND{\nkey; \netwo{(\func\ \inode_1)}; \nethree{(\func\ \inode_1)}}
\\ \aseq{\nkey}_2 = \aseq{\nkey}_1 \kssetminus \nkey
}{
   \tE[\tuple{\inode_1; [\fvs \mapsto \processor_1] \cons \ostream}]
   \configreduce
   \tE[\tuple{\inode_2; [\fvs \mapsto \processor_2] \cons \ostream}]
}
$}

\scalebox{\scalemath}{$
\inferrule*[left=\plaintitle{Fold}]{
   \nkey \in \aset{\nkey}_1
\\ \processor_i = \fold\ \func\ \ep_i\ \KS{\aseq{\nkey}_i}\ \textrm{for}\ i = 1,2
\\ \inode = \ND{\nkey; \ep'; \ep''}
\\ \ep_2 = \func\ \inode\ \ep_1
\\ \aseq{\nkey}_2 = \aseq{\nkey}_1 \kssetminus \nkey
}{
   \tE[\tuple{\inode; [\fvs \mapsto \processor_1] \cons \ostream}]
   \configreduce
   \tE[\tuple{\inode; [\fvs \mapsto \processor_2] \cons \ostream}]
}
$}

\scalebox{\scalemath}{$
\inferrule*[left=\plaintitle{Prop}]{
   \nkey \notin \hspace{-0.3cm}\bigcup_{\processor \in \ran(\batchgroup)} \hspace{-0.3cm}\target(\processor)
\\ \data_i = [\tuple{\ND{\nkey; \ep; \ep'}; \ostream_i}]\ \textrm{for}\ i = 1,2
\\ \ostream_1 = \batchgroup \cons \ostream_2
}{
   \tE[\data_1 \concat \data]
   \configreduce
   \tE[\data_2 \concat (\data \writeoptwo \batchgroup)]
}
$}

\scalebox{\scalemath}{$
\inferrule*[left=\plaintitle{Complete}]{
   \data_i = [\tuple{\inode; \ostream_i}]\ \textrm{for}\ i = 1,2
\\ \ostream_1 = [\fvs \mapsto \processor] \cons \ostream_2
\\ \target(\processor) = \emptyset
}{
   \tE[\data_1]
   \configreduce
   \tE[\data_2] \writeresult [\pmb{\circlearrowright} \fvs \mapsto \processor]
}
$}
%
% \inferrule*[left=\plaintitle{Final}]{
%    \tE[\tuple{\inode; \fvs \mapsto \processor \cons \ostream}] = \bE[\data \concat \tuple{\inode; \fvs \mapsto \processor \cons \ostream} \cons \data']
% \\ \K = \target(\processor)
% \\ \inode = \tuple{\nkey; \ep; \ep'}
% \\ \K = \KS{\emptyseq} \lor (\data' = \emptyseq \land \nkey \not\ksin \K)
% }{
%    \tE[\tuple{\inode; \fvs \mapsto \processor \cons \ostream}]
%    \configreduce
%    \tE[\tuple{\inode; \ostream}] \writeresult [\fvs \mapsto \pmb{\circlearrowright}(\processor)]
% }

\scalebox{\scalemath}{$
\inferrule*[left=\plaintitle{Last}]{
   \tE = \data \concat \bullet
\\ \inode = \ND{\nkey; \ep; \ep'}
\\ \nkey \notin \target(\processor)
}{
   \tE[\tuple{\inode; \fvs \mapsto \processor \cons \ostream}]
   \configreduce
   \tE[\tuple{\inode; \ostream}] \writeresult [\pmb{\circlearrowright} \fvs \mapsto \processor]
}
$}

\scalebox{\scalemath}{$
\inferrule*[left=\plaintitle{Opt}]
{
   \data_i = [\tuple{\inode; \ostream \concat \ostream_i \concat \ostream'}]\ \textrm{for}\ i = 1,2
\\ \ostream_1 \leadsto \ostream_2, \rstream
}{
   \tE[\data_1]
   \configreduce
   \tE[\data_2] \writeresult \rstream
}
$}

\scalebox{\scalemath}{$
\inferrule*[left=\plaintitle{Load}]{
   \lE[\ep] = \tuple{\data; \ostream; \rstream; \ep''}
\\ \tuple{\emptyseq; \emptyseq; \rstream; \ep}
   \configreduce
   \tuple{\emptyseq; \emptyseq; \rstream; \ep'}
}{
   \lE[\ep]
   \configreduce
   \lE[\ep']
}
$}

\scalebox{\scalemath}{$
\inferrule*[left=\plaintitle{Empty}]{
   \processor \neq \addn\ \integer
}{
   \tuple{\emptyseq; [\fvs \mapsto \processor] \cons \ostream; \rstream; \ep}
   \configreduce 
   \tuple{\emptyseq; \ostream; \{\pmb{\circlearrowright} \fvs \mapsto \processor\} \cup \rstream; \ep}
}
$}

\scalebox{\scalemath}{$
\inferrule*[left=\plaintitle{First}]{
   \processor \neq \addn\ \integer
}{
   \tuple{\data; [\fvs \mapsto \processor] \cons \ostream; \rstream; \ep}
   \configreduce 
   \tuple{\data \writeoptwo [\fvs \mapsto \processor]; \ostream; \rstream; \ep}
}
$}
% 
% \inferrule*[left=\plaintitle{Start$_\readoperation$}]{
%    \readoperation \config = \tuple{\data; \ostream; \rstream; \ep}, [\fvs \mapsto \processor]
% \\ \processor \neq \addn\ \integer
% }{
%    \config
%    \configreduce 
%    \tuple{\data \writeoperation [\fvs \mapsto \processor]; \ostream; \rstream; \ep}
% }

\scalebox{\scalemath}{$
\inferrule*[left=\plaintitle{Add}]{
   \nkey\ \fresh
}{
   \tuple{\data; [\fvs \mapsto \addn\ \integer] \cons \ostream; \rstream; \ep}
   \configreduce 
   \tuple{\ND{\nkey; \integer; \KS{\emptyseq}} \cons \data; \ostream; \{\fvs \xmapsto{\KS{\emptyseq}} \nkey\} \cup \rstream; \ep}
}
$}
\end{mathpar}
%}

\caption{\systemprefix{} Operational Semantics
%\dnote{prop has some issues. I think B1 and B2 have different keys, so I'm not sure why there are no subscripts. In fact, do we even need B2 to be written out?}
%\pnote{B1 and B2 are the same node (before and after the batch is removed)}
%\dnote{wondering if some shades may help here...}
%\pnote{where?}
}

\label{fig:clientserver}

\end{figure}

%
%\systemprefix{} is a parameterized calculus;
%the processing mode parameter $\modep$  allows us to study both
%eager data processing ($\modep = \mathtt{E}$) and
%lazy data processing ($\modep = \mathtt{L}$) under one system.
%
%The relation is defined in .
%%

%\dnote{my polishing pass reaches here; self notes: lazy realiztaion-> incremental load update; load has been introduced informally. make sure node deletion is mentioned later.} 

In this section, we define the behavior of continuous graph processing spanning the frontend and the backend. The main reduction system is presented in \S~\ref{sec:ourlang:1}. The semantics of TLO is an independent system that bridges with the main system via one reduction rule, whose details are in \S~\ref{sec:tlo}. 

\subsection{Semantics for Continuous Graph Processing}
\label{sec:ourlang:1}

Reduction relation
$\config \configreduce \config'$ in Fig.~\ref{fig:clientserver}
says that configuration $\config$
one-step reduces to configuration $\config'$. % under \emph{processing mode parameter} $\modep$.
We use $\configreduce^*$
to represent the reflexive and transitive
closure of $\configreduce$. Evaluation contexts play an important role in defining the semantics of \systemprefix{}, whose definitions appear in Fig.~\ref{fig:contexts}. 
For the convenience of our discussion, we classify $\configreduce$ reduction into 4 forms, based on \emph{where} a reduction happens.

\inlong{%%%%%%%%%%%%%%%%%%%%%%%%% begin inlong %%%%%%%%%%%%%%%%%%%%%%%%%

A \emph{frontend reduction} happens on the frontend. A backend reduction, i.e., one that happens in the graph, may either be a \emph{task reduction} and a \emph{load reduction}. The distinction between \emph{task} and \emph{load} is evocative of the task vs. data duality in programming languages: the former focuses on the evolution of \emph{operations}, whereas the latter addresses that of graph \emph{data}. Finally, a \emph{to-graph reduction} describes the behavior at the boundary between the top-level operation stream and the graph.

Recall that \systemprefix{} captures the behavior of both the frontend and backend, and the two are connected through the top-level operation stream and the result store. The different forms of evaluation contexts correspond to \emph{where} the reductions happen, or formally indeed, where the redex appears in the configuration. Specifically, the \emph{frontend context} $\fE$ and the \emph{backend context} $\bE$ corresponds to our intuition that the reduction happens at the frontend and the backend respectively. The frontend context is used to enable \emph{frontend reductions}: 

\begin{definition}[Frontend Reduction]
\label{def:frontendreduction}
A frontend reduction describes the behavior of the frontend, including its interaction with the top-level operation stream and the result store. We say a reduction step is a frontend reduction iff it is an instance of \plaintitle{Emit}, \plaintitle{Claim}, \plaintitle{Beta}, \plaintitle{Node}, \plaintitle{KSA$_1$}, \plaintitle{KSA$_2$}, \plaintitle{KSS$_1$}, \plaintitle{KSS$_2$}. 
\end{definition}

Within the backend, we further define the \emph{task context} $\tE$ and \emph{load context} $\lE$. The distinction between \emph{task} and \emph{load} is evocative of the task vs. data duality in concurrent/parallel programming languages: the former focuses on the evolution of \emph{operations}, whereas the latter addresses that of graph \emph{data}. To be more specific, we define the reductions that may happen appear in these two contexts are:

%Indeed, operation evolution and data evolution are interconnected, so the distinction here is only meant for our informal description. As we shall see, this is useful for understanding the non-deterministic behavior latent in \systemprefix{}, and later in \S~\ref{sec:parallel}, helpful for gaining intuition on the behavior of parallelism.

\begin{definition}[In-Graph Task Reduction]
\label{def:taskreduction}
An \emph{in-graph task reduction}, or \emph{task reduction} for short, describes the behavior of operation processing in the graph. We say a reduction step is a task reduction iff it is an instance of \plaintitle{Map}, \plaintitle{Fold}, \plaintitle{Prop}, \plaintitle{Complete}, \plaintitle{Last}, or \plaintitle{Opt}.  
\end{definition}

\begin{definition}[In-Graph Load Reduction]
\label{def:loadreduction}
An \emph{in-graph load reduction}, or \emph{load reduction} for short, describes the behavior of processing the \emph{load expression} carried inside a single station, where the load expression is defined as (i) its graph node as an expression, or (ii) the base-case argument (the second argument) of a $\fold$ operation in its streamlet. We say a reduction step is a load reduction iff it is an instance of \plaintitle{Load}.
%\dnote{I think if we had allowed the uparrow to carry into the stream, we could further simplify this as just the range of a streamlet as a hole}
%\pnote{I don't follow}
\end{definition}

To complete our taxonomy, one more form of reduction exists in \systemprefix{}:

\begin{definition}[To-Graph Reduction]
\label{def:tographreduction}
A \emph{to-graph reduction} describes the behavior at the boundary between the top-level operation stream and the graph. We say a reduction step is a to-graph reduction iff it is an instance of \plaintitle{Empty}, \plaintitle{First}, or \plaintitle{Add}. 
\end{definition}

%This 4-category taxonomy of reductions based on \emph{where} they happen carries one subtlety: a load reduction may depend on a frontend reduction. This is shown in \plaintitle{Load}, where the reduction of a load expression at the backend is conditioned by a reduction over a configuration whose backend and top-level operation stream are both set to $\emptyset$, \emph{de facto} only allowing for a frontend reduction. Intuitively, this means we consider every load expression forms its own runtime with a trivial configuration that has no backend graph or operation stream. This simplifies our definition because a load reduction can thus depend on a \plaintitle{Beta}, \plaintitle{Node}, or \plaintitle{Claim} reduction, effectively allowing the computations they represent to happen at the backend of graph processing. 

%At the first glance, the fact that a backend (load) reduction depends on a frontend reduction seems counter-intuitive

}%%%%%%%%%%%%%%%%%%%%%%%%% end inlong %%%%%%%%%%%%%%%%%%%%%%%%%

%Let us now have a detailed look at the semantics of continuous graph processing. The backbone of \systemprefix{} is the second category below, on task reduction. 

\paragraph{1) Frontend Reduction}

%The frontend behavior is straightforward. 

Rules with the $\fE$ evaluation context enable reductions that happen on the frontend. The pair of \plaintitle{Emit} and \plaintitle{Claim} rules define the behavior of asynchronous operation processing at the frontend, with the former placing an operation on the top-level operation stream, and the latter reading from the result store. The definition here follows future semantics, where the fresh label in \plaintitle{Emit} is the future value.
We say an operation is emittable if all of its arguments are values, which we represent as metavariable $\strictprocessor$:
\[
\begin{array}{@{}l@{\ \ }l@{\ \ }lr}
\strictprocessor
& \defassign &
\addn\ \integer
\mid \map\ \func\ \K
\mid \fold\ \func\ \val\ \K
\\
\end{array}
\]
%shown in Fig.~\ref{fig:pgraph:grammar}.
Both nodes and key lists as first-class citizens can be constructed at the frontend. 
%For example, a programmer may construct a node and use it as the second argument (base-case) for a $\fold$ operation, and construct a key list and use it as the third argument (target) for the same operation. The frontend 
A node may be eliminated through \plaintitle{Node}. Key list concatenation and subtraction are supported are defined through \plaintitle{KSA} and \plaintitle{KSS} respectively.
The rest of the frontend computation is enabled by \plaintitle{Beta}, in a call-by-value style.

\paragraph{2) In-Graph Task Reduction}

%The bulk of reduction rules address in-graph task reductions at the backend. 

On the backend, in-graph processing may either be enabled by a \emph{task reduction} and a \emph{load reduction}, the first of which we describe now. Rules with the $\tE$ evaluation context enable reductions that perform a \emph{task}, i.e., a step on operation processing. 

The task that ``drives'' the data processing at the backend is propagation, an instance of \plaintitle{Prop}. A step of operation propagation involves two consecutive stations in the runtime graph. The reduction removes the \emph{head} element (the oldest element) from the the streamlet of the first station, and places it to the \emph{last} element (the youngest element) of the streamlet in the second station. 
It is important to observe that the selection of redex for propagation is non-deterministic according to the definition of $\tE$. In other words, propagation may happen between any adjacent two stations in the runtime graph. %As a result, there is no need to eagerly process the every operation, i.e., selecting the first operation in the runtime graph (operation labeled $\fvs_1$) for every step of propagation. 
\inlong{%%%%%%%%%%%%%%%%%%%%%%%%% begin inlong %%%%%%%%%%%%%%%%%%%%%%%%%
As shown in the example, the redex selection of the first reduction propagates $\fvs_1$, whereas that of the second reduction propagates $\fvs_2$. On the high level, the non-determinism in propagation redex selection naturally enables lazy operation processing: when the second reduction step happens, the processing of $\fvs_1$ is deferred. 
}%%%%%%%%%%%%%%%%%%%%%%%%% end inlong %%%%%%%%%%%%%%%%%%%%%%%%%
Furthermore, observe that \plaintitle{Prop} supports batched propagation: the unit for propagation may contain more than one operation. We will defer this discussion to \S~\ref{sec:tlo}, when batching is formally defined. 
\inlong{%%%%%%%%%%%%%%%%%%%%%%%%% begin inlong %%%%%%%%%%%%%%%%%%%%%%%%%
We illustrate the behavior of propagation with an example:

\begin{example}[Propagation]
\label{exm:prop}
Consider the following configuration
illustrated in Fig.~\ref{fig:datacentriclaziness}(c),
\[
\lb
\begin{array}{l@{\hspace{-10em}}l}
[\tuple{\ND{\nkey_{\mathtt{eve}}; \integer_{\mathtt{eve}}; \KS{\emptyseq}}; \\ &[[\fvs_1 \mapsto \mathtt{addRelationship\ \nkey_\mathtt{cam}\ \nkey_\mathtt{bob}}],\\ &[\fvs_2 \mapsto \mathtt{addRelationship\ \nkey_\mathtt{deb}\ \nkey_\mathtt{cam}}]]},
\\
\tuple{\ND{\nkey_{\mathtt{deb}}; \integer_{\mathtt{deb}}; \KS{[\nkey_{\mathtt{cam}}]}}; \emptyseq},
\\
\tuple{\ND{\nkey_{\mathtt{cam}}; \integer_{\mathtt{cam}}; \KS{\emptyseq}}; \emptyseq},
\\
\tuple{\ND{\nkey_{\mathtt{bob}}; \integer_{\mathtt{bob}}; \KS{\emptyseq}}; \emptyseq},
\\
\tuple{\ND{\nkey_{\mathtt{amy}}; \integer_{\mathtt{amy}}; \KS{\emptyseq}}; \emptyseq}]
\end{array}
; \emptyseq; \rstream; \ep \rb
\]
With \plaintitle{Prop}, this configuration can reduce to the following configuration in one step, as
illustrated in Fig.~\ref{fig:datacentriclaziness}(d),
\[
\lb
\begin{array}{l@{\hspace{-10em}}l}
[\tuple{\ND{\nkey_{\mathtt{eve}}; \integer_{\mathtt{eve}}; \KS{\emptyseq}};\\ &[[\fvs_2 \mapsto \mathtt{addRelationship\ \nkey_\mathtt{deb}\ \nkey_\mathtt{cam}}]]},
\\
\tuple{\ND{\nkey_{\mathtt{deb}}; \integer_{\mathtt{deb}}; \KS{[\nkey_{\mathtt{cam}}]}};\\ &[[\fvs_1 \mapsto \mathtt{addRelationship\ \nkey_\mathtt{cam}\ \nkey_\mathtt{bob}}]]},
\\
\tuple{\ND{\nkey_{\mathtt{cam}}; \integer_{\mathtt{cam}}; \KS{\emptyseq}}; \emptyseq},
\\
\tuple{\ND{\nkey_{\mathtt{bob}}; \integer_{\mathtt{bob}}; \KS{\emptyseq}}; \emptyseq},
\\
\tuple{\ND{\nkey_{\mathtt{amy}}; \integer_{\mathtt{amy}}; \KS{\emptyseq}}; \emptyseq}]
\end{array}
; \emptyseq; \rstream; \ep \rb
\]
%Instead of further propagating $\processor_1$,
As illustrated in Fig.~\ref{fig:datacentriclaziness}(e),
another step of \plaintitle{Prop} can reduce the configuration
to the following:
\[
\lb
\begin{array}{l@{\hspace{-10em}}l}
[\tuple{\ND{\nkey_{\mathtt{eve}}; \integer_{\mathtt{eve}}; \KS{\emptyseq}}; \emptyseq},
\\
\tuple{\ND{\nkey_{\mathtt{deb}}; \integer_{\mathtt{deb}}; \KS{[\nkey_{\mathtt{cam}}]}};\\ &[[\fvs_1 \mapsto \mathtt{addRelationship\ \nkey_\mathtt{cam}\ \nkey_\mathtt{bob}}],\\ &[\fvs_2 \mapsto \mathtt{addRelationship\ \nkey_\mathtt{deb}\ \nkey_\mathtt{cam}}]]},
\\
\tuple{\ND{\nkey_{\mathtt{cam}}; \integer_{\mathtt{cam}}; \KS{\emptyseq}}; \emptyseq},
\\
\tuple{\ND{\nkey_{\mathtt{bob}}; \integer_{\mathtt{bob}}; \KS{\emptyseq}}; \emptyseq},
\\
\tuple{\ND{\nkey_{\mathtt{amy}}; \integer_{\mathtt{amy}}; \KS{\emptyseq}}; \emptyseq}]
\end{array}
; \emptyseq; \rstream; \ep \rb
\]
\end{example}
}%%%%%%%%%%%%%%%%%%%%%%%%% end inlong %%%%%%%%%%%%%%%%%%%%%%%%%

The realizations of $\map$ and $\fold$ are defined by \plaintitle{Map} and \plaintitle{Fold}, over a single station as the redex. The task reduction for $\map$ realization happens when the key of the redex is included in the target key list, the second argument of the $\map$ operation. It further applies the mapping function (the first argument) to the current node, which computes a new node to update the current node. Following the convention in graph processing, our calculus does not allow a $\map$ operation to update the key of the node: even though the node payload and the graph topology can be changed in dynamic graphs, keys as unique identifiers of nodes do not change. This is ensured in \plaintitle{Map} by updating only the second and third components of the node based on the computation of the mapping function, while preserving the key. Similar to $\map$, a task reduction for $\fold$ realization happens when the key of the redex is included in its target key list (the third argument). The folding function is applied to the current node and the base-case expression, whose resulting expression becomes the base-case expression for further propagation. Both \plaintitle{Map} and \plaintitle{Fold} demonstrate the incremental nature of load update  (recall \S~\ref{subsec:ingraph}): when being applied, the $\map$ operation does not immediately evaluate the resulting payload expression or adjacency list expression to a value; similarly, the $\fold$ operation does not immediately evaluate the application of the folding function.

As the target of $\map$ and $\fold$ operations may contain multiple keys, completing each operation may involve multiple realizations. That said, both \plaintitle{Map} and \plaintitle{Fold} are defined over one station only. This is possible because the targets of these operations are updated ``as they go'': after the operation is realized at a node of key $\nkey$, the operation is placed back to its streamlet of the current station, except that its target no longer contains $\nkey$. %As a result, the next step of reduction may be an instance of \plaintitle{Prop}, and the processing of the operation in the remaining traversal sequence continues. 
%Taken as a whole, 
When the target of the $\map$ (or $\fold$) operation contains multiple keys, its processing is ``incremental'': the processing consists of many \plaintitle{Prop} steps occasionally interposed by \plaintitle{Map} (or \plaintitle{Fold}) steps. 
We will show an example of this incremental process shortly, in Example~\ref{ex:foldinc}. 

%Furthermore, note that the operation processing is \emph{lazy}: \systemprefix{} do not require these steps to be taken consecutively; instead, each \plaintitle{Prop}, \plaintitle{Map}, or \plaintitle{Fold} step during the processing of an operation can be intermingled with other reduction steps that do not concern the operation at all, such as a frontend reduction, or a backend reduction related to the processing of another operation.  \dnote{can be merged with the discussion on eager/lazy above} 

%The single-dam reduction for \plaintitle{Map} and \plaintitle{Fold} relies on the fact that the selector is updated. Recall however there is also a 
% The special key list $\KSstar$ introduces some complexity. In this case, every dam in the runtime graph subjects to  a \plaintitle{Map} (or \plaintitle{Fold}) reduction. After the operation is realized at a dam, we need to indicate that the operation is ready for propagating to the next dam, which we use the special symbol $\KShollowstar$. This leads to a simple design that ``toggles'' between $\KSstar$ and $\KShollowstar$: every \plaintitle{Map} (or \plaintitle{Fold}) reduction turns $\KSstar$ to $\KShollowstar$ (see the $\kssetminus$ operator in Definition~\ref{def:ksops}), and every \plaintitle{Prop} turns a $\KSstar$ or $\KShollowstar$ to $\KSstar$ (see the $\ksrestore$ operator in Definition~\ref{def:osa}).
% \dnote{revisit} 
%The toggling can be seen in the $\ksrestore$ operator defined in \S~\ref{sec:}. 

Finally, \plaintitle{Complete} and \plaintitle{Last} are a pair of rules to ``wrap up'' the processing of an operation. The former captures the case when a $\map$ or $\fold$ operation is successfully realized over every node defined by its target. The latter represents the case when the last node is reached in the graph. In both cases, the $\pmb{\circlearrowright}$ operator computes the result to be placed to the result store:
%\begin{center}
%\scalebox{\scalemath}{
%$
\[
\pmb{\circlearrowright} \fvs \mapsto \strictprocessor \df
  \begin{cases} \fvs \xmapsto{\K} 0 & \fif \strictprocessor = \map\ \func\ \K \arcr \fvs \xmapsto{\K} \val & \fif \strictprocessor = \fold\ \func\ \val\ \K \end{cases}
\]
%$
%}
%\end{center}
For a completed $\map$ operation, we return the default value of 0. For the $\fold$ operation, we return its incrementally updated base-case value. We further place the residual target key list to the result store, i.e., the remaining keys in the operation target key list of $\map$ or $\fold$ while the operation has reached the last node. 

%
%\dnote{describe how star works; probably with an example. It's too dry to explain otherwise. }

A quick case analysis can reveal that each task reduction only involves at most two consecutive stations in the station sequence (\plaintitle{Prop}), and often one station only (\plaintitle{Map}, \plaintitle{Fold}, \plaintitle{Complete}, \plaintitle{Last}, or \plaintitle{Opt}). %A load reduction only happens inside one station. 
In other words, both task reductions exhibit \emph{local} behaviors.
\inlong{%%%%%%%%%%%%%%%%%%%%%%%%% begin inlong %%%%%%%%%%%%%%%%%%%%%%%%%
From now on, we call the subset of stations in the station sequence where a task/load reduction happens a \emph{station neighborhood}.
}%%%%%%%%%%%%%%%%%%%%%%%%% end inlong %%%%%%%%%%%%%%%%%%%%%%%%%

%with reference to examples earlier, such as \textsc{CorePR}. As realization is achieved, the operation will be placed on the operation stream again for further propagation, with its selector updated (the \dnote{} operator). % is updated (the \dnote{fill} operator). 
%In \plaintitle{Map}, the first argument of $\map$ (a mapping function) is applied over the node (key, payload, and adjacency list) to compute its new payload, wheres the second argument of $\map$ (another mapping function) is applied over the node to compute its new adjacency list. The node is updated with the new payload and new adjacency list. In \plaintitle{Fold}, the first argument of $\fold$ is applied over the node and the second argument (the result of folding so far). The resulting expression further propagates down the stream.  

\paragraph{3) In-Graph Load Reduction}

On the backend, the other form of in-graph processing is a \emph{load reduction}, enabled by \plaintitle{Load}. Unlike task reductions that process \emph{operations}, load reductions process \emph{data}. What constitutes a data \emph{load} is evident by an inspection on the $\lE$ evaluation context, whose fulfilling redex we call a \emph{load expression}: (i) the graph node inside a station, or (ii) the base-case argument (the second argument) of a $\fold$ operation in the streamlet of a node.

%With both $\map$ and $\fold$ supporting lazy realization, a natural question to ask when the computation embodied by the mapping/folding function happens. This goal is achieved by \plaintitle{Load}. Recall in \S~\ref{sec:ourlang:1}, we descried that a load reduction may happen in a dam node, or an operation inside a dam streamlet. The former enables the computation embodied by a mapping function application to happen, whereas the latter enables the computation embodied by a folding function application. 

As revealed by \plaintitle{Load}, a load reduction depends on a frontend reduction: 
%This is shown in \plaintitle{Load}, where the reduction of a load expression at the backend is conditioned by a 
the premise of the rule is a reduction over a configuration whose backend and top-level operation stream are both set to $\emptyset$, \emph{de facto} only allowing for a frontend reduction. Intuitively, this means we consider every load expression forms its own runtime with a trivial configuration that has no backend graph or operation stream. This simplifies our definition because a load reduction can thus depend on a \plaintitle{Beta}, \plaintitle{Node}, or \plaintitle{Claim} reduction, effectively allowing the reductions they represent to happen at the backend of graph processing. 
%Toward the end of \S~\ref{sec:ourlang:1}, we also discussed that a load reduction is defined through a frontend reduction, so that reductions such as \plaintitle{Beta}, \plaintitle{Node}, or \plaintitle{Claim} can happen as sub-reductions of a load reduction. 
The last case is especially important, in that it enables a dependent operation to claim its argument in the form of a future, while processing at the backend (recall \S~\ref{sec:coresocial}). %Thanks to phase distinction enforced by our type system, \plaintitle{Emit} cannot be applied on the backend for well-typed programs. 

Before we move on, let us illustrate the behavior of task and load reductions, especially on how a propagation step, a realization step, and a load reduction step interleave with each other, through an example:

\begin{example}[Incremental Folding]
\label{ex:foldinc}
Consider a configuration
%based on a desugared version of the program \textsc{CorePR} in Fig.~\ref{fig:pagerank},
%\pnote{I've removed the reference to PR, since we no longer use fpSum. Or, maybe we could keep the reference but say it's loosely based on PR}
where the backend
%only
consists of two stations, with nodes $\inode_1$ and $\inode_2$,
and a $\fold$ operation has
%now
been propagated to the first station.
The operation has a folding function $\func$ representing a function which sums up the payloads of all target nodes (this is a simplified version of the \textsc{CorePR} example),
and a target key list of $\KS{[\nkey_1, \nkey_2]}$. 
%$
%\tuple{[\tuple{\inode_1; [[\fvs \mapsto \fold\ \func\ \ND{\_; 0; \KS{\emptyseq}}\ \KS{[\nkey_1, \nkey_2]}]]},  \tuple{\inode_2; \emptyseq}]; \emptyseq; \emptyset; \ep}
%$.
The following is one reduction sequence which
ends in the $\fold$ being completed:
\begin{center}
\scalebox{\scalemath}{
$\begin{array}{@{}r@{\ }l@{\hspace{-4em}}r@{\,}l@{\,}r@{\,}r@{}}
&\tuple{[\tuple{\inode_1; [[\fvs \mapsto \fold\ \func\ \inode_0\ \KS{[\nkey_1, \nkey_2]}]]},&  \tuple{\inode_2; \emptyseq}];& \emptyseq;& \emptyset; \ep}
\\
(\plaintitle{Fold}) \configreduce
&\tuple{[\tuple{\inode_1; [[\fvs \mapsto \fold\ \func\ (\func\ \inode_1\ \inode_0)\ \KS{[\nkey_2]}]]},&  \tuple{\inode_2; \emptyseq}];& \emptyseq;& \emptyset; \ep}
\\
(\plaintitle{Prop}) \configreduce
&\tuple{[\tuple{\inode_1; \emptyseq},&  \tuple{\inode_2; [[\fvs \mapsto \fold\ \func\ (\func\ \inode_1\ \inode_0)\ \KS{[\nkey_2]}]]}];& \emptyseq;& \emptyset; \ep}
\\
(\plaintitle{Fold}) \configreduce
&\tuple{[\tuple{\inode_1; \emptyseq},& \tuple{\inode_2; [[\fvs \mapsto \fold\ \func\ (\func\ \inode_2\ (\func\ \inode_1\ \inode_0))\ \KS{\emptyseq}]]}];& \emptyseq;& \emptyset; \ep}
\\
(\plaintitle{Load}) \configreduce^*
&\tuple{[\tuple{\inode_1; \emptyseq},& \tuple{\inode_2; [[\fvs \mapsto \fold\ \func\ \inode_0'\ \KS{\emptyseq}]]}];& \emptyseq;& \emptyset; \ep}
\\
(\plaintitle{Last}) \configreduce
&\tuple{[\tuple{\inode_1; \emptyseq},& \tuple{\inode_2; \emptyseq}];& \emptyseq;& \{\fvs \xmapsto{\KS{\emptyseq}} \inode_0'\}; \ep}
\end{array}$
}
\end{center}
where
$\inode_i = \ND{\nkey_i; i; \KS{\emptyseq}}$
for $i = 0, 1,2$
and
$\inode_0' = \tuple{\nkey_0; 3; \KS{\emptyseq}}$.
\end{example}
%\dnote{I don't understand the bit that nk is k2. Can you just say it inforamlly what fpSum tries to do? I wonder if it would be easier if f is just some form of +; it's easier for people to see }
%\dnote{Maybe it is better to have N0 to represent that initial node. That way we don't have to see the super long reduction}
%\dnote{If N2 station is right aligned, it maybe easier to read}
%\pnote{I've simplified it now}

\paragraph{4) To-Graph Reduction}

The three rules that capture the behavior at the boundary of the top-level operation stream and the graph are simple. \plaintitle{Empty} considers the bootstrapping case where the graph so far contains no nodes. If the operation is a $\map$ or $\fold$ operation, a result is immediately returned. \plaintitle{First} removes the head element from the top-level operation stream, and places it as the last element of the streamlet associated with the first node. 

%Recall that operation streams (and streamlets) follow the FIFO access: the head element in a stream is the oldest unit placed into it, whereas the last element in a stream is the newest unit. 

According to \plaintitle{Add}, a new node is created with a freshly generated key. In \systemprefix{} we adopt a simple design for node addition: they are always placed at the beginning of the graph station sequence. This can be seen in \plaintitle{Add}. It also explains why an $\addn$ reduction is a to-graph reduction not an in-graph one.
\inlong{
In \S~\ref{subsec:add}, we will discuss some alternative designs. 
}

\subsection{Temporal Locality Optimization}
\label{sec:tlo}

\begin{figure}[t]

\centering

\begin{tabular}{r}
\multicolumn{1}{@{}c@{}}{
\raisebox{-0.14cm}{
\begin{tikzpicture}
\coordinate (x) at (0,0);
\coordinate (y) at (\linewidth-\pgflinewidth,0);
\draw (x) -- (y);
\end{tikzpicture}
}
}
\\
\systembox{$\ostream \leadsto \ostream, \rstream$}
\end{tabular}
\scalebox{\scalemath}{
$\begin{array}{@{}l@{\ \ \ }r@{\ }c@{\ }ll}

\htitle{Batch}
&
[\batchgroup, \batchgroup']
& \leadsto &
[\batchgroup \concat \batchgroup'],
\emptyset
\\[2.0ex]

\htitle{Unbatch}
&
[\batchgroup_1 \concat \batchgroup_2]
& \leadsto &
[\batchgroup_1, \batchgroup_2],
\emptyset
& \fif
\batchgroup_i \neq \emptyseq\ \textrm{for}\ i = 1,2
\\[2.0ex]

\htitle{ReorderD}
&
[[\fvs_1 \mapsto \processor_1],
[\fvs_2 \mapsto \processor_2]]
& \leadsto &
[[\fvs_2 \mapsto \processor_2],
[\fvs_1 \mapsto \processor_1]],
\emptyset
& \fif
\target(\processor_1) \cap \target(\processor_2) = \emptyset
\\[2.0ex]

\htitle{ReorderRR}
&
[[\fvs_1 \mapsto \processor_1], [\fvs_2 \mapsto \processor_2]]
& \leadsto &
[[\fvs_2 \mapsto \processor_2], [\fvs_1 \mapsto \processor_1]],
\emptyset
& \fif
\processor_i = \fold\ \func_i\ \ep_i\ \K_i\ \textrm{for}\ i=1,2
\\[2.0ex]

\htitle{ReorderRW}
&
[[\fvs_1 \mapsto \map\ \func_1\ \KS{\aseq{\nkey}_1}],
& \leadsto &
[[\fvs_2 \mapsto \fold\ (\func_2 \dcomp{\aset{\nkey}_1} \func_1)\ \ep\ \K_2],
\\
&
[\fvs_2 \mapsto \fold\ \func_2\ \ep\ \K_2]]
&&
[\fvs_1 \mapsto \map\ \func_1\ \KS{\aseq{\nkey}_1}]],
\emptyset
\\[2.0ex]

\htitle{FuseM}
&
[[\fvs_1 \mapsto \map\ \func_1\ \K],
& \leadsto &
[[\fvs_1 \mapsto \map\ (\func_2 \circ \func_1)\ \K]],
& \fif
(\func_2 \circ \func_1) \not\feq \mathtt{Id}
\\
&
[\fvs_2 \mapsto \map\ \func_2\ \K]]
&&
\{\fvs_2 \xmapsto{\KS{\emptyseq}} 0\}
\\[2.0ex]

\htitle{FuseMId}
&
[[\fvs_1 \mapsto \map\ \func_1\ \K],
& \leadsto &
[],
& \fif
(\func_2 \circ \func_1) \feq \mathtt{Id}
\\
&
[\fvs_2 \mapsto \map\ \func_2\ \K]]
&&
\{\fvs_1 \xmapsto{\KS{\emptyseq}} 0, \fvs_2 \xmapsto{\KS{\emptyseq}} 0\}
\\[2.0ex]

% \htitle{FuseMSplit}
% &
% [[\fvs \mapsto \map\ \func\ \func'\ \K]]
% & \leadsto &
% [[\fvs \mapsto \map\ \func\ \func'\ \K_1],
% [\fvs' \mapsto \map\ \func\ \func'\ \K_2]],
% \emptyseq
% \\ && \fif &
% \K_1 \kscup \K_2 = \K,
% \\ &&&
% \K_1 \kscap \K_2 = \emptyseq,
% \\ &&&
% \fvs'\ \fresh
% \\

% \htitle{FuseID}
% &
% [[\fvs \mapsto \map\ \func\ \K]]
% & \leadsto &
% \emptyseq,
% [\fvs \xmapsto{\emptyseq} 0]
% \\ && \fif &
% \func \feq \lambda x . \lambda y . \lambda z . y,
% \func' \feq \lambda x . \lambda y . \lambda z . z
% \\

\htitle{Reuse}
&
[[\fvs_1 \mapsto \fold\ \func\ \ep\ \KS{\aseq{\nkey}_1}],
& \leadsto &
[[\fvs_1 \mapsto \fold\ \func\ \ep\ \KS{\aseq{\nkey}_1}],
& \fif
\aset{\nkey}_1 \subseteq \aset{\nkey}_2,
\\
&
[\fvs_2 \mapsto \fold\ \func\ \ep\ \KS{\aseq{\nkey}_2}]]
&&
[\fvs_2 \mapsto \fold\ \func\ (\fromserver\ \fvs_1)\ \KS{\aseq{\nkey}_2 \kssetminus \aseq{\nkey}_1}]],
\emptyset
&
\fif\ \textrm{$\func$ is commutative}
\\%[2ex]

% \htitle{ReuseMqoE}
% &
% [[\fvs_1 \mapsto \fold\ \func\ \ep\ \K],
% & \leadsto &
% [[\fvs_1 \mapsto \fold\ \func\ \ep\ \K],
% &
% \fif\ \textrm{$\func$ is commutative}
% \\
% &
% [\fvs_2 \mapsto \fold\ \cE[\func]\ \ep'\ \K]]
% &&
% [\fvs_2 \mapsto \fold\ \cE[\fromserver\ \fvs_1]\ \ep'\ \K]],
% \emptyset
% \\

%\htitle{ReuseMqoE}
%&
%[[\fvs_1 \mapsto \fold\ \func_1\ \ep\ \KS{[\nkey]}],
%& \leadsto &
%[[\fvs_1 \mapsto \fold\ \func_1\ \ep\ \KS{[\nkey]}],
%& \fif \func_2 \equiv \lambda x . \lambda y . \func\ x\ (\func_1\ x\ y)
%\\
%&
%[\fvs_2 \mapsto \fold\ \func_2\ \ep\ \KS{[\nkey]}]]
%&&
%[\fvs_2 \mapsto \fold\ \func_2'\ \ep\ \KS{[\nkey]}]],
%\emptyset
%& \func_2' = \lambda x . \lambda y . \func\ x\ (\fromserver\ \fvs_1)
%\\

\end{array}$
}

\caption{Temporal Locality Optimization
%\dnote{I suggest we cut the last rule: we don't have any discussion on it, and its a bit hard to decipher. If you do so, it's best to rename the ReuseMQOD back to Reuse. If you do so, make sure you update the text, including related work.} 
% \pnote{for MqoE, change K to singleton set
% then see if we can can do it. $\func'$ has to have $\func$ applied to input or something, or $\func' = \func'' \circ \func$, something like that, then replace somehow the entire $\func\ x\ y$}
%\dnote{a small problem with the side conditions; those are program expressions. We are essentially saying that you have to reduce one expression to reduce another. If we really need a form like this, we need to use the frontend reduction as a side condition. I think a better way is simply take the sets out of the KS, and then apply a standard set operation.} 
%\pnote{Which side condition? I think those are all meta-theory notations. The program expressions are $\setadd$ and $\setsubtract$}
}
\label{fig:fusion}

\end{figure}

The \plaintitle{Opt} rule bridges the main reduction relation ($\configreduce$) with the $\leadsto$ relation, which incarnates different forms of temporal locality optimization. 
%Temporal locality optimization bridges with the streaming graph reduction system through rule \plaintitle{Opt}. The specific optimization rewriting rules are defined through reduction relation
Defined in Fig.~\ref{fig:fusion}, the $\ostream \leadsto \ostream', \rstream$ relation says that operation stream $\ostream$ reduces 
to operation stream $\ostream'$ in one step, while producing result
$\rstream$.
%
%Operations within the graph may be
%rearranged or combined before they are realized,
%therefore reducing the overall number of computation steps involved.
%
%Note that this relation is unique
%for lazy graph processing,
%as eagerness prevents such optimizations.

\htitle{Batch} and \htitle{Unbatch} allow units in the in-graph operation streams to be batched and unbatched. As the \plaintitle{Opt} rule can be applied over the streamlet in any station, batching and unbatching may happen in-graph at an arbitrary station. 
\inlong{%%%%%%%%%%%%%%%%%%%%%%%%% begin inlong %%%%%%%%%%%%%%%%%%%%%%%%%
Recall that \plaintitle{Prop} supports batched propagation, so an alternative way of reducing the configuration we showed in Example~\ref{exm:prop} exists:

\begin{example}[Batched Propagation]
Consider the starting configuration of Example~\ref{exm:prop},
corresponding to the graph illustrated in Fig.~\ref{fig:datacentriclaziness}(c).
According to \htitle{Batch}, the configuration can one-step
reduce to the following, where the two operations are part of
the same stream unit:
\[
\lb
\begin{array}{l@{\hspace{-10em}}l}
[\tuple{\ND{\nkey_{\mathtt{eve}}; \integer_{\mathtt{eve}}; \KS{\emptyseq}};\\ &[[\fvs_1 \mapsto \mathtt{addRelationship\ \nkey_\mathtt{cam}\ \nkey_\mathtt{bob}}, \fvs_2 \mapsto \mathtt{addRelationship\ \nkey_\mathtt{deb}\ \nkey_\mathtt{cam}}]]},
\\
\tuple{\ND{\nkey_{\mathtt{deb}}; \integer_{\mathtt{deb}}; \KS{[\nkey_{\mathtt{cam}}]}}; \emptyseq},
\\
\tuple{\ND{\nkey_{\mathtt{cam}}; \integer_{\mathtt{cam}}; \KS{\emptyseq}}; \emptyseq},
\\
\tuple{\ND{\nkey_{\mathtt{bob}}; \integer_{\mathtt{bob}}; \KS{\emptyseq}}; \emptyseq},
\\
\tuple{\ND{\nkey_{\mathtt{amy}}; \integer_{\mathtt{amy}}; \KS{\emptyseq}}; \emptyseq}]
\end{array}
; \emptyseq; \rstream; \ep \rb
\]
Corresponding with Example~\ref{ex:batch} in \S~\ref{subsec:tlo},
one instance of \plaintitle{Prop} can then
propagate the stream unit to $\nkey_{\mathtt{deb}}$, since neither
operation targets $\nkey_\mathtt{eve}$:
\[
\lb
\begin{array}{l@{\hspace{-10em}}l}
[\tuple{\ND{\nkey_{\mathtt{eve}}; \integer_{\mathtt{eve}}; \KS{\emptyseq}}; \emptyseq},
\\
\tuple{\ND{\nkey_{\mathtt{deb}}; \integer_{\mathtt{deb}}; \KS{[\nkey_{\mathtt{cam}}]}};\\ &[[\fvs_1 \mapsto \mathtt{addRelationship\ c\ b}, \fvs_2 \mapsto \mathtt{addRelationship\ \nkey_\mathtt{deb}\ \nkey_\mathtt{cam}}]]},
\\
\tuple{\ND{\nkey_{\mathtt{cam}}; \integer_{\mathtt{cam}}; \KS{\emptyseq}}; \emptyseq},
\\
\tuple{\ND{\nkey_{\mathtt{bob}}; \integer_{\mathtt{bob}}; \KS{\emptyseq}}; \emptyseq},
\\
\tuple{\ND{\nkey_{\mathtt{amy}}; \integer_{\mathtt{amy}}; \KS{\emptyseq}}; \emptyseq}]
\end{array}
; \emptyseq; \rstream; \ep \rb
\]
\end{example}
}%%%%%%%%%%%%%%%%%%%%%%%%% end inlong %%%%%%%%%%%%%%%%%%%%%%%%%

While operations propagate through the graph, \htitle{Batch} and \htitle{Unbatch} can be flexibly applied.
% so that the same operation may be propagated in a batch in some graph neighborhoods, but not in others. 
The reader may notice that many task reduction rules, such as \plaintitle{Map} and \plaintitle{Fold}, are defined with a singleton stream unit (batch). This is because any batched stream unit can be unbatched first via \htitle{Unbatch}, realized, and then batched again via \htitle{Batch} for further propagation.

%\dnote{Furthermore, \plaintitle{} supports batched propagation. The original configuration in Example~\ref{exm:prop} is first reduced to the following according to a batching rule that we will introduce , 

% \begin{example}[Batching and Unbatching]
% Consider a configuration with the graph illustrated in Fig.\ref{fig:reordering}(a).
% Two instances of \htitle{Batch} can merge all three operations into a single stream unit. The stream unit can then, using only two instances of \plaintitle{Prop}, propagate to the node $\mathtt{bob}$.
% A single instance of \htitle{Unbatch} will then separate the first
% operation, $\mathtt{queryNode\ b}$, to be realized with an instance of \plaintitle{Fold}.
% \end{example}

%Before we introduce optimization rules on reordering, let us define a helper operator
%for map--fold function composition:
%

Reordering is supported by three rules. \htitle{ReorderD} says that two operations with disjoint target key lists can be reordered in the operation stream (without having any impact on the result). 

\begin{example}[Operation Reordering]
Imagine we have two mapping operations
which target two lists of disjoint keys:
$\fvs \mapsto \map\ \ep\ \KS{[\nkey_1, \nkey_2]}$
and
$\fvs' \mapsto \map\ \ep'\ \KS{[\nkey_3, \nkey_4]}$.
According to \htitle{ReorderD}, they may be swapped.
\end{example}

\htitle{ReorderRR} says that two $\fold$ operations can be reordered, as both are ``read'' in nature. Finally, \htitle{ReorderRW} shows a $\map$ operation and a $\fold$ operation may still be reordered even if they have overlapping targets. The insight behind is that a $\fold$ can ``skip ahead'' of a $\map$ if the former alters its folding function as applying the mapping function of the latter first. This rule relies on a helper operator for composing a mapping function and a folding function together,
%\begin{definition}[Map-Fold Composition]
%\label{def:composition}
%We define map--fold function composition operator
where $\func \dcomp{\aset{\nkey}} \func'$ is defined as 
$
\lambda x . \lambda y .
\func\ (\mathbf{if}\ \neone{x} \in \aset{\nkey}\ \mathbf{then}\ \func'\ x\ \mathbf{else}\ x)\ y
%\func\ (\func'\ x)\ y
$.
%\end{definition}
%
%
%For simplicity, \htitle{ReorderRW} only applies to the case when the two operations have identical targets. It is not challenging to extend the rule to consider two different but overlapping targets, as $\textbf{if}\dots\textbf{then}\dots\textbf{else}$ is encodable. 

% \begin{example}[Checked Operation Reordering]
% Consider Example~\ref{ex:superstepblend}.
% According to rule \htitle{ReorderRW}, the $\fold$ operation is able to skip ahead of the straggling $\map$ operation, if the mapping function is applied to any node that the $\map$ operation targets before passing it to the $\fold$ function.
% \end{example}

We support operation fusion with two rules. In \htitle{FuseM}, two $\map$ operations with identical target key lists can be fused into a single operation by composing their mapping functions. A special case of $\map$-$\map$ fusion is illustrated in \htitle{FuseMId}, where two mapping functions compose to one equivalent to the identity function. In that case, both operations can be removed all together from the operation streams: they ``cancel out'' on each other. %The $\circ$ operator is the composition function. The $\feq$ relation is standard equivalence relation of $\lambda$ calculus. We elide these definitions here. 

%We define $\feq$ as the function equality relation.
%\pnote{define $\feq$}
%\pnote{equivalence through alpha, eta, and maybe beta}

%\pnote{update rule names}
\begin{example}[Map-Map Fusion]
Consider the pair of operations $\mathtt{addRelationship\ b\ f}$ and $\mathtt{deleteRelationship\ b\ f}$
from Fig.~\ref{fig:serverclient}. In Example~\ref{ex:cancel}, it was discussed that the two may cancel each other out. Formally, this is demonstrated by the rule \htitle{FuseMId}, with the composition of the two mapping functions as follows being equivalent to the identity function:
\[
(\lambda x . \ND{\neone{x}; \netwo{x}; \nethree{x} \setsubtract \KS{[\nkey_\mathtt{fred}]}})
\circ
(\lambda x . \ND{\neone{x}; \netwo{x}; \nethree{x} \setadd \KS{[\nkey_\mathtt{fred}]}})
%
%\lambda x . \lambda y . \lambda z .
%((\lambda y . \lambda z . z \kssetminus \ep')\
% x\
% ((\lambda x . \lambda y . \lambda z . y)\ x\ y\ z)\
% ((\lambda x . \lambda y . \lambda z . \ep' \kscons z)\ x\ y\ z)
%)
\]
\end{example}

Finally, \htitle{Reuse} demonstrates the idea behind reusing. Observe here that the targets of the two $\fold$ operations form a subsetting relationship. With this rule, the second $\fold$ operation does not need to ``redo'' the computation that is to be carried out by the first $\fold$. The algebraic requirement of commutativity should become clear through one counterexample. 

\begin{example}[Reuse in the Presence of Non-Commutative Folding Functions]
Consider the following configuration with two nodes and two delayed operations:
\[
\tuple{[\tuple{\inode_1; [[\fvs \mapsto \fold\ \func\ \inode_0\ \KS{[\nkey_2]}], [\fvs' \mapsto \fold\ \func\ \inode_0\ \KS{[\nkey_1, \nkey_2, \nkey_3]}]]}, \tuple{\inode_2; \emptyseq}]; \emptyseq; \emptyset; 0}
\]
where $\func = \lambda x . \lambda y . \ND{\neone{y}; \netwo{x} - \netwo{y}; \nethree{y}}$,
$\inode_i = \ND{\nkey_i; i; \KS{\emptyseq}}$ for $i = 0,1,2$.
Without any TLO, we know $\fvs'$ eventually should compute to 
$\ND{\nkey_0; 1; \KS{\emptyseq}}$.
If \htitle{Reuse} were applied (ignoring the commutativity requirement on $\func$), 
the configuration would one-step reduce to
\[
\tuple{[\tuple{\inode_1; [[\fvs \mapsto \fold\ \func\ \inode_0\ \KS{[\nkey_2]}], [\fvs' \mapsto \fold\ \func\ (\fromserver\ \fvs)\ \KS{[\nkey_1, \nkey_3]}]]}, \tuple{\inode_2; \emptyseq}]; \emptyseq; \emptyset; 0}
\]
which would compute an incorrect value
$\ND{\nkey_0; -1; \KS{\emptyseq}}$ for $\fvs'$.
%would then be computed for the result of the second $\fold$ operation.
\end{example}

%\pnote{write about MQO}

Our TLO support aims at demonstrating its principles behind, and its role in constructing continuous graph processing. Despite the diversity of TLOs --- from batching, reordering, fusing, to reusing --- they all share the spirit of ``short-circuiting'' the operation(s) before they are fully realized, leading to reduced need for computation.

\section{\systemprefix{} Type System}
\label{sec:safety}

\begin{figure}[t]
\centering

\scalebox{\scalemath}{
$\begin{array}{@{}l@{\ \ }l@{\ \ }lr}

\tau
& \defassign &
\mathsf{int}
\mid \mathsf{key}
\mid \mathsf{future}[\tau]
\mid \mathsf{kl}
\mid \mathsf{node}
%\mid \mathsf{bottom}
\mid \tau \cdfunc \tau
& \textit{type}
\\

\cdstatus
& \defassign &
\dirtyexp
\mid \cleanexp
& \textit{emittability}
\\

\Gamma
& \defassign &
\aseq{\VAR : \tau}
& \textit{typing environment}
\\

\end{array}$
}

\begin{mathpar}
\scalebox{\scalemath}{$
\inferrule*[left=\ttitle{Int}]{}{
   \Gamma \vdash \integer : \mathsf{int} \bs \cleanexp
}
$}

\scalebox{\scalemath}{$
\inferrule*[left=\ttitle{Key}]{}{
   \Gamma \vdash \nkey : \mathsf{key} \bs \cleanexp
}
$}

\scalebox{\scalemath}{$
\inferrule*[left=\ttitle{KS}]{
   \aset{\Gamma \vdash \ep : \mathsf{key} \bs \cdstatus}
}{
   \Gamma \vdash \KS{\aseq{\ep}} : \mathsf{kl} \bs {\textstyle \bigvee \aset{\cdstatus}}
}
$}
% 
% \inferrule*[left=\ttitle{KSAll$_1$}]{}{
%    \Gamma \vdash \KSstar : \mathsf{kl} \bs \cleanexp
% }
% 
% \inferrule*[left=\ttitle{KSAll$_2$}]{}{
%    \Gamma \vdash \KShollowstar : \mathsf{kl} \bs \cleanexp
% }

\scalebox{\scalemath}{$
\inferrule*[left=\ttitle{Node}]{
   \Gamma \vdash \ep : \mathsf{key} \bs \cdstatus
\\ \Gamma \vdash \ep' : \mathsf{int} \bs \cdstatus'
\\ \Gamma \vdash \ep'' : \mathsf{kl} \bs \cdstatus''
}{
   \Gamma \vdash \ND{\ep; \ep'; \ep''} : \mathsf{node} \bs (\cdstatus \lor \cdstatus' \lor \cdstatus'')
}
$}

\scalebox{\scalemath}{$
\inferrule*[left=\ttitle{Var}]{}{
   \Gamma \vdash x : \Gamma\{x\} \bs \cleanexp
}
$}

\scalebox{\scalemath}{$
\inferrule*[left=\ttitle{Claim}]{
   \Gamma \vdash \ep : \mathsf{future}[\tau] \bs \cdstatus
}{
   \Gamma \vdash \fromserver\ \ep : \tau \bs \cdstatus
}
$}

\scalebox{\scalemath}{$
\inferrule*[left=\ttitle{ENode1}]{
   \Gamma \vdash \ep : \mathsf{node} \bs \cdstatus
}{
   \Gamma \vdash \neone{\ep} : \mathsf{key} \bs \cdstatus
}
$}

\scalebox{\scalemath}{$
\inferrule*[left=\ttitle{ENode2}]{
   \Gamma \vdash \ep : \mathsf{node} \bs \cdstatus
}{
   \Gamma \vdash \netwo{\ep} : \mathsf{integer} \bs \cdstatus
}
$}

\scalebox{\scalemath}{$
\inferrule*[left=\ttitle{ENode3}]{
   \Gamma \vdash \ep : \mathsf{node} \bs \cdstatus
}{
   \Gamma \vdash \nethree{\ep} : \mathsf{kl} \bs \cdstatus
}
$}

\scalebox{\scalemath}{$
\inferrule*[left=\ttitle{KSA}]{
   \Gamma \vdash \ep : \mathsf{kl} \bs \cdstatus
\\ \Gamma \vdash \ep' : \mathsf{kl} \bs \cdstatus'
}{
   \Gamma \vdash \ep \setadd \ep' : \mathsf{kl} \bs (\cdstatus \vee \cdstatus')
}
$}

\scalebox{\scalemath}{$
\inferrule*[left=\ttitle{KSS}]{
   \Gamma \vdash \ep : \mathsf{kl} \bs \cdstatus
\\ \Gamma \vdash \ep' : \mathsf{kl} \bs \cdstatus'
}{
   \Gamma \vdash \ep \setsubtract \ep' : \mathsf{kl} \bs (\cdstatus \vee \cdstatus')
}
$}

\scalebox{\scalemath}{$
\inferrule*[left=\ttitle{Add}]{
   \Gamma \vdash \ep : \mathsf{int} \bs \cdstatus
}{
   \Gamma \vdash \toserver\ \addn\ \ep : \mathsf{future}[\mathsf{key}] \bs \dirtyexp
}
$}

\scalebox{\scalemath}{$
\inferrule*[left=\ttitle{Map}]{
   \Gamma \vdash \ep : \mathsf{node} \cleanfunc \mathsf{node} \bs \cdstatus
\\ \Gamma \vdash \ep' : \mathsf{kl} \bs \cdstatus'
}{
   \Gamma \vdash \toserver\ \map\ \ep\ \ep' : \mathsf{future}[\mathsf{int}] \bs \dirtyexp
}
$}

\scalebox{\scalemath}{$
\inferrule*[left=\ttitle{Fold}]{
   \Gamma \vdash \ep : \mathsf{node} \cleanfunc \mathsf{node} \cleanfunc \mathsf{node} \bs \cdstatus
\\ \Gamma \vdash \ep' : \mathsf{node} \bs \cdstatus'
\\ \Gamma \vdash \ep'' : \mathsf{kl} \bs \cdstatus''
}{
   \Gamma \vdash \toserver\ \fold\ \ep\ \ep'\ \ep'' : \mathsf{future}[\mathsf{node}] \bs \dirtyexp
}
$}

\scalebox{\scalemath}{$
\inferrule*[left=\ttitle{Abs}]{
   \Gamma \concat [\VAR : \tau] \vdash \ep : \tau' \bs \cdstatus
\\
}{
   \Gamma \vdash \lambda \VAR : \tau . \ep : \tau \xrightarrow{\cdstatus} \tau' \bs \cleanexp
}
$}

\scalebox{\scalemath}{$
\inferrule*[left=\ttitle{App}]{
   \Gamma \vdash \ep : \tau \cdfunc \tau' \bs \cdstatus'
\\ \Gamma \vdash \ep' : \tau \bs \cdstatus''
}{
   \Gamma \vdash \ep\ \ep' : \tau' \bs (\cdstatus \lor \cdstatus' \lor \cdstatus'')
}
$}

\scalebox{\scalemath}{$
\inferrule*[left=\ttitle{Fix}]{
   \Gamma \vdash \ep : \tau \cdfunc \tau \bs \cdstatus'
}{
   \Gamma \vdash \mathbf{fix}\ \ep : \tau \bs (\cdstatus \lor \cdstatus')
}
$}
%
%
%\inferrule*[left=\ttitle{Bottom}]{}{
%   \Gamma \vdash \bot : \mathsf{bot} \bs \cleanexp
%}
%
%\inferrule*[left=\ttitle{Sub}]{}{
%   \mathsf{bot} \subtype \tau
%}
%
%\inferrule*[left=\ttitle{Sub}]{
%   \Gamma \vdash \ep : \tau \bs \cdstatus
%\\ \tau \subtype \tau'
%}{
%   \Gamma \vdash \ep : \tau' \bs \cdstatus
%}
\end{mathpar}

\caption{\systemprefix{} Type System}
\label{fig:safety}

\end{figure}

Fig.~\ref{fig:safety} defines a type system for \systemprefix{}, where typing judgement
$\Gamma \vdash \ep : \tau \bs \cdstatus$ says that given typing environment
$\Gamma$, expression $\ep$ has type $\tau$ with \emph{emittability} $\cdstatus$.
Metavariable $\cdstatus$ ranges over booleans, where a true value ($\dirtyexp$) indicates the
expression may emit an operation whereas a false value ($\cleanexp$) indiates it must not.
Operator $\Gamma\{\VAR\}$ is defined as $\tau$ where $\VAR': \tau$ is the right most occurrence
in $\Gamma$ such that $\VAR = \VAR'$.

Types are either a key type $\mathsf{key}$, a payload type $\mathsf{int}$, a key list type $\mathsf{kl}$, a node type $\mathsf{node}$, a future type $\mathsf{future}[\tau]$ where $\tau$ is the type of the result represented by the future, or a function type $\tau \cdfunc \tau$. In the last form, emittability $\cdfunc$ is the effect
of the function, which we will explain next. When a function has type $\tau \dirtyfunc \tau$, we informally
say that the function is \emph{latently emittable}. 

\subsection{Phase Distinction}

The primary goal of the type system is to enforce phase distinction: whereas 
the evaluation of an expression at the frontend is unrestricted, the evaluation at the backend
cannot lead to an operation emission. We establish phase distinction through a type-and-effect
system. Observe that an operation emission could only happen at the backend if the functions that serve as the arguments of operations were latently emittable. As a result, the key to enforcing phase distinction is to make
sure these arguments are not latently emittable. Note that in our type system, both \ttitle{Map} and \ttitle{Fold}
ensure that their argument functions --- be it the mapping function or the folding function --- have function types that are not latently emittable. 

The majority of the rules in the type system address the emittability of an expression. Predictably, \ttitle{Add}, \ttitle{Map}, \ttitle{Fold} all have the emittability of $\dirtyexp$. Emittability is disjunctive, as shown in rules such as \ttitle{App}, \ttitle{Node}, and \ttitle{KS}. 
%We call any expression with emittability $\dirtyexp$ an \emph{emitting expressions}.
%
%Emitting expressions are not of interest by themselves.
%However, emitting expressions are necessary to track
%because a function with an expression body that is emitting
%cannot be used as a function argument to $\map$ or $\fold$ operations.
%We call such functions \emph{emitting functions},
%and have type $\dirtyfunc$.
%Non-emitting functions have type $\cleanfunc$
% Rule \ttitle{Abs$_1$} and \ttitle{Abs$_2$}
% says that a function has an emitting function type if its
% body either emits or itself has an emitting function type.
% %
% Rule \ttitle{App} says that the emittability of
% function application is the disjunction of the emittability
% between the function's emittability, the first expression's
% emittability, and the second expression's emittability.
% %
% Rule \ttitle{Fix} says that the $\mathbf{fix}$ expression
% inherits its emittability from its argument.
% %
% Rule \ttitle{Add}, \ttitle{Map}, and \ttitle{Fold}
% say that all $\toserver\ \processor$ operations
% are emitting and all function arguments to the operations
% must be non-emitting.
% %
% The rest of the functions are self explanatory.
\inshort{To revisit Example~\ref{ex:emit}, the program will not type check, because expression 
$\mathtt{mapVal\ g\ [k]}$ would violate phase distinction.}

\inlong{%%%%%%%%%%%%%%%%%%%%%%%%% begin inlong %%%%%%%%%%%%%%%%%%%%%%%%%
Let us revisit Example~\ref{ex:emit}, and the program will not type check. 
\begin{example}[Emittability]
%
%The desugared version of
%the line
%\pnote{I changed this example to only desugar the let-in, and nothing else. I think if we start desugaring the target key list then we'd also have to desugar the $\mathtt{mapVal}$, which would mean the function $\mathtt{f}$ would also have to change.  I think all of that is too much of a distraction, and actually we probably could change the example to not even desugar let-in either.}
In Example~\ref{ex:emit}, expression 
$\mathtt{mapVal\ g\ [k]}$ would violate phase distinction.
The following is the last step of a would-be derivation, where \textbf{let}-\textbf{in} is encoded:
%\dnote{I'm not sure whether this is how let-in is encoded. It should be something like a function application.}
%\pnote{Not sure I follow. There is a function application already below}
%attempting to establish \ttitle{Map}
%using the line's desugared and partially substituted form:
% With its desugared and partially substituted form
% \[
% \toserver\ \map\ (\mathtt{\lambda nk . \lambda np . \lambda na . ((\lambda \_ . \mathtt{np})\ (\toserver\ \map\ f\ fA\ \{key\})))\ fA}\ \KS{[key]}
% \]
\begin{mathpar}
% \inferrule*[right=\ttitle{Map}]{
%    \inferrule*[right=\ttitle{Abs}]{
%       \inferrule*[right=\ttitle{Abs}]{
%          \inferrule*[right=\ttitle{Abs}]{
%             %\inferrule*[right=\ttitle{App}]{
%             %}{
%             %   \Gamma \concat [\mathtt{np} : \mathsf{key}, \mathtt{np} : \mathsf{int}, \mathtt{na} : \mathsf{kl}] \vdash \mathtt{(\lambda \_ . \mathtt{np})\ (\toserver\ \map\ f\ fA\ \{key\})} \bs \cleanexp
%             %}
%          }{
%             \Gamma \concat [\mathtt{np} : \mathsf{key}, \mathtt{np} : \mathsf{int}] \vdash \mathtt{\lambda na : \mathsf{kl} . ((\lambda \_ . \mathtt{np})\ (\toserver\ \map\ f\ fA\ \{key\}))} : \mathsf{kl} \cleanfunc \mathsf{int} \bs \cleanexp
%          }
%       }{
%          \Gamma \concat [\mathtt{np} : \mathsf{key}] \vdash \mathtt{\lambda np : \mathsf{int}. \lambda na : \mathsf{kl} . ((\lambda \_ . \mathtt{np})\ (\toserver\ \map\ f\ fA\ \{key\}))} : \mathsf{int} \cleanfunc \mathsf{kl} \cleanfunc \mathsf{int} \bs \cleanexp
%       }
%    }{
%       \Gamma \vdash \mathtt{\lambda nk : \mathsf{key}. \lambda np : \mathsf{int}. \lambda na : \mathsf{kl} . ((\lambda \_ . \mathtt{np})\ (\toserver\ \map\ f\ fA\ \{key\}))} : \mathsf{key} \cleanfunc \mathsf{int} \cleanfunc \mathsf{kl} \cleanfunc \mathsf{int} \bs \cleanexp
%    }
% }{
%    \Gamma \vdash
%    \toserver\ \map\ (\mathtt{\lambda nk . \lambda np . \lambda na . ((\lambda \_ . \mathtt{np})\ (\toserver\ \map\ f\ fA\ \{key\})))\ fA}\ \KS{[key]}
%    : \mathsf{future}[\mathsf{int}]\bs \dirtyexp
% }
% 
\scalebox{\scalemath}{$
\inferrule*[right=\ttitle{Map}]{
   \inferrule*[right=\ttitle{Abs}]{
   }{
      \Gamma \vdash (\mathtt{\lambda \tuple{\_;\ payload;\ \_} : \mathsf{node} . ((\lambda \_ . \mathtt{payload})\ (\toserver\ \mathtt{mapVal}\ f\ [k]))}) : \mathsf{node} \cleanfunc \mathsf{node} \bs \cleanexp
   }
}{
   \Gamma \vdash
   \toserver\ \mathtt{mapVal}\ (\mathtt{\lambda \tuple{\_;\ payload;\ \_} : \mathsf{node} . ((\lambda \_ . \mathtt{payload})\ (\toserver\ \mathtt{mapVal}\ f\ [k]))})\ [\mathtt{k}]
   : \mathsf{future}[\mathsf{int}] \bs \dirtyexp
}
$}
\end{mathpar}
%The judgement is stuck.
We however cannot find a derivation for the expression
$\mathtt{(\lambda \_ . \mathtt{x})\ (\toserver\ mapVal\ f\ [k])}$ with type $\mathsf{node}$ and emittability $\cleanexp$.
%\dnote{a pass more}

\end{example}
}%%%%%%%%%%%%%%%%%%%%%%%%% end inlong %%%%%%%%%%%%%%%%%%%%%%%%%

Beyond emittablity, our type system also performs sanity checks to ensure the types of operation arguments match that of the graph node representation. For example, the first argument of the $\fold$ operation has type $\mathsf{node} \cleanfunc \mathsf{node} \cleanfunc \mathsf{node}$, which says that the function is a binary function over nodes.

%\dnote{Key not found as a non-goal}

%For a description on runtime typing, including its optimization, please refer to our accompanying technical report.

%\inlong{%%%%%%%%%%%%%%%%%%%%%%%%% begin inlong %%%%%%%%%%%%%%%%%%%%%%%%%

\subsection{Runtime Typing}

Our type system can be implemented either as a static system or a dynamic system. The former is useful with the ``closed world'' assumption: the entire processing operations are known before the program starts. The latter is more appropriate with the ``open world'' (see Sec.~\ref{sec:syntax}), where the forms of operations and their arguments may not be known until run time.
 %In practice, the interactive nature of this latter use scenario dictates that the functions associated with graph-processing operations --- such as the first two arguments of $\map$ and the first argument of $\fold$ in \systemprefix{} --- are relatively simple, and dynamic typing should only incur a small overhead.
%\dnote{not sure we want this last sentence}
In this section, we define runtime typing in Fig.~\ref{fig:welltyped}. (For static typing, the runtime typing rules are useful for the proof.) % For dynamic typing, these rules will define the type checking process at run time.
%Runtime typing is defined in Fig.~\ref{fig:welltyped}. 
Judgment $\Gamma \vdash_\mathtt{c} \config: \tau \bs \cdstatus$ says configuration $\config$ has type $\tau$ with emittability $\cdstatus$ under typing environment $\Gamma$. To support runtime typing of future values, we further extend the typing environment to include the form of $\fvs: \tau$, and further introduce future value typing as follows:
\begin{mathpar}
%\scalebox{\scalemath}{$
\inferrule*[left=\ttitle{Future}]{}{
   \Gamma \vdash \fvs : \Gamma\{\fvs\}
}
%$}
\end{mathpar}

Configuration typing is defined in \rttitle{Configuration}. It requires all keys in the graph to be distinct, and all future labels are distinct. %Note that the latter may spread out in the operation stream, the graph streamlets, and the result store. Furthermore, 
A well-typed configuration must ensure (i) all values in the result store are well-typed; (ii) the backend graph runtime is well-typed; (iii) all operations in the top-level operation stream are well-typed; and finally (iv) the frontend expression is well typed. Backend graph runtime typing is defined by judgment $\Gamma \vdash_{\mathtt{b}} \data: \Gamma'$, which says that backend graph $\data$ is well-typed under typing judgment $\Gamma$, and all future values in the graph are also well-typed as in $\Gamma'$. Operation stream typing, stream unit typing, and station typing are defined by judgments $\vdash_\mathtt{o}$, $\vdash_\mathtt{u}$, and $\vdash_\mathtt{d}$, whose forms are similar to $\vdash_\mathtt{b}$.

The most interesting observation of runtime typing is that \emph{order matters}: we first type the result store in configuration typing, then type the \emph{last} station and go backward one by one until the first station is typed, and finally type the frontend expression. 
%Within each stream, we type the head element first (\rttitle{StreamUnit}). 
%(1) in configuration typing, we first type the values in the result store, and prepare the typing environment for graph runtime typing accordingly. (2) In graph runtime typing, we start with typing the \emph{last} station, as seen in \rttitle{Graph}, and going backward one by one until the first station is typed. Graph typing prepares the typing environment for top-level operation stream typing. (3) In operation stream (unit) typing, we type the head operation first (the oldest operation in the stream), as seen in \rttitle{StreamUnit}, and go one by one until the newest operation in the stream is typed. (4) Finally, frontend expression typing relies on the typing environment prepared by top-level operation stream typing. 
This ``relay race'' nature of typing reflects the potential chronological dependencies of operations: a later emitted operation may contain a future label of an operation emitted before it. % in its argument or inside the function body of its argument, 
%so that typing this operation must rely on the typing of all operations that have been emitted before it. 

At the first glance, the definition here may appear expensive to implement in a dynamic typing system. %: it relies on the typing of the result store, the backend graph runtime, and the top-level operation stream, before the frontend expression can be typed. 
In practice, configuration typing defined in  Fig.~\ref{fig:welltyped} can be optimized with an efficient \emph{cached typing environment} design: whenever an operation is emitted, the dynamic typing system can immediately type it, and store the future value and its corresponding type in the cached type environment. 
%As configuration typing relies on the temporal ordering of typing operations, this caching 
%This design ensures that when a frontend expression is typed, all the future values of previously emitted operations already appear in the cached typing environment. 
As a result, dynamic typing does not need to resort to the typing of any component of the configuration other than the frontend at all. 

% \begin{definition}[Well-Formed Graph]
% Predicate $\wf(\data)$ holds iff
% \begin{enumerate}
% \item (Key Uniqueness) if $\data = \data_1 \concat \tuple{\tuple{\nkey_1; \ep_1; \ep_1'}; \ostream_1} \cons \data_2 \cons \tuple{\tuple{\nkey_2; \ep_2; \ep_2'}; \ostream_2} \cons \data_3$
%       for some $\data_1$, $\data_2$, and $\data_3$, then $\nkey_1 \neq \nkey_2$;
% \item (Label Uniqueness) If $\data = [\tuple{\inode_1; \ostream_1}, \ldots, \tuple{\inode_m; \ostream_m}]$
%       then $\dom(\ostream_1)$, \ldots, $\dom(\ostream_m)$ are all disjoint;
% %\item (Well-Typed Program) $\ep$ is well-typed.
% \end{enumerate}
% \end{definition}
% 
% \begin{definition}[Well-Formed Configuration]
% %
% Predicate $\wf(\tuple{\data; \ostream; \rstream; \ep})$ holds iff
% \begin{enumerate}
% \item (Graph Well-Formedness) $\wf(\data)$;
% \item (Label Uniqueness) If $\data = [\tuple{\inode_1; \ostream_1}, \ldots, \tuple{\inode_m; \ostream_m}]$
%       then $\dom(\ostream_1)$, \ldots, $\dom(\ostream_m)$, $\dom(\ostream)$, $\dom(\rstream)$ are all disjoint.
% %\item (Top-Level Dependent Result Resolution) $\fv(\ostream) \cap \dom(\rstream) = \emptyseq$
% \end{enumerate}
% \end{definition}

\begin{figure}[t]
\centering

\begin{mathpar}
\scalebox{\scalemath}{$
\inferrule*[left=\rttitle{Configuration}]{
   \data = \aseq{\tuple{\ND{\nkey; \ep'; \ep''}; \ostream'}}^m
\\ \nkey_1, \ldots, \nkey_m \ \textrm{are all distinct}
\\ \dom(\ostream), \dom(\rstream), \dom(\ostream'_1), \ldots, \dom(\ostream'_m) \ \textrm{are all distinct}
\\ \rstream = \aset{\fvs \xmapsto{\K} \val}
\\ \aset{\emptyseq \vdash \val : \tau \bs \cleanexp}
\\ \Gamma \concat \aseq{\fvs : \tau} \vdash_\mathtt{b} \data : \Gamma'
\\ \Gamma' \vdash_\mathtt{o} \ostream : \Gamma''
\\ \Gamma'' \vdash \ep : \tau' \bs \cdstatus
}{
   \Gamma \vdash_\mathtt{c} \tuple{\data; \ostream; \rstream; \ep} : \tau' \bs \cdstatus
}
$}

\scalebox{\scalemath}{$
\inferrule*[left=\rttitle{Graph}]{
   \Gamma \vdash_\mathtt{d} \tuple{\inode; \ostream} : \Gamma'
\\ \Gamma \concat \Gamma' \vdash_\mathtt{b} \data : \Gamma''
}{
   \Gamma \vdash_{\mathtt{b}} \data \concat [\tuple{\inode; \ostream}] : \Gamma''
}
$}

\scalebox{\scalemath}{$
\inferrule*[left=\rttitle{GraphEmpty}]{}{
   \Gamma \vdash_\mathtt{b} \emptyseq : \Gamma
}
$}

\scalebox{\scalemath}{$
\inferrule*[left=\rttitle{Station}]{
   \Gamma \vdash \ep : \mathsf{key} \bs \cleanexp
\\ \Gamma \vdash \ep' : \mathsf{int} \bs \cleanexp
\\ \Gamma \vdash \ep'' : \mathsf{kl} \bs \cleanexp
\\ \Gamma \vdash_\mathtt{o} \ostream : \Gamma'
}{
   \Gamma \vdash_{\mathtt{d}} \tuple{\ND{\ep; \ep'; \ep''}; \ostream} : \Gamma'
}
$}

\scalebox{\scalemath}{$
\inferrule*[left=\rttitle{Stream}]{
   \Gamma \vdash_\mathtt{u} \batchgroup_1 \concat \ldots \concat \batchgroup_u : \Gamma'
}{
   \Gamma \vdash_\mathtt{o} \aseq{\batchgroup} : \Gamma'
}
$}

\scalebox{\scalemath}{$
\inferrule*[left=\rttitle{StreamUnit}]{
   \Gamma \vdash \toserver\ \processor : \tau \bs \cleanexp
\\ \Gamma \concat [\fvs : \tau] \vdash_\mathtt{u} \batchgroup : \Gamma'
}{
   \Gamma \vdash_{\mathtt{u}} \fvs \mapsto \processor \cons \batchgroup : \Gamma'
}
$}

\scalebox{\scalemath}{$
\inferrule*[left=\rttitle{StreamUnitEmpty}]{}{
   \Gamma \vdash_\mathtt{u} \emptyseq : \Gamma
}
$}
\end{mathpar}

\caption{Runtime Typing}
\label{fig:welltyped}

\end{figure}

%}%%%%%%%%%%%%%%%%%%%%%%%%% end inlong %%%%%%%%%%%%%%%%%%%%%%%%%

%\pnote{talk about $\Gamma$ extension}

%%%%%%%%%%%%%%%%%%%%%%%%%%%%%%%%%%%%%%%%%%%%%%%%%%%%%%%%%
%%%%%%%%%%%%%%%%%%%%%%%%%%%%%%%%%%%%%%%%%%%%%%%%%%%%%%%%%
%%%%%%%%%%%%%%%%%%%%%%%%%%%%%%%%%%%%%%%%%%%%%%%%%%%%%%%%%
%%%%%%%%%%%%%%%%%%%%%%%%%%%%%%%%%%%%%%%%%%%%%%%%%%%%%%%%%
%%%%%%%%%%%%%%%%%%%%%%%%%%%%%%%%%%%%%%%%%%%%%%%%%%%%%%%%%
%%%%%%%%%%%%%%%%%%%%%%%%%%%%%%%%%%%%%%%%%%%%%%%%%%%%%%%%%
%%%%%%%%%%%%%%%%%%%%%%%%%%%%%%%%%%%%%%%%%%%%%%%%%%%%%%%%%
%%%%%%%%%%%%%%%%%%%%%%%%%%%%%%%%%%%%%%%%%%%%%%%%%%%%%%%%%
%%%%%%%%%%%%%%%%%%%%%%%%%%%%%%%%%%%%%%%%%%%%%%%%%%%%%%%%%
%%%%%%%%%%%%%%%%%%%%%%%%%%%%%%%%%%%%%%%%%%%%%%%%%%%%%%%%%

\section{Meta-Theory}
\label{sec:properties}

%In this section, we establish important properties of \systemprefix{}. The most important result on correctness is that our calculus, despite inherent non-determinism (\S~\ref{sec:reduct}), can produce the same results as a deterministic baseline, which we now introduce in \S~\ref{subsec:eager}. 

\begin{figure}[t]
\centering

\scalebox{\scalemath}{
$\begin{array}{@{}l@{\ \ }l@{\ \ }lr}

\eagerfE
& \defassign &
\tuple{\eagerdata; \emptyseq; \rstream; \cE}
& \textit{eager frontend context}
\\

\eagerbE
& \defassign &
\tuple{\bullet; \emptyseq; \rstream; \ep}
& \textit{eager backend context}
\\

\eagertE
& \defassign &
\eagerbE[\eagerdata \concat \bullet \concat \eagerdata]
%\ \ \textrm{where}\ \dom(\eagertE) \in \{ \tuple{\eagerinode; [[\fvs \mapsto \strictprocessor]]} \cons \eagerdata \}
& \textit{eager task context (redex load-free)}
\\

\eagerlE
& \defassign &
% \tE[[\tuple{\tuple{\nkey; \bullet; \ep}; \emptyseq}]]
\eagerbE[\eagerdata \concat \tuple{\cE; \ostream} \cons \eagerdata]
& \textit{eager load context}
% \\ & \mid &
% % \tE[[\tuple{\tuple{\nkey; \integer; \bullet}; \emptyseq}]]
% \eagerbE[\eagerdata \concat \tuple{\tuple{\nkey; \integer; \bullet}; \ostream} \cons \eagerdata]
\\ & \mid &
\eagerbE[\eagerdata \concat \tuple{\eagerinode; [[\fvs \mapsto \fold\ \func\ \cE\ \K]]} \cons \eagerdata]
\\

\inshort{
\eagerdata
& \defassign &
\aseq{\tuple{\eagerinode; \emptyseq}}
& \textit{dry backend}
\\
}

% \eagerinode
% & \defassign &
% \tuple{\nkey; \integer; \K}
% % \tuple{\nkey; \integer; \E}
% & \textit{solid node}
% \\

\end{array}$
}

\caption{Evaluation Context for Eager Processing}
\label{fig:eagercontexts}

\end{figure}

%\section{\eagerlang{}: Eager Processing in Streaming Graphs}
%\label{sec:eagerlang}

%\pnote{don't need this section anymore?} \dnote{correct, except that we should try to move this example to the meta-theory section to show how the evaluation context is restricting the reduction} 

\inlong{%%%%%%%%%%%%%%%%%%%%%%%%% begin inlong %%%%%%%%%%%%%%%%%%%%%%%%%

\subsection{Baseline: Deterministic Eager Graph Processing}
\label{subsec:eager}

Our baseline of choice is deterministic eager processing, the conventional form of graph processing where operations are handled one at a time (see \S~\ref{subsec:ingraph}).
%First, let us introduce a number of simple definitions:
%\begin{definition}[Dry Stream and Dry Backend]
%We say an operation stream $\ostream$ is dry if $\ostream = \emptyset$. We say the backend (graph or graph neighborhood) $\data$ is dry if %every streamlet in $\data$ is dry. Formally, we use $\eagerdata$ to represent a dry backend, defined in Fig.~\ref{fig:eagercontexts}. 
%\end{definition}
%\begin{definition}[Solid Node and Solid Backend]
%We say a node $\inode$ is solid if both its payload and its adjacency list are values. We say the backend (graph or graph neighborhood) is solid if the first node in the backend is solid, its associated streamlet is a singleton 
% the initialization expression of the $\fold$ expression in its (singleton) streamlet. 
%$\ostream = \emptyset$. We say the backend (graph or graph neighborhood) $\data$ is dry if every streamlet in $\data$ is dry. Formally, we %use $\eagerdata$ to represent a dry backend, defined in Fig.~\ref{fig:eagercontexts}. 
%\end{definition}
Instead of coming up with another reduction system to capture the essence of eager data processing, we now define this baseline by utilizing the same reduction system described in \S~\ref{sec:highlevelsemantics}. The key insight is that eager processing can be viewed as a restrictive form of \systemprefix{}, where restrictions are placed to turn a non-deterministic execution semantics into a deterministic one. Specifically: 

\begin{itemize}
    \item \textbf{R1}: A frontend reduction may only happen when both the top-level operation stream and in-graph operation stream are \emph{dry}, i.e., they do not contain any operation. This is aligned with our intuition of ``one at a time'': the frontend may only emit an operation when no other operations are being processed, and upon doing so, the frontend must wait until its result is placed into the result stream. 
%    \item \textbf{R2}: A task reduction or a load reduction inside the backend can only happen when there is no operation in the top-level operation stream.
%    \dnote{Is this one needed? Now I think R1 would make sure there is at most one operation in the stream. Does this mean if there is something in the top-level operation stream, we can't do anything with in-graph operation stream anyways? }%\dnote{I'm a little confused because, are START and ADD tasks?} 
%    \pnote{This one probably isn't strictly needed. If there is something in the top-level operation stream then that implies that there is nothing to do in-graph}
    \item \textbf{R2}: The backend can only perform a task reduction if its redex station neighborhood is \emph{load free}, i.e., the first station in the neighorhood consists of a node value and a (singleton) streamlet whose operation consists of arguments that are all values. In other words, no load reduction can happen inside the station. 
%    \item \textbf{R4}: With each node, the backend must evaluate the three forms of loads in order: the node payload, the node adjacency list, and if applicable, the initialization expression of the $\fold$ expression in its (singleton) streamlet.
%    \dnote{Is this requirement still needed? Now map only changes the node, and fold only changes the streamlet. }
%    \pnote{i don't think this is needed}
\end{itemize}

In \S~\ref{sec:metatheory:properties}, we will rigorously establish that \systemprefix{} under these restrictions does have a deterministic semantics. 
Interestingly, these restrictions can be embodied by redefining the evaluation contexts, and there is \emph{no need to alter any reduction rules}. %This approach is attractive because it not only simplifies our semantic specification, but more simplifies the proof. 
Rigorously, we represent deterministic eager processing as the $\eagerreduce$ reduction relation, defined as identical as the $\configreduce$ we introduced in Fig.~\ref{sec:highlevelsemantics}, except that the $\fE$, $\bE$, $\tE$, $\lE$ evaluation contexts are replaced with $\eagerfE$, $\eagerbE$, $\eagertE$, $\eagerlE$ evaluation contexts in Fig.~\ref{fig:eagercontexts}.
We use $\eagerreducestar$
to represent the reflexive and transitive
closure of $\eagerreduce$.
For the eager task context $\eagertE$, we further require any redex that fills its hole (more rigorously, any element in the domain of the context fulfillment function) is load-free. \textbf{R1} is enforced by the definition of $\eagerfE$, and \textbf{R2} is enforced by the definition of $\eagertE$. % \textbf{R4} is enforced by the definition of $\eagerlE$. \dnote{update}

}%%%%%%%%%%%%%%%%%%%%%%%%% end inlong %%%%%%%%%%%%%%%%%%%%%%%%%

% The \eagerlang{} provides a minimal core on streaming graphs
% and illustrates the \emph{eager} aspect of propagation
% and processing.
% Eager graph processing is defined by relation
% $\data, \batchgroup  \ereduce \data', \rstream$
% as in Fig~\ref{fig:opsem:eager}.
% It says that the graph $\data$ one-step reduces to
% $\data'$ when being processed with an operation stream $\ostream$ in the eager processing mode,
% and it produces the result stream $\rstream$.
% %
% Let us first discuss the
% operation realization rules.

% In eager processing,
% the issuance of an operation must be one by one. % To be consistent with lazy data processing however
% (so that we can formally relate the two systems),
% we choose to define the relation over the stream instead of one operation/result.
% \plaintitle{$\rightarrow$} serves as the bridge between this
% general multi-issuance form and all other rules,
% which require one operation/result to be processed at a time.

% The essence of eager processing is that processing each
% operation entails the traversal of the
% potentially large graph \emph{before} another operation can be processed.
\inlong{%%%%%%%%%%%%%%%%%%%%%%%%% begin inlong %%%%%%%%%%%%%%%%%%%%%%%%%
We now demonstrate the deterministic nature of eager processing through an example:

\begin{example}[Eager Processing]
\label{ex:eagerprocessing}
%
% Consider the configuration ($\config$)
% presented in Example~\ref{ex:lazyprocessing}.
Consider the following configuration with a graph with two nodes
and two emit expression at the frontend.
\[
\tuple{[\tuple{\inode_1; \emptyseq}, \tuple{\inode_2; \emptyseq}]; \emptyseq; \emptyset; \ep}
\]
where $\K = \KS{[\nkey_1]}$,
$\K' = \KS{[\nkey_2]}$,
$\inode_i = \ND{\nkey_i; \ep_i; \ep_i'}$ for $i = 1, 2$,
$\ep = (\lambda \_ . \ep')\ (\toserver\ \map\ \func_1\ \K; \ep')$
and $\ep' = \toserver\ \map\ \func_2\ \K'$.
The reduction derivation which leads to the
two operations at the frontend being completed
is as follows:
\begin{center}
\scalebox{0.65}{
$\begin{array}{@{}r@{\ }l@{\,}l@{\hspace{-7em}}r@{\,}r@{\,}r@{}}
& \tuple{\data && \ostream;& \rstream;& \ep} \\
\cmidrule(r){2-6}
&\tuple{[\tuple{\inode_1; \emptyseq},& \tuple{\inode_2; \emptyseq}];& \emptyseq;& \emptyset;& \ep}
\\{}
(\plaintitle{Emit}) \configreduce
&\tuple{[\tuple{\inode_1; \emptyseq},& \tuple{\inode_2; \emptyseq}];& [[\fvs \mapsto \map\ \func_1\ \K]];& \emptyset;& (\lambda \_. \ep')\ \fvs}
\\{}
(\plaintitle{First}) \configreduce
&\tuple{[\tuple{\inode_1; [[\fvs \mapsto \map\ \func_1\ \K]]},& \tuple{\inode_2; \emptyseq}];& \emptyseq;& \emptyset;& (\lambda \_. \ep')\ \fvs}
\\{}
(\plaintitle{Map}) \configreduce
&\tuple{[\tuple{\ND{\nkey_1; \netwo{(\func_1\ \inode_1)}; \nethree{(\func_1\ \inode_1)}}; [[\fvs \mapsto \map\ \func_1\ \KS{\emptyseq}]]},& \tuple{\inode_2; \emptyseq}];& \emptyseq;& \emptyset&; (\lambda \_. \ep')\ \fvs}
\\{}
(\plaintitle{Load}) \configreduce^*
&\tuple{[\tuple{\ND{\nkey_1; \integer_1; \nethree{(\func_1\ \inode_1)}}; [[\fvs \mapsto \map\ \func_1\ \KS{\emptyseq}]]},& \tuple{\inode_2; \emptyseq}];& \emptyseq;& \emptyset&; (\lambda \_. \ep')\ \fvs}
\\{}
(\plaintitle{Load}) \configreduce^*
&\tuple{[\tuple{\ND{\nkey_1; \integer_1; \K_1}; [[\fvs \mapsto \map\ \func_1\ \KS{\emptyseq}]]},& \tuple{\inode_2; \emptyseq}];& \emptyseq;& \emptyset&; (\lambda \_. \ep')\ \fvs}
\\{}
(\plaintitle{Complete}) \configreduce
&\tuple{[\tuple{\ND{\nkey_1; \integer_1; \K_1}; \emptyseq},& \tuple{\inode_2; \emptyseq}];& \emptyseq;& [\fvs \xmapsto{\KS{\emptyseq}} 0]&; (\lambda \_. \ep')\ \fvs}
\\{}
(\plaintitle{Claim}) \configreduce
&\tuple{[\tuple{\ND{\nkey_1; \integer_1; \K_1}; \emptyseq},& \tuple{\inode_2; \emptyseq}];& \emptyseq;& [\fvs \xmapsto{\KS{\emptyseq}} 0]&; (\lambda \_. \ep')\ 0}
\\{}
(\plaintitle{Beta}) \configreduce
&\tuple{[\tuple{\ND{\nkey_1; \integer_1; \K_1}; \emptyseq},& \tuple{\inode_2; \emptyseq}];& \emptyseq;& [\fvs \xmapsto{\KS{\emptyseq}} 0]&; \ep'}
\\{}
(\plaintitle{Emit}) \configreduce
&\tuple{[\tuple{\ND{\nkey_1; \integer_1; \K_1}; \emptyseq},& \tuple{\inode_2; \emptyseq}];& [[\fvs' \mapsto \map\ \func_2\ \K']];& [\fvs \xmapsto{\KS{\emptyseq}} 0]&; \fvs'}
\\{}
(\plaintitle{First}) \configreduce
&\tuple{[\tuple{\ND{\nkey_1; \integer_1; \K_1}; [[\fvs' \mapsto \map\ \func_2\ \K']]},& \tuple{\inode_2; \emptyseq}];& \emptyseq;& [\fvs \xmapsto{\KS{\emptyseq}} 0]&; \fvs'}
\\{}
(\plaintitle{Prop}) \configreduce
&\tuple{[\tuple{\ND{\nkey_1; \integer_1; \K_1}; \emptyseq},& \tuple{\inode_2; [[\fvs' \mapsto \map\ \func_2\ \K']]}];& \emptyseq;& [\fvs \xmapsto{\KS{\emptyseq}} 0]&; \fvs'}
\\{}
(\plaintitle{Map}) \configreduce
&\tuple{[\tuple{\ND{\nkey_1; \integer_1; \K_1}; \emptyseq},& \tuple{\ND{\nkey_2; \netwo{(\func_2\ \inode_2)}; \nethree{(\func_2\ \inode_2)}}; [[\fvs' \mapsto \map\ \func_2\ \KS{\emptyseq}]]}];& \emptyseq;& [\fvs \xmapsto{\KS{\emptyseq}} 0]&; \fvs'}
\\{}
(\plaintitle{Load}) \configreduce^*
&\tuple{[\tuple{\ND{\nkey_1; \integer_1; \K_1}; \emptyseq},& \tuple{\ND{\nkey_2; \integer_2; \nethree{(\func_2\ \inode_2)}}; [[\fvs' \mapsto \map\ \func_2\ \KS{\emptyseq}]]}];& \emptyseq;& [\fvs \xmapsto{\KS{\emptyseq}} 0]&; \fvs'}
\\{}
(\plaintitle{Load}) \configreduce^*
&\tuple{[\tuple{\ND{\nkey_1; \integer_1; \K_1}; \emptyseq},& \tuple{\ND{\nkey_2; \integer_2; \K_2}; [[\fvs' \mapsto \map\ \func_2\ \KS{\emptyseq}]]}];& \emptyseq;& [\fvs \xmapsto{\KS{\emptyseq}} 0]&; \fvs'}
\\{}
(\plaintitle{Complete}) \configreduce
&\tuple{[\tuple{\ND{\nkey_1; \integer_1; \K_1}; \emptyseq},& \tuple{\ND{\nkey_2; \integer_2; \K_2}; \emptyseq}];& \emptyseq;& [\fvs' \xmapsto{\KS{\emptyseq}} 0, \fvs \xmapsto{\KS{\emptyseq}} 0]&; \fvs'}
\\{}
\end{array}$}
\end{center}
\end{example}
}%%%%%%%%%%%%%%%%%%%%%%%%% end inlong %%%%%%%%%%%%%%%%%%%%%%%%%

% As the example suggests,
% one step of program configuration reduction can only
% be established through a derivation that requires potentially multiple \plaintitle{$\downarrow$} instances,
% intuitively representing the traversal of the graph.
% Indeed, the only leaf nodes in the derivation in the eager semantic system
% must be instances of \plaintitle{Final}.
% \plaintitle{$\downarrow$}, \plaintitle{Map}, and \plaintitle{Fold}
% are inductively defined and cannot be leaf nodes.
% As we shall see, constructing a derivation of lazy processing does not require a graph traversal. 
% \dnote{is this still valid? }
% \pnote{no, I've commented it out}

\inlong{
\subsection{Properties}
\label{sec:metatheory:properties}
}

% \pnote{whichever way, need at least one lemma
% which does case analysis on the TLO rules}

% We now study the properties of \systemprefix{}. 
% %with a focus on establishing the relationship between
% %the unrestricted \systemprefix{} and \systemprefix{} under its eager restrictions.
% %
% Proof details can be found in an accompanying technical report~\cite{techreport}.
% \pnote{update tech report}

\inlong{
For the rest of this section,
we say a configuration $\config$ is well-typed
iff $\emptyseq \vdash_\mathtt{c} \config : \tau \bs \cdstatus$
for some $\tau$ and $\cdstatus$.
%\dnote{need work} 
}
\inshort{
We now state important properties for \systemprefix{}.
%
%We first introduce several definitions.
We say a backend is \emph{dry} if it follows the form $\aseq{\tuple{\eagerinode; \emptyseq}}$, written as $\eagerdata$.
We say a configuration $\config$ is well-typed
iff $\emptyseq \vdash_\mathtt{c} \config : \tau \bs \cdstatus$
for some $\tau$ and $\cdstatus$.
We define function $\init(\ep, \data)$ to compute
the initial configuration of frontend program $\ep$ given initial backend $\data$.
Specifically,
$\init(\ep, \data) \df \tuple{\data; \emptyseq; \emptyset; \ep}$.
The function $\init(\ep, \data)$ is only defined if
$\tuple{\data; \emptyseq; \emptyset; \ep}$ is well-typed.
%$\tuple{\data; \emptyseq; \emptyset; \ep}$ is well-typed.
According to this definition, a program does
not have to start with an empty graph;
it can start with a graph represented by $\data$.
}

\inlong{
Before we continue,
let us define function $\init(\ep, \data)$ to compute
the initial configuration of frontend program $\ep$ given initial graph $\data$.
Specifically,
$\init(\ep, \data) \df \tuple{\data; \emptyseq; \emptyset; \ep }$.
The function $\init(\ep, \data)$ is only defined if
$\tuple{\data; \emptyseq; \emptyset; \ep}$ is well-typed.
According to this definition, a program does
not have to start with an empty graph;
it can start with a graph represented by $\eagerdata$.
}

\paragraph{1) Type Soundness}

\begin{lemma}[Type Preservation]
\label{lem:preservation}
If
$\Gamma \vdash_\mathtt{c} \config : \tau \bs \cdstatus$,
and $\config \configreduce \config'$
then
$\Gamma \vdash_\mathtt{c} \config' : \tau \bs \cdstatus'$
where $\cdstatus = \cleanexp$ implies $\cdstatus' = \cleanexp$.
\end{lemma}
\inlong{
\begin{proof}
Case analysis on $\configreduce$.
The most interesting rule is \plaintitle{Emit}:

\textbf{Case \plaintitle{Emit}}
We have
\[
\config = \tuple{\data; \ostream; \rstream; \cE[\toserver \strictprocessor]}
\]
and
\[
\config' = \tuple{\data; \ostream \concat [[\fvs \mapsto \strictprocessor]]; \rstream; \cE[\fvs]}
\]
where $\fresh\ \fvs$.
Case analysis on $\strictprocessor$ will show that whatever type
\ttitle{Map},
\ttitle{Fold},
or
\ttitle{Add},
gives $\toserver\ \strictprocessor$ in $\config$, ($\mathsf{future}[\tau]$ for some $\tau$),
the same type is given to $\fvs$
in \rttitle{StreamUnit}
and accessed by $\rttitle{Future}$.

The rest of the rules follow from
that frontend reductions can change the $\cdstatus$ from $\dirtyexp$ to $\cleanexp$, but not the other way around,
and that task reductions which realize operations ---
\plaintitle{Map} and \plaintitle{Fold} ---
do not make copies of any operations, cannot make new operations (phase distinction), and if a result is created then the operation is removed.

\end{proof}
}

%\dnote{need progress lemma}

%\dnote{need type soundness theorem}

\begin{lemma}[Progress]
\label{lem:progress}
For any $\config$ which is well-typed,
then either
$\config = \tuple{\eagerdata; \emptyseq; \rstream; \val}$
for some $\eagerdata$ and $\rstream$
or there exists some $\config' \neq \config$
and
$\config \configreduce \config'$.
\end{lemma}
\inlong{
\begin{proof}
If the frontend is not a value then either (1) a frontend reduction can occur; or (2) the frontend is waiting on a result from the server so there is a to-graph, load, or task reduction that can occur.
If the frontend is a value then the graph is not dry, so there is a to-graph, load, or task reduction that can occur.
\end{proof}
}

%
%\begin{proof}
%%
%Every operation in \eagerlang{} either realizes or reaches the final node. More rigorously, induction over $\ereduce$ and case analysis on the last step of derivation. 
%For \ourlang{}, the theorem follows a simple case analysis over $\lreduce$ rules. 
%%and every operation in \ourlang{} can, at the least, be inserted in the root of the data structure
%%with a \plaintitle{Lazy}.
%%
%\end{proof}

\begin{theorem}[Soundness]
For any program $\ep$ and backend
$\data$,
if
$\init(\ep, \data) = \config$
then either there exists $\config'$ such that $\config \configreduce^* \config'$
where
$\config' = \tuple{\eagerdata; \emptyseq; \rstream; \val}$
or $\config$ diverges.
\end{theorem}
\inlong{
\begin{proof}
Induction on $\configreduce^*$ and case analysis on the last step, using Lemma~\ref{lem:preservation} and Lemma~\ref{progress}.
\end{proof}
}

\begin{corollary}[Phase Distinction]
For any well-typed configuration $\config$,
if $\config \configreduce \config'$, then either (1) the reduction is an instance of \plaintitle{Emit}, or (2) the reduction is not an instance of \plaintitle{Emit}, and its derivation does not contain an instance of \plaintitle{Emit}.
\end{corollary}

Recall that \plaintitle{Emit} is defined with the frontend context $\fE$. Case (1) says that operation emission may happen at the frontend.
On the backend, recall that the only reduction that may contain a subderivation
of \plaintitle{Emit} would be an instance of \plaintitle{Load}.
Case (2) says that such a derivation is not possible.
In other words, operation emission cannot happen on the backend.
The importance of phase distinction is that it contributes to
result determinism, which we elaborate next.

\paragraph{2) Result Determinism}

%We first define reduction relation $\config \eagerreduce \config'$
%as identical to $\config \configreduce \config'$
%except every occurence of
%$\fE$, $\bE$, $\tE$, and $\lE$
%are replaced with
%$\eagerfE$, $\eagerbE$, $\eagertE$, and $\eagerlE$
%respectively.
%We call this reduction \eagerlang{}.

% \begin{lemma}[Well-Formedness Preservation]
% \label{lem:wf}
% %
% If $\wf(\config_1)$, and $\config_1 \configreduce \config_2$, then $\wf(\config_2)$. 
% \end{lemma}
%
%\begin{proof}
%Case analysis on \plaintitle{} rules and \plaintitle{} rules.
%\end{proof}

%\systemprefix{} demonstrates different
%traits in determinism depending on whether
%it is under restriction.

%Through the discussion in \S~\ref{subsec:ingraph}, it is self-evident that \systemprefix{} without any restriction supports non-deterministic executions. More relevant here is the restrictive form of eager data processing is deterministic: 

%\paragraph{3) Observable Equivalence}

With generality as a design goal, \systemprefix{} is guided with a design rationale that we should place as few restrictions on the evaluation order as possible, leading to a semantics inherent with non-deterministic executions. One example is the non-deterministic redex selection for propagation which we described in \S~\ref{sec:highlevelsemantics}. More generally, a simple case analysis of evaluation contexts in Fig.~\ref{fig:contexts} should make clear that \systemprefix{} is endowed with non-deterministic redex selection between:
\begin{itemize}
    \item \emph{a frontend reduction and a backend reduction}: given a configuration, either $\fE$ or $\bE$ can be used for selecting the redex of the next step of reduction;    
    \item \emph{task reductions over different stations}: according to $\tE$, the redex can be an arbitrary station in the runtime graph, where the task is an instance of \plaintitle{Map}, \plaintitle{Fold}, \plaintitle{Complete}, and \plaintitle{Last}, or two adjacent stations, where the task is an instance of \plaintitle{Prop};
    \item \emph{load reductions inside different stations}: according to $\lE$, the redex can be any load expression inside an arbitrary station;
    \item \emph{a task reduction and a load reduction}: either $\tE$ and $\lE$ can be used for redex selection. 
\end{itemize}

Non-deterministic executions are good news for generality and adaptability (see \S~\ref{subsec:ingraph}),
but they are a challenge to correctness: do different reduction sequences from the same configuration produce the same result?
% We affirmatively answer this question now. 

%, and more importantly, if they
%all indeed produce the same result,
%how is that result related to what is produced by
%the restricted \systemprefix{} (observable equivalence).
%%
%We answer these questions next.

% We are now ready to establish some high-level properties about our calculus.
% %
% First, even though \systemprefix{} is non-deterministic,
% %(Lemma~\ref{lem:nonde}),
% all reduction sequences from the same configuration produce the same result:

\inlong{

\begin{definition}[Configuration Equivalence]
We define equivalence relation $\config \sim \config'$ by the rules in Fig.~\ref{fig:configequiv},
where helper relation
$\data \stackrel{b}{\sim} \data'$
% and
% $\data \stackrel{\mathit{op}}{\sim} \data'$
says that the payload and edge list of all stations in the backends are term equivalent,
and that the intention of the delayed operations in
the stations are equivalent.
The most interesting parts of the $\stackrel{b}{\sim}$ relation is the
\emph{write effect} and \emph{read result} equivalence check.
These are defined with helper functions $\fundef{writes}$ and $\fundef{reads}$, respectively.
To check the write effect, the mapping operations for every node are composed, and checked against the two operation streams for term equivalence.
To check the read result, every fold operation
is first gathered together with the functions of any relevant mapping operations composed to the fold function. Then the, $\stackrel{\cup}{\equiv}$ relation checks that, for every fold operation, the fold functions are term equivalent (using the helper function $\fundef{expr}$), regardless of the node values that are eventually applied as arguments.
The $\fundef{expr}$ function also performs any inter-fold substitution, to allow for more general term equality in the $\stackrel{\cup}{\equiv}$ relation.
\end{definition}
\begin{figure}[t]
\centering

\begin{mathpar}
%\inferrule{}{
%   \config \sim \config
%}
%
\scalebox{\scalemath}{$
\inferrule{
   \data
   \stackrel{b}{\sim}
   \data'
}{
   \tuple{\data; \ostream; \rstream; \ep}
   \sim
   \tuple{\data'; \ostream; \rstream; \ep}
}
$}

\scalebox{\scalemath}{$
\inferrule{}{
   \emptyseq \stackrel{b}{\sim} \emptyseq
}
$}

\scalebox{\scalemath}{$
\inferrule{
   \data \stackrel{b}{\sim} \data'
\\ \ep_1 \feq \ep_1'
\\ \ep_2 \feq \ep_2'
\\ \forall \nkey' . \fundef{writes}(\fundef{fl}(\ostream), \nkey') \feq \fundef{writes}(\fundef{fl}(\ostream'), \nkey')
% \\ \forall \nkey' . \fundef{reads}(\ostream, \nkey') \stackrel{\cup}{\feq} \fundef{reads}(\ostream', \nkey')
\\ \fundef{reads}(\fundef{fl}(\ostream)) \stackrel{\cup}{\feq} \fundef{reads}(\fundef{fl}(\ostream'))
}{
   \tuple{\tuple{\nkey; \ep_1; \ep_2}; \ostream} \cons \data \stackrel{b}{\sim} \tuple{\tuple{\nkey; \ep_1'; \ep_2'}; \ostream'} \cons \data'
}
$}

\scalebox{\scalemath}{$
\inferrule{}{
   \fundef{writes}(\emptyseq, \nkey) \df \mathtt{Id}
}
$}

\scalebox{\scalemath}{$
\inferrule{
   \nkey \in \aset{\nkey}
}{
   \fundef{writes}(\ostream \concat [[\fvs \mapsto \map\ \func\ \KS{\aseq{\nkey}}]], \nkey) \df \func \circ \fundef{writes}(\ostream, \nkey)
}
$}

\scalebox{\scalemath}{$
\inferrule{
   \nkey \notin \target(\processor) \lor \processor = \fold\ \func\ \ep\ \K
}{
   \fundef{writes}(\ostream \concat [[\fvs \mapsto \processor]], \nkey) \df \fundef{writes}(\ostream, \nkey)
}
$}

\scalebox{\scalemath}{$
\inferrule{}{
   \fundef{reads}(\emptyseq) \df \{ \}
}
$}

% \inferrule{
%    \nkey \ksin \K
% }{
%    \fundef{reads}(\ostream \concat [[\fvs \mapsto \fold\ \func\ \ep\ \K]], \nkey) \df \{ \fvs \mapsto (\func \mathbin{\overset{\{ \nkey \}}{\circ{\circ}}} \fundef{writes}(\ostream, \nkey))\ \ep \} \cup \fundef{reads}(\ostream, \nkey)
% }
%
\scalebox{\scalemath}{$
\inferrule{
   \func' = (((\func \dcomp{\{ \nkey_1 \}} \fundef{writes}(\ostream, \nkey_1)) \dcomp{\{ \nkey_2 \}} \fundef{writes}(\ostream, \nkey_2)) \ldots) \dcomp{\{ \nkey_m \}} \fundef{writes}(\ostream, \nkey_m)
}{
   \fundef{reads}(\ostream \concat [[\fvs \mapsto \fold\ \func\ \ep\ \KS{\aseq{\nkey}^m}]]) \df \{ \fvs \mapsto \fold\ \func'\ \ep\ \K \} \cup \fundef{reads}(\ostream, \nkey)
}
$}

\scalebox{\scalemath}{$
\inferrule{}{
   \fundef{reads}(\ostream \concat [[\fvs \mapsto \map\ \func\ \K]]) \df \fundef{reads}(\ostream)
}
$}

\scalebox{\scalemath}{$
\inferrule{
   \forall i \in [1, \ldots, m] , j \in [1, \ldots, m'] .
   \fvs_i = \fvs'_j
   \implies
   \fundef{expr}(\processor_i)
   \feq
   \fundef{expr}(\processor'_j)
}{
   \aset{\fvs \mapsto \processor}^{m}
   \stackrel{\cup}{\feq}
   \aset{\fvs' \mapsto \processor'}^{m'}
}
$}

\scalebox{\scalemath}{$
\inferrule{}{
   \fundef{expr}(\fold\ \func\ \ep\ \KS{\aseq{\nkey}^m}, \aset{\fvs \mapsto \processor}) \df (\func\ \ND{\nkey_1; \_; \_}\ (\func\ \ND{\nkey_2; \_; \_}\ (\ldots (\func\ \ND{\nkey_m; \_; \_}\ \ep))))[\aset{\fundef{expr}(\processor, \aset{\fvs \mapsto \processor}) / \fromserver\ \fvs}]
}
$}

\scalebox{\scalemath}{$
\inferrule{}{
   \fundef{fl}(\aseq{\batchgroup}^m) \df \batchgroup_1 \concat \batchgroup_2 \concat \ldots \concat \batchgroup_m
}
$}
\end{mathpar}

\caption{
Configuration Equivalence: $\config \sim \config'$}
\label{fig:configequiv}

\end{figure}

\begin{lemma}[Composition Reorder]
\label{lem:compreorder}
Given functions $\func$, $\func'$, and $\func'$,
and keys $\nkey$ and $\nkey'$,
where $\nkey \neq \nkey'$,
then
$
(\func \dcomp{\{\nkey\}} \func') \dcomp{\{\nkey'\}} \func''
\feq
(\func \dcomp{\{\nkey'\}} \func'') \dcomp{\{\nkey\}} \func'
$
\end{lemma}

% \begin{lemma}[Composition Decompose]
% \label{lem:compdecompose}
% Given functions $\func$ and $\func'$,
% and target key list $\K = \KS{\aseq{\nkey}^m}$,
% then
% %
% $
% (((\func \mathbin{\overset{\{\nkey_1\}}{\circ{\circ}}} \func') \dcomp{\{\nkey_2\}} \func') \ldots) \dcomp{\{\nkey_m\}} \func'
% \feq
% \func \mathbin{\overset{\K}{\circ{\circ}}} \func'
% $
% \end{lemma}

\begin{lemma}[Fold Target Decompose]
\label{lem:foldtargetdecompose}
Given functions $\func$ and $\func'$,
expression $\ep$,
and target key lists
$\K = \KS{\aseq{\nkey}^m}$
and
$\K' \KS{\aseq{\nkey'}^{m'}}$
then the operation
$\fvs \mapsto \fold\ (\func \dcomp{\aset{\nkey'}} \func')\ \ep\ \K$
is equivalent to
$\fvs \mapsto \fold\ (((\func \dcomp{\{\nkey_1\}} (\mathbf{if}\ \nkey_1 \in \aset{\nkey'} \ \mathbf{then}\ \func' \ \mathbf{else}\ \mathtt{Id})) \ldots) \dcomp{\nkey_m} (\mathbf{if}\ \nkey_m \in \aset{\nkey'} \ \mathbf{then}\ \func' \ \mathbf{else}\ \mathtt{Id}))\ \ep\ \K$.
\end{lemma}

\begin{lemma}[Configuration Equivalence Over TLO Reduction]
\label{lem:tlosim}
Given well-typed $\config$ and $\config \configreduce \config'$
where the reduction is an instance of \plaintitle{Opt},
then $\config \sim \config'$.
\end{lemma}
\begin{proof}
Case analysis on the \htitle{} rule.

\textbf{Case \htitle{Batch} and \htitle{Unbatch}}
The write effect and read result conditions ignore batches,
so $\config \sim \config'$ holds.

\textbf{Case \htitle{ReorderD}}
Before the reduction we have
$\ostream = \ostream_1 \concat [[\fvs_1 \mapsto \processor_1], [\fvs_2 \mapsto \processor_2]] \concat \ostream_2$
and after we have
$\ostream' = \ostream_1 \concat [[\fvs_2 \mapsto \processor_2], [\fvs_1 \mapsto \processor_1]] \concat \ostream_2$.
We perform case analysis on the operations $\processor_1$ and $\processor_2$,
there are four cases:
(1) $\fold$ $\fold$:
The read result condition holds since the $\fundef{reads}$ function is not affected by the order of fold operations.
(2) $\fold$ $\map$:
The write effect is preserved since the $\fundef{writes}$ function ignores $\fold$ operations.
The read result condition is preserved
since the $\fundef{writes}$ function will return $\mathtt{Id}$ for any node in target set of the $\fold$ operation, since the target lists of the $\map$ and $\fold$ operation are disjoint.
(3) $\map$ $\fold$:
Identical reasoning as the previous case.
(4) $\map$ $\map$:
Since, for any key, $\fundef{writes}(\ostream, \nkey)$ will include either the first $\map$ function or the second, never both, since the target key lists are disjoint.

\textbf{Case \htitle{ReorderRR}:}
Since the $\fundef{reads}$ function does not depend on the order of fold operations.

\textbf{Case \htitle{ReorderRW}:}
The write effect is preserved since the $\fundef{writes}$ function ignores $\fold$ operations.
For the read result,
assume the fold operation is $\fvs_2 \mapsto \fold\ \func_2\ \ep_2\ \K_2$ and the map operation is $\fvs_1 \mapsto \map\ \func_1\ \K_1$.
Before the reduction, the $\fundef{reads}$ entry for $\fvs_2$ will be
\[
\fold\ (((\func_2 \dcomp{\{ \nkey_1 \}} \fundef{writes}(\ostream \concat [[\fvs_1 \mapsto \map\ \func_1\ \K_1]], \nkey_1)) \ldots) \dcomp{\{ \nkey_m \}} \fundef{writes}(\ostream \concat [[\fvs_1 \mapsto \map\ \func_1\ \K_1]], \nkey_m))\ \ep_2\ \K_2
\]
where $\K_2 = \KS{\aseq{\nkey}^m}$.
After the rule, the fold operation is transformed to $\fvs_2 \mapsto \fold_2\ (\func_2 \dcomp{\aset{\nkey'}} \func_1)\ \ep_2\ \K_2$
where $\K_1 = \KS{\aseq{\nkey'}^{m'}}$.
and the $\fundef{reads}$ entry for $\fvs_2$ will be
\[
\fold\ ((((\func_2 \dcomp{\aset{\nkey'}} \func_1) \dcomp{\{ \nkey_1 \}} \fundef{writes}(\ostream, \nkey_1)) \ldots) \dcomp{\{ \nkey_m \}} \fundef{writes}(\ostream, \nkey_m))\ \ep_2\ \K_2
\]
The pre-reduction $\fundef{reads}$ entry can be rewritten to the following using the definition of $\fundef{writes}$:
\[
\fold\ (((\func_2 \dcomp{\{ \nkey_1 \}} ((\mathbf{if}\ \nkey_1 \in \aset{\nkey'} \ \mathbf{then}\ \func_1 \ \mathbf{else}\ \mathtt{Id}) \circ \fundef{writes}(\ostream, \nkey_1))) \ldots) \dcomp{\{ \nkey_m \}} ((\mathbf{if}\ \nkey_m \in \aset{\nkey'} \ \mathbf{then}\ \func_1 \ \mathbf{else}\ \mathtt{Id}) \circ \fundef{writes}(\ostream, \nkey_m)))\ \ep_2\ \K_2
\]
%which from the definition of $\dcomp{\K}$ can be written to the following
%\[
%\fold\ (((((\func_2 \dcomp{\K_1} \func_1) \dcomp{\{ \nkey_1 \}} \fundef{writes}(\ostream, \nkey_1)) \ldots) \dcomp{\{ \nkey_m \}} \fundef{writes}(\ostream \concat [[\fvs_1 \mapsto \map\ \func_1\ \K_1]], \nkey_1))\ \ep_2\ \K_2
%\]
%
The post-reduction $\fundef{reads}$ entry can be rewritten to the following using Lemma~\ref{lem:foldtargetdecompose}:
\[
\fold\ ((((((\func_2 \dcomp{\{ \nkey_1 \}} (\mathbf{if}\ \nkey_1 \in \aset{\nkey'} \ \mathbf{then}\ \func_1 \ \mathbf{else}\ \mathtt{Id})) \ldots) \dcomp{\{ \nkey_m \}} (\mathbf{if}\ \nkey_m \in \aset{\nkey'} \ \mathbf{then}\ \func_1 \ \mathbf{else}\ \mathtt{Id})) \dcomp{\{\nkey_1\}} \fundef{writes}(\ostream, \nkey_1)) \ldots) \dcomp{\{\nkey_m\}} \fundef{writes}(\ostream, \nkey_m))\ \ep_2\ \K_2
\]
which can be rewritten to the rewritten pre-reduction $\fundef{reads}$ entry using Lemma~\ref{lem:compreorder}, 

\textbf{Case \htitle{FuseM}:}
Before the reduction we have
$\ostream = \ostream_1 \concat [[\fvs_1 \mapsto \map\ \func_1\ \K], [\fvs_2 \mapsto \map\ \func_2\ \K]] \concat \ostream_2$
and after we have
$\ostream' = \ostream_1 \concat [[\fvs_2 \mapsto \map\ (\func_2 \circ \func_1)\ \K]] \concat \ostream_2$.
For every $\nkey \in \K$, $\fundef{writes}(\ostream, \nkey)$ returns
$\fundef{writes}(\ostream_2, \nkey) \circ \func_2 \circ \func_1 \circ \fundef{writes}(\ostream_1, \nkey)$.,
the $\fundef{writes}(\ostream', \nkey)$ returns
$\fundef{writes}(\ostream_2, \nkey) \circ (\func_2 \circ \func_1) \circ \fundef{writes}(\ostream_1, \nkey)$,
which are equivalent due to associativity of function composition.

\textbf{Case \htitle{FuseMId}:}
Before the reduction we have
$\ostream = \ostream_1 \concat [[\fvs_1 \mapsto \map\ \func_1\ \K], [\fvs_2 \mapsto \map\ \func_2\ \K]] \concat \ostream_2$
and after we have
$\ostream' = \ostream_1 \concat [] \concat \ostream_2$.
For every $\nkey \in \K$, $\fundef{writes}(\ostream, \nkey)$ returns
$\fundef{writes}(\ostream_2, \nkey) \circ \func_2 \circ \func_1 \circ \fundef{writes}(\ostream_1, \nkey)$.,
the $\fundef{writes}(\ostream', \nkey)$ returns
$\fundef{writes}(\ostream_2, \nkey) \circ \fundef{writes}(\ostream_1, \nkey)$.
From the side condition of \htitle{FuseMId}, we know that $\func_2 \circ \func_1 \feq \mathtt{Id}$. Replacing this in the result of $\fundef{writes}(\ostream, \nkey)$ makes the two cases equivalent.

\textbf{Case \htitle{Reuse}:}
Before the reduction we have
$\ostream = \ostream_1 \concat [[\fvs_1 \mapsto \fold\ \func\ \ep\ \K_1], [\fvs_2 \mapsto \fold\ \func\ \ep\ \K_2]] \concat \ostream_2$
and after we have
$\ostream' = \ostream_1 \concat [[\fvs_1 \mapsto \fold\ \func\ \ep\ \K_1], [\fvs_2 \mapsto \fold\ \func\ (\fromserver\ \fvs_1)\ (\K_2 \kssetminus \K_1)]] \concat \ostream_2$.
The $\fundef{reads}$ entry for $\fvs_1$ before and after the reduction is identical.
The $\fundef{reads}$ entry for $\fvs_2$ before the reduction is equivalent to
\[
\fvs_2 \mapsto \fold\ (\func \dcomp{\aset{\nkey}} \func')\ \ep\ \K_2
\]
for some $\aset{\nkey}$ and $\func'$,
and after the reduction it is equivalent to
\[
\fvs_2 \mapsto \fold\ (\func \dcomp{\aset{\nkey}} \func')\ (\fromserver\ \fvs_1)\ (\K_2 \kssetminus \K_1)
\]
The result of $\fundef{expr}$ is the following for the operation before the reduction:
\[
\func''\ \ND{\nkey_1; \_; \_}\ (\ldots\ (\func''\ \ND{\nkey_m; \_; \_}\ \ep))
\]
where $\K_2 = \KS{\aseq{\nkey}^m}$
and $\func'' = \func \dcomp{\aset{\nkey}} \func'$,
and the following for the operation after the reduction:
\[
\func''\ \ND{\nkey'_1; \_; \_}\ (\ldots\ (\func''\ \ND{\nkey'_{m'}; \_; \_}\ (\func''\ \ND{\nkey''_1; \_; \_}\ (\ldots\ (\func''\ \ND{\nkey''_{m''}; \_; \_}\ \ep))))
\]
where $\K_2 \kssetminus \K_1 = \KS{\aseq{\nkey'}^{m'}}$.
and $\K_1 = \KS{\aseq{\nkey''}^{m''}}$.
From the side conditions of \htitle{Reuse},
we have that $\K_1 \kssubseteq \K_2$ and $\func$ is commutative.
We then know that $\K_2 = (\K_2 \kssetminus \K_1) \cup \K_1$,
meaning the two results of $\fundef{expr}$ are equivalent.

%\textbf{Case \htitle{ReuseMqoE}:}
%Before the reduction we have
%$\ostream = \ostream_1 \concat [[\fvs_1 \mapsto \fold\ \func_1\ \ep\ \KS{[\nkey]}], [\fvs_2 \mapsto \fold_2\ \func\ \ep\ \KS{[\nkey]}]] \concat \ostream_2$
%where $\func_2 \feq \lambda x . \lambda y . \func\ x\ (\func_1\ x\ y)$
%and after we have
%$\ostream' = \ostream_1 \concat [[\fvs_1 \mapsto \fold\ \func_1\ \ep\ \KS{[\nkey]}], [\fvs_2 \mapsto \fold\ (\lambda x . \lambda y . \func\ x\ (\fromserver\ \fvs_1))\ \KS{[\nkey]}]] \concat \ostream_2$.
%The $\fundef{reads}$ entry for $\fvs_1$ before and after the reduction is identical.
%The $\fundef{reads}$ entry for $\fvs_2$ before the reduction is equivalent to
%\[
%\fvs_2 \mapsto \fold\ (\func_2 \dcomp{\aset{\nkey}} \func')\ \ep\ \KS{[\nkey]}
%\]
%for some $\aset{\nkey}$ and $\func'$,
%and after the reduction it is equivalent to
%\[
%\fvs_2 \mapsto \fold\ ((\lambda x . \lambda y . \func\ x\ (\fromserver\ \fvs_1)) \dcomp{\aset{\nkey}} \func')\ (\fromserver\ \fvs_1)\ \KS{[\nkey]}
%\]
%The result of $\fundef{expr}$ is the following for the operation before the reduction, after substituting $\func_2$ with the side condition of \htitle{ReuseMqoE}:
%\[
%((\lambda x . \lambda y . \func\ x\ (\func_1\ x\ y)) \dcomp{\aset{\nkey}} \func')\ \ND{\nkey; \_; \_}\ \ep
%\]
%and the following for the operation after the reduction:
%\[
%((\lambda x . \lambda y . \func\ x\ (\func_1\ \ND{\nkey;\ _; \_}\ \ep)) \dcomp{\aset{\nkey}} \func')\ \ND{\nkey; \_; \_}\ \ep
%\]
%The two terms are equivalent after beta reduction.

\end{proof}

For the rest of this section, we only consider finite reduction sequences.
There are cases where a reduction sequence involving $\configreduce$ can be infinite, but we are not concerned with proving the confluence of divergent reductions.

%\pnote{expand if we only use in one location}
We use the following notation:
% $\config_1 \downarrow \config_2 \Leftrightarrow \exists \config_3 . \config_1 \configreduce^* \config_2 \land \config_2 \configreduce^* \config_3$;
$\config_x \mathrel{\overset{\sim}{\downarrow}} \config_y \Leftrightarrow \exists \config_u , \config_v . \config_x \configreduce^* \config_u \land \config_y \configreduce^* \config_v \land \config_u \sim \config_v$.

\begin{figure}[t]
\centering
\begin{tabular}{cc}
\begin{subfigure}{0.5\textwidth}
\centering
\begin{tikzpicture}
\node at (0,2) (c1) {$\config_x$};
\node at (-1,1) (c2) {$\config_y$};
\node at (1,1) (c3) {$\config_z$};
\node at (-0.3,0) (joinleft) {$\config_u$};
\node at (0.3,0) (joinright) {$\config_v$};
\node at ($(joinright) !0.5! (joinleft)$) {$\sim$};
\draw[->] (c1) -- (c2);
\draw[->] (c1) -- (c3);
\draw[->,dashed] (c2) -- node[midway,above,sloped] {$*$} (joinleft);
\draw[->,dashed] (c3) -- node[midway,above,sloped] {$*$} (joinright);
\end{tikzpicture}
\caption{P1}
\end{subfigure}
&
\begin{subfigure}{0.5\textwidth}
\centering
\begin{tikzpicture}
\node at (-0.3,2) (c1) {$\config_x$};
\node at (0.3,2) (c2) {$\config_z$};
\node at ($(c1) !0.5! (c2)$) {$\sim$};
\node at (-1,1) (c3) {$\config_y$};
\node at (-0.3,0) (joinleft) {$\config_u$};
\node at (0.3,0) (joinright) {$\config_v$};
\node at ($(joinleft) !0.5! (joinright)$) {$\sim$};
\draw[->] (c1) -- (c3);
\draw[->,dashed] (c3) -- node[midway,above,sloped] {$*$} (joinleft);
\draw[->,dashed] (c2) -- node[midway,above,sloped] {$*$} (joinright);
\end{tikzpicture}
\caption{P2}
\end{subfigure}
\end{tabular}
\caption{Two Properties of Local Confluence Modulo $\sim$}
\label{fig:localtwocases}
\end{figure}

The rest of the properties are inspired from confluence properties and proofs for term rewriting systems~\cite{4567923}.

\begin{lemma}[Local Confluence Modulo $\sim$ (P1)]
\label{lem:localp1}
For any well-typed $\config_x$, $\config_y$, and $\config_z$,
if
$\config_x \configreduce \config_y$
and
$\config_x \configreduce \config_z$
then
$\config_y \mathrel{\overset{\sim}{\downarrow}} \config_z$.
\end{lemma}
\begin{proof}
Case analysis on the two reductions:

\textbf{Case Both reductions are frontend reductions:}
The choice of frontend reduction is deterministic, so $\config_y = \config_z$.

\textbf{Case Both reductions are to-graph reductions:}
The choice of to-graph reduction is deterministic, so $\config_y = \config_z$.

\textbf{Case Both reductions are load reductions:}
There are two cases:
(1) each load reduction step occurs at a distinct station;
(2) both load reduction steps occur at the same station
but either (2a) one reduces the node and one reduces a fold operation or (2b) both reduce separate fold operations.
For both (1) and (2), $\config_u$ can be constructed such that $\config_y \configreduce \config_u$ and $\config_z \configreduce \config_v$ (each reduction replicates the load reduction performed in the other first step).

\textbf{Case Both reductions are task reductions:}
Case analysis on the task reductions, there are three cases:
(1) Neither is \plaintitle{Opt}:
The choice of task reduction at a given station is deterministic.
$\config_u$  can be constructed such that $\config_y \configreduce \config_u$ and $\config_v$ such that $\config_z \configreduce \config_v$ (each reduction replicates the task reduction performed in the other first step).
(2) One is \plaintitle{Opt}:
Assume $\config_x \configreduce \config_y$ is an instance of \plaintitle{Opt}.
From Lemma~\ref{lem:tlosim} we have that $\config_x \sim \config_y$.
It is then possible to either replicate the task reduction in $\config_x \configreduce \config_z$ to $\config_y$ or to replicate the \plaintitle{Opt} reduction to $\config_z$.
(3) Both are \plaintitle{Opt}:
From Lemma~\ref{lem:tlosim} we have that $\config_x \sim \config_y$ and $\config_x \sim \config_z$,
so $\config_y \sim \config_z$.

\textbf{Case 5: One reduction is a frontend or to-graph reduction and the other is a task or load reduction}
Since the choice of frontend and to-graph reduction is deterministic.

\textbf{Case 6: One reduction is a load reduction and one is a task reduction}
There are two cases:
(1) the task reduction moves the load reduction via an instance of \plaintitle{Prop} or \plaintitle{Opt};
(2) the task reduction does not move the load reduction.
For (2) the reductions do not interfere and can be replicated.
For (1), using Lemma~\ref{lem:tlosim} for any \plaintitle{Opt} reduction
and replicating a moved load reduction in the case of \plaintitle{Prop}.
\end{proof}

% \pnote{write out}
% For any two reductions of type frontend, to-graph, task (not counting \plaintitle{Opt}), and load, then the reductions can be replicated for P1.
% For P2 \pnote{\ldots}, any reduction involving an operation stream which is different between configurations can be replicated using a combination of load and task reductions. For example, a propagation of a batched operation can be replicated by propagating the two non-batched operations.
% %
% For any \plaintitle{Opt} reduction, P2 follows from Lemma~\ref{lem:tlosim}.
% For P1, any reduction involving an operation stream which is different between configurations can be replicated using a combination of load and task reductions.

\begin{lemma}[Local Confluence Modulo $\sim$ (P2)]
\label{lem:localp2}
For any well-typed $\config_x$, $\config_y$, and $\config_z$,
if
$\config_x \sim \config_z$
and
$\config_x \configreduce \config_y$
then
$\config_y \mathrel{\overset{\sim}{\downarrow}} \config_z$.
\end{lemma}
\begin{proof}
Case analysis on the reduction:

\textbf{Case Frontend or to-graph reduction:}
Replicate the reduction since frontend and to-graph reductions
do not interfere with any optimized operations.

\textbf{Case Load or task reduction:}
If \plaintitle{Opt} then Lemma~\ref{lem:tlosim} gives us that $\config_x \sim \config_y$, so $\config_y \sim \config_z$.
Otherwise, replicate the load or task reduction.
\end{proof}

P1 and P2 for local confluence modulo $\sim$ are illustrated in Fig.~\ref{fig:localtwocases}

\begin{definition}[A Normal Form]
We say $\config'$ is a normal form of $\config$
if $\config \configreduce^* \config'$
and $\config' = \tuple{\eagerdata; \emptyseq; \emptyset; \val}$.
\end{definition}

\newcommand{\bconfig}{\bar{\config}}
\begin{lemma}[Confluence Modulo $\sim$]
\label{lem:confluence}
For any well-typed $\config_x$, $\config_x'$, $\config_y$, and $\config_y'$,
if
$\config_x \sim \config_y$, and
$\config_x \configreduce ^* \config_x'$, and
$\config_y \configreduce ^* \config_y'$,
then there exists $\bconfig_x$ and $\bconfig_y$
such that
$\config_x' \configreduce^* \bconfig_x$, and
$\config_y' \configreduce^* \bconfig_y$, and
$\bconfig_x \sim \bconfig_y$.
\end{lemma}
%\pnote{try to simplify for our specific use case}
%\pnote{reword}
\begin{proof}
Induction on $\configreduce$ and the following property:
\[
P(\config_x, \config_y): \config_x \sim \config_y \implies
[\forall \config_x', \config_y' . \config_x \configreduce^* \config_x' \land \config_y \configreduce ^* \config_y' \implies \config_x' \mathrel{\overset{\sim}{\downarrow}} \config_y']
\]
Given $\config_x$, $\config_y$,
$\config_x'$, and $\config_y'$,
where
$\config_x \sim \config_y$, and
$\config_x \stackrel{n}{\configreduce} \config_x'$
and
$\config_y \stackrel{m}{\configreduce} \config_y'$,
we must show that there exists
$\bconfig_x$ and $\bconfig_y$
such that
$\config_x' \configreduce ^* \bconfig_x$,
and
$\config_y' \configreduce ^* \bconfig_y$,
and $\bconfig_x \sim \bconfig_y$.

If $n = 0$ and $m = 0$ then the result directly follows.
Otherwise, assume $n > 0$ such that $\config_x \configreduce \config_{x_1} \configreduce ^* \config_x'$.
Applying P2 to $\config_x$, $\config_y$, and $\config_{x_1}$, we get $\config_u$ and $\config_v$ such that $\config_{x_1} \configreduce ^* \config_u$, and $\config_y \configreduce^* \config_v$, and $\config_u \sim \config_v$.
There are two cases:

\textbf{Case 1} $m = 0$.
Let $\bconfig_x'$ and $\bconfig_u$ and $\bconfig_v$ be normal forms of $\config_x'$, $\config_u$, and $\config_v$.
$\bconfig_x' \sim \bconfig_u$ by the induction hypothesis $P(\config_{x_1}, \config_{x_1})$,
and $\bconfig_u \sim \bconfig_v$ by the induction hypothesis $P(\config_u, \config_v)$, which completes the proof.
The diagram for this case is shown in Fig.~\ref{fig:confluence:case1}.

\textbf{Case 2} $m > 0$.
Assume
$\config_y \configreduce \config_{y_1} \configreduce^* \config_y'$. There are two sub cases:

\textbf{Case 2a} $\config_v = \config_y$.
Applying P2 to $\config_y$, $\config_u$, and $\config_{y_1}$, we get $\config_w$ and $\config_z$ such that
$\config_u \configreduce^* \config_w$, and
$\config_{y_1} \configreduce^* \config_z$, and
$\config_w \sim \config_z$.
Let $\bconfig_x'$, $\bconfig_w$, and $\bconfig_z$ be normal forms of $\config_x'$, $\config_w$, and $\config_z$, respectively.
$\bconfig_x' \sim \bconfig_w$ by the induction hypothesis $P(\config_{x_1}, \config_{x_1})$, and
$\bconfig_w \sim \bconfig_z$ by $P(\config_w, \config_z)$, and
$\bconfig_z \sim \bconfig_y'$ by $P(\config_{y_1}, \config_{y_1})$,
which completes the proof.
The diagram for this case is shown in Fig.~\ref{fig:confluence:case2a}.

\textbf{Case 2b} Otherwise,
Assume $\config_y \configreduce \config_t \configreduce^* \config_v$.
Applying P1 to $\config_y$, $\config_{y_1}$, and $\config_t$, we get $\config_w$ and $\config_z$ such that
$\config_u \configreduce^* \config_w$, and
$\config_{y_1} \configreduce^* \config_z$,
and $\config_w \sim \config_z$.
Let $\bconfig_x'$, $\bconfig_u$, $\bconfig_v$, $\bar\config_w$, $\bconfig_z$, and $\bconfig_y'$ be normal forms of $\config_x'$, $\config_u$, $\config_v$, $\config_w$, $\config_z$, and $\config_y'$, respectively.
$\bconfig_x' \sim \bconfig_u$ by the induction hypothesis $P(\config_{x_1}, \config_{x_1})$, and
$\bconfig_u \sim \bconfig_v$ by $P(\config_u, \config_v)$, and
$\bconfig_v \sim \bconfig_w$ by $P(\config_t, \config_t)$, and
$\bconfig_w \sim \bconfig_z$ by $P(\config_w, \config_z)$, and
$\bconfig_z \sim \bconfig_y'$ by $P(\config_{y_1}, \config_{y_1})$,
which completes the proof.
The diagram for this case is shown in Fig.~\ref{fig:confluence:case2b}.

\end{proof}
\begin{figure}[t]
\centering

\begin{tabular}{cc}
\begin{subfigure}{0.5\textwidth}
\centering
\begin{tikzpicture}
\node at (-0.3,2) (cx) {$\config_x$};
\node at (0.9,2) (cy) {$\config_y$};
\node at ($(cy) + (0.6,0)$) {$= \config_y'$};
\node[scale=2] at ($(cx) !0.5! (cy)$) {$\sim$};
\node at (-1.2,1) (cx1) {$\config_{x_1}$};
\draw[->] (cx) -- (cx1);
\node at (-1.8,0) (cx') {$\config_x'$};
\node at (-0.6,0) (cu) {$\config_u$};
\draw[->,dashed] (cx1) -- node[midway,above,sloped] {$*$} (cx');
\draw[->,dashed] (cx1) -- node[midway,above,sloped] {$*$} (cu);
\node at (-1.8,-1) (cbarx') {$\bconfig_x'$};
\draw[->,dashed] (cx') -- node[midway,above,sloped] {$*$} (cbarx');
\node at (-0.6,-1) (cbaru) {$\bconfig_u$};
\draw[->,dashed] (cu) -- node[midway,above,sloped] {$*$} (cbaru);
\node at (0.9,0) (cv) {$\config_v$};
\draw[->,dashed] (cy) -- node[midway,above,sloped] {$*$} (cv);
\node at (0.9,-1) (cbarv) {$\bconfig_v$};
\draw[->,dashed] (cv) -- node[midway,above,sloped] {$*$} (cbarv);
\node[scale=2] at ($(cbarx') !0.5! (cbaru)$) {$\sim$};
\node[scale=2] at ($(cu) !0.5! (cv)$) {$\sim$};
\node[scale=2] at ($(cbaru) !0.5! (cbarv)$) {$\sim$};
\node at (0.0,1.0) {$P2$};
\node at (-1.2,-0.5) {Ind};
\node at (0.15,-0.5) {Ind};
\end{tikzpicture}
\caption{Case 1}
\label{fig:confluence:case1}
\end{subfigure}
&
\begin{subfigure}{0.5\textwidth}
\centering
\begin{tikzpicture}
\node at (-0.3,2) (cx) {$\config_x$};
\node at (0.9,2) (cy) {$\config_y$};
\node at ($(cy) + (0.6,0)$) {$= \config_v$};
\node[scale=2] at ($(cx) !0.5! (cy)$) {$\sim$};
\node at (-1.2,1) (cx1) {$\config_{x_1}$};
\draw[->] (cx) -- (cx1);
\node at (-1.8,-1) (cbarx') {$\bconfig_x'$};
\draw[->,dashed] (cx') -- node[midway,above,sloped] {$*$} (cbarx');
\node at (-1.8,0) (cx') {$\config_x'$};
\draw[->,dashed] (cx1) -- node[midway,above,sloped] {$*$} (cx');
\node at (0.3,1) (cu) {$\config_u$};
\node at (-0.3,0) (cw) {$\config_w$};
\draw[->,dashed] (cu) -- node[midway,above,sloped] {$*$} (cw);
\draw[->,dashed] (cx1) -- node[above,midway,sloped] {$*$} (cu);
\node at (-0.3,-1) (cbarw) {$\bconfig_w$};
\draw[->,dashed] (cw) -- node[midway,above,sloped] {$*$} (cbarw);
\node at (1.6,1) (cy1) {$\config_{y_1}$};
\draw[->] (cy) -- (cy1);
\node at (0.9,0) (cz) {$\config_z$};
\draw[->,dashed] (cy1) -- node[midway,above,sloped] {$*$} (cz);
\node at (0.9,-1) (cbarz) {$\bconfig_z$};
\draw[->,dashed] (cz) -- node[midway,above,sloped] {$*$} (cbarz);
\node at (2.1,0) (cy') {$\config_y'$};
\draw[->,dashed] (cy1) -- node[midway,above,sloped] {$*$} (cy');
\node at (2.1,-1) (cbary') {$\bconfig_y'$};
\draw[->,dashed] (cy') -- node[midway,above,sloped] {$*$} (cbary');
\node[scale=2,rotate=45] at ($(cy) !0.5! (cu)$) {$\sim$};
\node[scale=2] at ($(cw) !0.5! (cz)$) {$\sim$};
\node[scale=2] at ($(cbarx') !0.5! (cbarw)$) {$\sim$};
\node[scale=2] at ($(cbarw) !0.5! (cbarz)$) {$\sim$};
\node[scale=2] at ($(cbarz) !0.5! (cbary')$) {$\sim$};
\node at (-0.15,1.5) {$P2$};
\node at (0.9,1.0) {$P2$};
\node at (-1.05,-0.5) {Ind};
\node at (0.3,-0.5) {Ind};
\node at (1.55,-0.5) {Ind};
\end{tikzpicture}
\caption{Case 2a}
\label{fig:confluence:case2a}
\end{subfigure}
\\
\multicolumn{2}{c}{
\begin{subfigure}{1.0\textwidth}
\centering
\begin{tikzpicture}
\node at (-0.3,2) (cx) {$\config_x$};
\node at (0.9,2) (cy) {$\config_y$};
\node[scale=2] at ($(cx) !0.5! (cy)$) {$\sim$};
\node at (-2.4,1) (cx1) {$\config_{x_1}$};
\draw[->] (cx) -- (cx1);
\node at (-3.0,0) (cx') {$\config_x'$};
\draw[->,dashed] (cx1) -- node[midway,above,sloped] {$*$} (cx');
\node at (-3.0,-1) (cbarx') {$\bconfig_x'$};
\draw[->,dashed] (cx') -- node[midway,above,sloped] {$*$} (cbarx');
\node at (-1.8,0) (cu) {$\config_u$};
\draw[->,dashed] (cx1) -- node[midway,above,sloped] {$*$} (cu);
\node at (-1.8,-1) (cbaru) {$\bconfig_u$};
\draw[->,dashed] (cu) -- node[midway,above,sloped] {$*$} (cbaru);
\node at (-0.3,0) (cv) {$\config_v$};
\node at (-0.3,-1) (cbarv) {$\bconfig_v$};
\draw[->,dashed] (cv) -- node[midway,above,sloped] {$*$} (cbarv);
\node at (0.9,0) (cw) {$\config_w$};
\node at (0.9,-1) (cbarw) {$\bconfig_w$};
\draw[->,dashed] (cw) -- node[midway,above,sloped] {$*$} (cbarw);
\node at (2.1,0) (cz) {$\config_z$};
\node at (2.1,-1) (cbarz) {$\bconfig_z$};
\draw[->,dashed] (cz) -- node[midway,above,sloped] {$*$} (cbarz);
\node at (3.3,0) (cy') {$\config_y'$};
\node at (3.3,-1) (cbary') {$\bconfig_y'$};
\draw[->,dashed] (cy') -- node[midway,above,sloped] {$*$} (cbary');
\node at (2.7,1) (cy1) {$\config_{y_1}$};
\draw[->] (cy) -- (cy1);
\draw[->,dashed] (cy1) -- node[midway,above,sloped] {$*$} (cz);
\draw[->,dashed] (cy1) -- node[midway,above,sloped] {$*$} (cy');
\draw[dashed] (cy) -- (0.3,1);
\draw[->,dashed] (0.3,1) -- node[midway,above,sloped] {$*$} (cv);
\draw[->,dashed] (0.3,1) -- node[midway,above,sloped] {$*$} (cw);
\node[scale=2] at ($(cbarx') !0.5! (cbaru)$) {$\sim$};
\node[scale=2] at ($(cbaru) !0.5! (cbarv)$) {$\sim$};
\node[scale=2] at ($(cbarv) !0.5! (cbarw)$) {$\sim$};
\node[scale=2] at ($(cbarw) !0.5! (cbarz)$) {$\sim$};
\node[scale=2] at ($(cbarz) !0.5! (cbary')$) {$\sim$};
\node[scale=2] at ($(cu) !0.5! (cv)$) {$\sim$};
\node[scale=2] at ($(cw) !0.5! (cz)$) {$\sim$};
\node at (-2.4,-0.5) {Ind};
\node at (-1.05,-0.5) {Ind};
\node at (0.3,-0.5) {Ind};
\node at (1.5,-0.5) {Ind};
\node at (2.7,-0.5) {Ind};
\node at (-1.05,1) {$P2$};
\node at (1.5,1) {$P1$};
\end{tikzpicture}
\caption{Case 2b}
\label{fig:confluence:case2b}
\end{subfigure}
}
\end{tabular}

\caption{Cases for Confluence}
\label{fig:confluence}

\end{figure}
} % inlong

\begin{theorem}[Determinism]
\label{lem:confluence}
For any frontend program $\ep$ and backend
$\data$,
if
$\init(\ep, \data) \configreduce^* \tuple{\eagerdata_1; \emptyseq; \rstream_1; \val_1}$
and
$\init(\ep, \data) \configreduce^* \tuple{\eagerdata_2; \emptyseq; \rstream_2; \val_2}$
then
$\eagerdata_1 = \eagerdata_2$
and
$\dom(\rstream_1) = \dom(\rstream_2)$
and
$\forall \fvs \in \dom(\rstream_1) . \rstream_1(\fvs) = \rstream_2(\fvs)$
and
%$\rstream$ is a permutation of $\rstream'$ and
$\val_1 \feq \val_2$.
\end{theorem}
\inlong{
\begin{proof}
Follows from Lemma~\ref{lem:confluence}.
\end{proof}
}

This important result says that
despite the non-deterministic execution exhibited by the asynchronous processing between the frontend and the backend (see \S~\ref{sec:coresocial}),
despite the non-deterministic choices in propagation and realization in the backend (see \S~\ref{subsec:ingraph}),
despite non-deterministic executions over load expressions resulting from lazy realization (see \S~\ref{subsec:ingraph}),
despite the in-graph TLO (see \S~\ref{subsec:tlo}),
all executions produce the same final graph,
the same results,
and the same values modulo term equivalence in $\lambda$ calculus.
Here, term equivalence is needed because of the TLO rules such as fusing.
It is also important to observe this Theorem can only be established with the support of phase distinction.

Finally, eager graph processing (see \S~\ref{subsec:dcalculus}) can be modeled by redefining evaluation contexts \emph{without altering any reduction rules}. Intuitively, this means that eager data processing is a restrictive instance of \systemprefix{}. Rigorously, we represent eager processing as the $\eagerreduce$ reduction relation, defined as identical as the $\configreduce$ we introduced in Fig.~\ref{sec:highlevelsemantics}, except that the $\fE$, $\bE$, $\tE$, $\lE$ evaluation contexts are replaced
with $\eagerfE$, $\eagerbE$, $\eagertE$, $\eagerlE$ evaluation contexts in Fig.~\ref{fig:eagercontexts}.
We use $\eagerreducestar$
to represent the reflexive and transitive
closure of $\eagerreduce$.
\inshort{
We say a backend $\data$ is \emph{load-free} if any load expression in any station in $\data$ is a value.
}
For the eager task context $\eagertE$, we further require any element in the domain of its fulfillment function to be load-free.
%\dnote{need some edit}
\inlong{
\begin{lemma}[$\eagerreduce$ Deterministic Reduction]
\label{lem:determe}
For any program $\ep$ and backend $\eagerdata$,
if $\init(\ep, \eagerdata) \eagerreducestar \config$
where $\config = \tuple{\data; \ostream; \rstream; \ep}$
and
$\ep \neq \val$
then there exists a unique $\config'$ such that $\config \eagerreduce \config'$.
\end{lemma}
\begin{proof}
Induction on $\eagerreduce$, and case analysis on the last step of derivation.
The eager contexts limit each step to one option.
\end{proof}
}
\inshort{
A trivial case analysis will reveal $\eagerreduce$ is deterministic,
conforming to our intuition of one-at-a-time processing.
}

\begin{corollary}[\systemprefix{} With Regard to Eager Processing]
For any frontend program $\ep$ and backend
$\data$,
if
$\init(\ep, \data) \eagerreducestar \tuple{\eagerdata_1; \emptyseq; \rstream_1; \val_1}$%
and
$\init(\ep, \data) \configreduce^* \tuple{\eagerdata_2; \emptyseq; \rstream_2; \val_2}$%
then
$\eagerdata_1 = \eagerdata_2$
and
$\dom(\rstream_1) = \dom(\rstream_2)$
and
$\forall \fvs \in \dom(\rstream_1) . \rstream_1(\fvs) = \rstream_2(\fvs)$
and
$\val_1 \feq \val_2$.
% $\eagerdata = \eagerdata'$
% and
% $\dom(\rstream) = \dom(\rstream')$
% and
% $\forall \fvs \in \dom(\rstream) . \rstream(\fvs) = \rstream'(\fvs)$
% and
% %$\rstream$ is a permutation of $\rstream'$ and
% $\val = \val'$.
%
\end{corollary}

The simple corollary however carries an important message: the general, less restrictive graph processing of \systemprefix{} preserves the computation results of conventional graph processing.
In a nutshell, TLO and IOP are both \emph{sound} optimizations.

\inlong{%%%%%%%%%%%%%%%%%%%%%%%%% begin inlong %%%%%%%%%%%%%%%%%%%%%%%%%
\section{A Parallel Variant: \parlang{}}
\label{sec:parallel}

\begin{figure}[t]
\centering

$\begin{array}{@{}l@{\ \ }l@{\ \ }lr}

\mathbb{TF}
& \defassign &
\tuple{\data \concat \bullet \concat \data; \ostream; \rstream; \cE}
& \textit{task--frontend parallel context}
\\

\mathbb{LF}
& \defassign &
\tuple{\data \concat \mathbb{SL} \concat \data; \ostream; \rstream; \cE}
& \textit{load-frontend parallel context}
\\

\mathbb{TT}
& \defassign &
\tuple{\data \concat \bullet \concat \data \concat \bullet \concat \data; \ostream; \rstream; \ep}
& \textit{task-task parallel context}
\\

\mathbb{TL}
& \defassign &
\tuple{\data \concat \bullet \concat \data \concat \mathbb{SL} \concat \data; \ostream; \rstream; \ep}
& \textit{task-load parallel context}
\\ & \mid &
\tuple{\data \concat \mathbb{SL} \concat \data \concat \bullet \concat \data; \ostream; \rstream; \ep}
\\

\mathbb{LL}
& \defassign &
\tuple{\data \concat \mathbb{SL} \concat \data \concat \mathbb{SL} \concat \data; \ostream; \rstream; \ep}
& \textit{load-load parallel context}
\\ & \mid &
\tuple{\data \concat \mathbb{DL} \concat \data; \ostream; \rstream; \ep}
\\

\mathbb{SL}
& \defassign &
[\tuple{\cE; \ostream}]
& \textit{single load context}
%\\ & \mid &
%[\tuple{\tuple{\nkey; \ep; \bullet}; \ostream}]
\\ & \mid &
[\tuple{\inode; \ostream \concat [\fvs \mapsto \fold\ \func\ \cE\ \K] \cons \ostream}]
\\

\mathbb{DL}
& \defassign &
[\tuple{\ND{\cE; \cE; \ep}; \ostream}]
& \textit{double load context}
\\ & \mid &
[\tuple{\ND{\nkey; \cE; \cE}; \ostream}]
\\ & \mid &
[\tuple{\cE; \ostream \concat [\fvs \mapsto \fold\ \func\ \cE\ \K] \cons \ostream}]
% \\ & \mid &
% [\tuple{\ND{\nkey; \ep; \cE}; [\fvs \mapsto \fold\ \func\ \cE\ \K] \cons \ostream}]
\\

\end{array}$

\caption{\parlang{} Evaluation Contexts}
\label{fig:parcontexts}

\end{figure}

\begin{figure}[t]
\centering

\begin{tabular}{r}
\multicolumn{1}{@{}c@{}}{
\raisebox{-0.14cm}{
\begin{tikzpicture}
\coordinate (x) at (0,0);
\coordinate (y) at (\linewidth-\pgflinewidth,0);
\draw (x) -- (y);
\end{tikzpicture}
}
}
\\
\systembox{$\config \parreduce \config$}
\end{tabular}
\begin{mathpar}
\inferrule*[left=\ptitle{Serial}]{
   \config \configreduce \config'
}{
   \config \parreduce \config'
}

\inferrule*[left=\ptitle{TF}]{
   \mathbb{TF}[\data][\ep]
   \configreduce
   \mathbb{TF}[\data'][\ep] \writeresult \rstream
\\ \mathbb{TF}[\data][\ep]
   \configreduce
   \mathbb{TF}'[\data][\ep']
}{
   \mathbb{TF}[\data][\ep]
   \parreduce
   \mathbb{TF}'[\data'][\ep'] \writeresult \rstream
}

\inferrule*[left=\ptitle{LF}]{
   \mathbb{LF}[\ep_1][\ep_2]
   \configreduce
   \mathbb{LF}[\ep_1'][\ep_2]
\\ \mathbb{LF}[\ep_1][\ep_2]
   \configreduce
   \mathbb{LF}'[\ep_1][\ep_2']
}{
   \mathbb{LF}[\ep_1][\ep_2]
   \parreduce
   \mathbb{LF}'[\ep_1'][\ep_2']
}

\inferrule*[left=\ptitle{TT}]{
   \mathbb{TT}[\data_1][\data_2]
   \configreduce
   \mathbb{TT}[\data_1'][\data_2] \writeresult \rstream_1
\\ \mathbb{TT}[\data_1][\data_2]
   \configreduce
   \mathbb{TT}[\data_1][\data_2'] \writeresult \rstream_2
}{
   \mathbb{TT}[\data_1][\data_2]
   \parreduce
   (\mathbb{TT}[\data_1'][\data_2'] \writeresult \rstream_1) \writeresult \rstream_2
}

\inferrule*[left=\ptitle{LL}]{
   \mathbb{LL}[\ep_1][\ep_2]
   \configreduce
   \mathbb{LL}[\ep_1'][\ep_2]
\\ \mathbb{LL}[\ep_1][\ep_2]
   \configreduce
   \mathbb{LL}[\ep_1][\ep_2']
}{
   \mathbb{LL}[\ep_1][\ep_2]
   \parreduce
   \mathbb{LL}[\ep_1'][\ep_2']
}

\inferrule*[left=\ptitle{TL}]{
   \mathbb{TL}[\data][\ep]
   \configreduce
   \mathbb{TL}[\data'][\ep] \writeresult \rstream
\\ \mathbb{TL}[\data][\ep]
   \configreduce
   \mathbb{TL}[\data][\ep']
}{
   \mathbb{TL}[\data][\ep]
   \parreduce
   \mathbb{TL}[\data'][\ep'] \writeresult \rstream
}

\end{mathpar}

\caption{\parlang{} Operational Semantics}
\label{fig:parreduction}

\end{figure}

\dnote{consider moving this before meta theory; if so, we just need to rewrite the first sentence here} 

The combination of non-deterministic executions (\S~\ref{sec:highlevelsemantics}) and deterministic results (\S~\ref{sec:properties}) bring in significant opportunities for parallelism design. In this section, we introduce \parlang{}, a parallel calculus for continuous graph processing built on top of \systemprefix{}. Relation $\config \parreduce \config$ says configuration $\config$ one-step reduces to $\config'$ with parallel support, whose definition appears in Fig.~\ref{fig:parreduction}.
We use $\parreduce^*$
to represent the reflexive and transitive
closure of $\parreduce$.

Formally, this calculus naturally extends from \systemprefix{} without redefining any of the reduction rules of the latter. The most obvious "bridge" rule is \ptitle{Serial}, which says a serial reduction is also a reduction in our parallel calculus. The rest of all rules address parallelism. They rely on several newly introduced parallel evaluation contexts, shown in Fig.~\ref{fig:parcontexts}. Each context has two holes, which are two be fulfilled by two redexes in a parallel execution. It is not accidental that the forms of parallel evaluation contexts mirror those in \systemprefix{}: it is the non-deterministic choice of evaluation enabled by the latter that identifies the opportunities for parallel evaluation.

The parallelism between the backend and the frontend is captured by \ptitle{TF} and \ptitle{LF}.  There are two rules here, because the reduction that happens at the backend may either be a task reduction or a load reduction. We demonstrate backend-frontend parallelism with an example:

\begin{example}[Backend-Frontend Parallelism]
%\pnote{can we use $\fvs_c$? it's a little tricky though because, due to the encoding of these operations, some of the variables (like c in `addRelationship c b') are encoded as the target and need to be forced before being emitted, and others (like b in `addRelationship c b') are encoded as key sets but don't need to be forced before being emitted.}
Consider the following configuration 
illustrated in Fig.~\ref{fig:datacentriclaziness}(c),
where Line 11 of Fig.~\ref{fig:serverclient}
is about to be executed,
and the operations $\mathtt{addRelationship\ \nkey_\mathtt{deb}\ \nkey_\mathtt{cam}}$ and
$\mathtt{addRelationship\ \nkey_\mathtt{cam}\ \nkey_\mathtt{bob}}$ have both propagated to $\mathtt{eve}$.
\[
\lb
\begin{array}{l@{\hspace{-10em}}l}
[\tuple{\ND{\nkey_{\mathtt{eve}}; \integer_{\mathtt{eve}}; \KS{\emptyset}}; \\ &[[\fvs \mapsto \mathtt{addRelationship}\ \nkey_\mathtt{cam}\ \nkey_\mathtt{bob}], \\ &[\fvs' \mapsto \mathtt{addRelationship}\ \nkey_\mathtt{deb}\ \nkey_\mathtt{cam}]]},
\\
\tuple{\ND{\nkey_{\mathtt{deb}}; \integer_{\mathtt{deb}}; \KS{\emptyset}}; \emptyseq},
\\
\tuple{\ND{\nkey_{\mathtt{cam}}; \integer_{\mathtt{cam}}; \KS{\emptyset}}; \emptyseq},
\\
\tuple{\ND{\nkey_{\mathtt{bob}}; \integer_{\mathtt{bob}}; \KS{\emptyset}}; \emptyseq},
\\
\tuple{\ND{\nkey_{\mathtt{amy}}; \integer_{\mathtt{amy}}; \KS{\emptyset}}; \emptyseq}]
\end{array}
; \emptyseq; \rstream; \toserver\ \mathtt{addRelationship}\ \nkey_\mathtt{eve}\ \nkey_\mathtt{bob} \rb
\]
The frontend consists of a single operation, taken from Line 11
of Fig.~\ref{fig:serverclient}. According to rule \ptitle{TF}, the configuration may one-step reduce to the following where $\fvs''$ is fresh: 
\[
\lb
\begin{array}{l@{\hspace{-10em}}l}
[\tuple{\ND{\nkey_{\mathtt{eve}}; \integer_{\mathtt{eve}}; \KS{\emptyset}}; \\ &[[\fvs' \mapsto \mathtt{addRelationship}\ \nkey_\mathtt{deb}\ \nkey_\mathtt{cam}]]},
\\
\tuple{\ND{\nkey_{\mathtt{deb}}; \integer_{\mathtt{deb}}; \KS{\emptyset}}; \\ &[[\fvs \mapsto \mathtt{addRelationship}\ \nkey_\mathtt{cam}\ \nkey_\mathtt{bob}]]},
\\
\tuple{\ND{\nkey_{\mathtt{cam}}; \integer_{\mathtt{cam}}; \KS{\emptyset}}; \emptyseq},
\\
\tuple{\ND{\nkey_{\mathtt{bob}}; \integer_{\mathtt{bob}}; \KS{\emptyset}}; \emptyseq},
\\
\tuple{\ND{\nkey_{\mathtt{amy}}; \integer_{\mathtt{amy}}; \KS{\emptyset}}; \emptyseq}]
\end{array}
; [[\fvs'' \mapsto \mathtt{addRelationship}\ \nkey_\mathtt{eve}\ \nkey_\mathtt{bob}]]; \rstream; \fvs'' \rb
\]

\end{example}

% In the previous example the parallel units were split
% across two nodes.
% In reality, an implementation could allow parallel units
% to involve the same node as long as care is taken
% that the work modifies different parts of the node.
% For example an implementation could both append
% to the end of an operation store as well as remove an item
% from the front of the same operation store for propagation.
% \dnote{I suppose this paragaph is no longer needed? } 

% \begin{example}[Load Parallelism]
% Consider the following configuration,
% where the graph contains two nodes
% with payloads which can be reduced:
% \[
% \config = \tuple{\tuple{\tuple{\nkey; (\lambda \VAR . \ep)\ \val}; \emptyseq} \cons \tuple{\tuple{\nkey'; (\lambda \VAR . \ep')\ \val'}; \emptyseq}; \rstream; \ep}
% \]
% %
% There are two possible instances of the \plaintitle{Load} rule
% for the reduction $\config \configreduce \config'$:
% either the node with key $\nkey$ can reduce
% or the node with key $\nkey'$.
% \end{example}

In \ptitle{TT}, \ptitle{LL}, and \ptitle{TL}, \parlang{} allows parallel reductions between two tasks, two loads, or one task and one load, to happen. Recall that a task reduction may take on several forms, such as a realization (an instance of \plaintitle{Map} or \plaintitle{Fold}), a propagation (an instance of \plaintitle{Prop}), or a temporal locality optimization (an instance of \plaintitle{Opt}). Thus, \ptitle{TT} intuitively says that any of them can be evaluated in parallel with any other of them, as long as the two tasks do not overlap in the dam neighborhood. Similarly, \ptitle{LL} says any loads may reduce in parallel if they happen in different dams (following the first case of the $\mathbb{LL}$ definition based on $\mathbb{SL}$), or even inside the same dam as they reduce different loads in the dam (following the second case of the $\mathbb{LL}$ definition based on $\mathbb{DL}$). We now show task-task parallelism through an example:

% With eager, none.

% With lazy, the following rules
% all can be in parallel, assuming they
% don't modify the same neighborhood:
% \plaintitle{Lazy},
% \plaintitle{Final},
% \plaintitle{$\downarrow$},
% \plaintitle{Opt},
% \plaintitle{Map},
% and
% \plaintitle{Fold}.

\begin{example}[Task--Task Parallelism]
Consider the following configuration 
illustrated in Fig.~\ref{fig:datacentriclaziness}(f):
\[
\lb
\begin{array}{l}
[\tuple{\ND{\nkey_{\mathtt{eve}}; \integer_{\mathtt{eve}}; \KS{\emptyset}}; \emptyseq},
\\
\tuple{\ND{\nkey_{\mathtt{deb}}; \integer_{\mathtt{deb}}; \KS{\emptyset}}; [[\fvs' \mapsto \mathtt{addRelationship}\ \nkey_\mathtt{deb}\ \nkey_\mathtt{cam}]]},
\\
\tuple{\ND{\nkey_{\mathtt{cam}}; \integer_{\mathtt{cam}}; \KS{\emptyset}}; [[\fvs \mapsto \mathtt{addRelationship}\ \nkey_\mathtt{cam}\ \nkey_\mathtt{bob}]]},
\\
\tuple{\ND{\nkey_{\mathtt{bob}}; \integer_{\mathtt{bob}}; \KS{\emptyset}}; \emptyseq},
\\
\tuple{\ND{\nkey_{\mathtt{amy}}; \integer_{\mathtt{amy}}; \KS{\emptyset}}; \emptyseq}]
\end{array}
; \emptyseq; \rstream; \ep \rb
\]
According to rule \ptitle{TT}, the configuration may one-step reduce to the following configuration, illustrated in
Fig.~\ref{fig:datacentriclaziness}(h):
\[
\lb
\begin{array}{l}
[\tuple{\ND{\nkey_{\mathtt{eve}}; \integer_{\mathtt{eve}}; \KS{\emptyset}}; \emptyseq},
\\
\tuple{\ND{\nkey_{\mathtt{deb}}; \integer_{\mathtt{deb}}; \KS{\{\nkey_\mathtt{cam}\}}}; \emptyseq},
\\
\tuple{\ND{\nkey_{\mathtt{cam}}; \integer_{\mathtt{cam}}; \KS{\{\nkey_\mathtt{bob}\}}}; \emptyseq},
\\
\tuple{\ND{\nkey_{\mathtt{bob}}; \integer_{\mathtt{bob}}; \KS{\emptyset}}; \emptyseq},
\\
\tuple{\ND{\nkey_{\mathtt{amy}}; \integer_{\mathtt{amy}}; \KS{\emptyset}}; \emptyseq}]
\end{array}
; \emptyseq; \rstream; \ep \rb
\]
\end{example}

The most important property that \parlang{} enjoys is that it produces the same result as \systemprefix{}.

\begin{theorem}[Observable Equivalence between \parlang{} and \systemprefix{}]
\label{thm:parobserve}
For any well-typed program $\ep$ and graph $\data$,
if
$\init(\ep, \data) \eagerreducestar \tuple{\eagerdata_1; \emptyseq; \rstream_1; \val_1}$
and
$\init(\ep, \data) \parreduce^* \tuple{\eagerdata_2; \emptyseq; \rstream_2; \val_2}$
then
%$\rstream$ is a permutation of $\rstream'$ and
%$\val = \val'$.
then
$\eagerdata_1 = \eagerdata_2$
and
$\dom(\rstream_1) = \dom(\rstream_2)$
and
$\forall \fvs \in \dom(\rstream_1) . \rstream_1(\fvs) = \rstream_2(\fvs)$
and
%$\rstream$ is a permutation of $\rstream'$ and
$\val_1 \feq \val_2$.
%\dnote{will edit this later. I think this one perhaps can be proved by bisimulation? Or is it still conflucence? } 
\end{theorem}

\paragraph{Explicit Data Parallelism}

Before we move on, a question we wish to answer is how the parallelism support in \parlang{} is related to classic ``divide-and-conquer'' data parallelism. Indeed, this classic style \emph{could} be optionally built into \parlang{} through the
rules in Fig.~\ref{fig:classicd}.

\newcommand{\parnode}{P}
\newcommand{\pE}{\mathbb{P}}
\begin{figure}[t]
\centering

$\begin{array}{@{}l@{\ \ }l@{\ \ }lr}
\data
& \defassign &
\aseq{\station}
\mid \parnode || \parnode
\\
\parnode
& \defassign &
\tuple{\data; \ostream; \rstream}
& \textit{parallel dam unit}
\\
\pE
& \defassign &
\bullet
\mid \tuple{\pE; \ostream; \rstream; \ep}
\mid \tuple{\pE; \ostream; \rstream}
\mid \pE || \parnode
\mid \parnode || \pE
& \textit{parrallel context}
\\
\end{array}$

$\begin{array}{r@{\ }c@{\ }l}
(\parnode || \parnode') \writeoperation \batchgroup & \df &
  \parnode \writeoperation \batchgroup || \parnode' \writeoperation \batchgroup
  \\
\tuple{\data; \ostream; \rstream} \writeoperation \batchgroup & \df &
  \tuple{\data \writeoperation \batchgroup; \ostream \concat [\batchgroup]; \rstream}
  \\
\pE[\tuple{\bullet; \ostream; \rstream; \ep}] \writeresult \rstream' & \df &
  \pE[\tuple{\bullet; \ostream; \rstream \cup \rstream'; \ep}]
  \\
\pE[\tuple{\bullet; \ostream; \rstream}] \writeresult \rstream' & \df &
  \pE[\tuple{\bullet; \ostream; \rstream \cup \rstream'}]
  \\
\end{array}$
\begin{mathpar}
\inferrule*[left=\ptitle{CDSplit}]{}{
   \pE[\eagerdata_1 \concat \data_2]
   \parreduce
   \pE[\tuple{\eagerdata_1; \emptyseq; \emptyset} || \tuple{\data_2; \emptyseq; \emptyset}]
}

\inferrule*[left=\ptitle{CDJoin}]{}{
   \pE[\tuple{\eagerdata_1; \ostream_1; \emptyset} || \tuple{\data_2; \ostream_2; \rstream_2}]
   \parreduce
   (\pE \writeresult \rstream_2)[\eagerdata_1 \concat \data_2]
}

% \inferrule*[left=\ptitle{CDEmit}]{}{
%    \tuple{\data; \batchgroup \cons \ostream; \rstream; \ep}
%    \parreduce
%    \tuple{\data \writeoperation \batchgroup; \ostream; \rstream; \ep}
% }
% 
\inferrule*[left=\ptitle{CDClaimMap}]{
   \fvs \mapsto \map\ \func\ \K_0 \in \ostream_i\ \textrm{for}\ i=1,2
\\ \fvs \xmapsto{\K_i} 0 \in \rstream_i\ \textrm{for}\ i=1,2
\\ \K = (\K_0 \setminus \K_1) \setminus \K_2
}{
   \pE[\tuple{\data_1; \ostream_1; \rstream_1} || \tuple{\data_2; \ostream_2; \rstream_2}]
   \parreduce
   (\pE \writeresult \{ \fvs \xmapsto{\K} 0 \})[\tuple{\data_1; \ostream_1; \rstream_1 \setminus \fvs} || \tuple{\data_2; \ostream_2; \rstream_2 \setminus \fvs}]
}

\inferrule*[left=\ptitle{CDClaimFold}]{
   \fvs \mapsto \fold\ \func\ \val_0\ \K_0 \in \ostream_i\ \textrm{for}\ i=1,2
\\ \fvs \xmapsto{\K_i} \val_i \in \rstream_i\ \textrm{for}\ i=1,2
\\ \K = (\K_0 \setminus \K_1) \setminus \K_2
\\ \val = \func\ \val_1\ \val_2
\\ \tuple{\func; \val_0}\ \textrm{forms a monoid}
}{
   \pE[\tuple{\data_1; \ostream_1; \rstream_1} || \tuple{\data_2; \rstream_2}]
   \parreduce
   (\pE \writeresult \{ \fvs \xmapsto{\K} \val \})[\tuple{\data_1; \rstream_1 \setminus \fvs} || \tuple{\data_2; \ostream_2; \rstream_2 \setminus \fvs}]
}
% 
% \inferrule*[left=\ptitle{ClassicD}]{
%    \batchgroup = [\fvs \mapsto \processor]
% \\ \processor \neq \addn\ \integer
% \\ \tuple{\data_i; [\batchgroup]; \rstream; \ep}
%    \parreduce^*
%    \tuple{\data_i'; \emptyseq; \rstream; \ep} \writeresult \rstream_i\ \textrm{for}\ i = 1,2
% }{
%    \tuple{\eagerdata_1 \concat \data_2; \batchgroup \cons \ostream; \rstream; \ep}
%    \parreduce
%    \tuple{\data_1' \concat \data_2'; \ostream; \rstream; \ep} \writeresult ((\rstream_1 \cup \rstream_2) \setminus \fvs) \writeresult [\fvs \xmapsto{?} \fundef{join}(\processor, \rstream_1(\fvs), \rstream_2(\fvs)]
% }
\end{mathpar}

\caption{Classic Data Parallelism
\pnote{need rule for doing serial reduction on the split parts?}
\pnote{decision: extend C to include that form and extend the contexts to go into parallel units}
}
\label{fig:classicd}

\end{figure}

%\pnote{describe change in grammar and extension of $\writeresult$ and $\writeoperation$}

There are two things to note about the rules. First, observe that in this ``divide-and-conquer'' data processing style, the first partition of the graph ($\eagerdata_1$ in the rule \ptitle{CDSplit}) must be dry. Otherwise, any operation that remains in a streamlet in the first partition may never have a chance to be realized in the second partition of the graph, should its target happen to be there. More importantly, observe that ``divide-and-conquer'' data processing finally needs to consider how the results of processing from different partitions can be joined together, which happens to be operation-specific. As shown in \ptitle{CDClaimFold} , joining over $\fold$ results is only possible if the folding function happens to be \emph{associative}, and the initial value of $\fold$ applied to both partitions happens to be the \emph{identity}. In short, a monoidal requirement is needed.
Both \ptitle{CDClaimMap} and \ptitle{CDClaimFold} subtract the residual target key sets of both copies of the result from the initial target key set of the operation to produce a final residual target key set.
\ptitle{CDJoin} says that a previous split can be joined when every operation emitted to the two halves has been claimed.

\parlang{} does not include the \ptitle{CD-} rules by default.
%for two reasons. First, as \systemprefix{} and \parlang{} aim at defining essential features of streaming graphs, we are less interested in features that depend on the specifics of operations. Second, 
With significant parallelism opportunities as we described in Fig.~\ref{fig:parreduction}, the need for classic data parallelism may be less: as much as it is important to address how to parallelize the processing of one single operation (achieved by classic data parallelism), it is more unique in the context of continuous graph processing to elucidate how to parallelize the processing of all operations in a stream (achieved by \parlang{}). Nonetheless, we think the presentation of the \ptitle{CD-} rules is useful for placing \parlang{} in context: unlike classic data parallelism, all flavors of parallelism enabled by \parlang{} are \emph{agnostic} of specific properties of the operations.

%To support this form of parallelism,
%we require the graph to have no currently delayed operations.
%\pnote{also could probably take all operations delayed in $\data_a$ and send to $\data_b$, but that's extreme}
% We also restrict the result of $\fold$ operations
% essentially to a semigroup with associative binary operator $\oplus$.
% Namely, for operation $\fold\ \func\ \val\ \K$ we require
% \[
% \func\ \nkey_1\ \ep_1\ \ep_1'\ \val \oplus \func\ \nkey_2\ \ep_2\ \ep_2'\ (\func\ \nkey_3\ \ep_3\ \ep_3'\ \val)
% =
% \func\ \nkey_1\ \ep_1\ \ep_1'\ (\func\ \nkey_2\ \ep_2\ \ep_2'\ \val) \oplus \func\ \nkey_3\ \ep_3\ \ep_3'\ \val
% \]
% for any $\inode_1$, $\inode_2$, and $\inode_3$
% where $\inode_i = \tuple{\nkey_i; \ep_i; \ep_i'}$, for $i = 1,2,3$.

}%%%%%%%%%%%%%%%%%%%%%%%%% end inlong %%%%%%%%%%%%%%%%%%%%%%%%%

\inshort{%%%%%%%%%%%%%%%%%%%%%%%%% begin inlong %%%%%%%%%%%%%%%%%%%%%%%%%

\section{Technical Report}

The accompanying technical report under submission consists of 
%(i) the support of runtime typing;
(i) an extension to support explicit parallelism;
(ii) an extension to support explicit exception handling of ``key not found'' errors (recall that  \systemprefix{} is able to track ``key not found'' through residual target (\S~\ref{sec:opandresultstreams}), but our core system does not report errors upon claiming a result whose residual target is non empty);
(iii) an extension to support node addition at the end of the traversal sequence, or at a position in the traversal sequence (recall that our core calculus supports the simple form of adding a node at the beginning of the traversal sequence);
(iv) an extension to support ``map all'' and ``fold all''.
In addition to the mechanized proofs, we also provide a version of manual proofs in the technical report.

}%%%%%%%%%%%%%%%%%%%%%%%%% end inlong %%%%%%%%%%%%%%%%%%%%%%%%%

\inlong{%%%%%%%%%%%%%%%%%%%%%%%%% begin inlong %%%%%%%%%%%%%%%%%%%%%%%%%

\section{Discussion}
\label{sec:discussion}

We now sketch some dimensions of design that are orthogonal to the core in the first two subsections. The last one places the work in the context.

\subsection{Error Handling}

In \S~\ref{sec:syntax}, we discussed that our frontend language intentionally allows programmers to name a key as a primitive value. As a result, it is possible that an operation (such as a $\map$ and $\fold$) that name a key in its target key list, but that key cannot be found in the graph. 

The runtime of \systemprefix{} is capable of tracking if any key in the target key list is not processed. Recall that in \S~\ref{sec:opandresultstreams}, we described the residual target key list associated with each entry in the result store. In a language extension with explicit exception handling support, modeling ``key not found'' as an exception is a simple extension. The only change is to replace \plaintitle{Claim} with the following rules:
\begin{mathpar}
\scalebox{\scalemath}{$
\inferrule*[left=\plaintitle{ClaimNoResidual}]{
   \fE[\fromserver\ \fvs] = \tuple{\data; \ostream; \rstream; \ep}
\\ \fvs \xmapsto{\KS{\emptyseq}} \val \in \rstream
}{
   \fE[\fromserver\ \fvs]
   \configreduce
   \fE[\val]
}
$}

\scalebox{\scalemath}{$
\inferrule*[left=\plaintitle{ClaimWithResidual}]{
   \fE[\fromserver\ \fvs] = \tuple{\data; \ostream; \rstream; \ep}
\\ \fvs \xmapsto{\K} \val \in \rstream
\\ \K \neq \KS{\emptyseq}
}{
   \fE[\fromserver\ \fvs]
   \configreduce
   \fE[\mathbf{exception}(\K)]
}
$}
\end{mathpar}
where \textbf{exception}(K) is a value of this extended language. Now that the residual target key list is carried by the exception, a programmer can further inspect it and decide on what steps to be taken as part of exception handling.

With a type system in place, no type errors (such as mismatched function arguments) can happen in \systemprefix{}.

\subsection{Alternative Node Addition Support}
\label{subsec:add}

In \systemprefix{}, we choose to define a simple semantics for node addition: the newly added node always appears at the beginning of the station sequence. In this section, we define some alternatives. First, it is simple to support node addition at the \emph{end} of the station sequence, even though it is somewhat cluttering. To support this flavor, let us introduce the expression $\mathtt{addLast}\ \ep$, and extend the operation semantics with the following rules:
\begin{mathpar}
\inferrule*[left=\plaintitle{AddatLast}]{
   \tE = \data \concat \bullet
\\ \data_i = \tuple{\inode; \ostream_i}\ \textrm{for}\ i = 1,2
\\ \ostream_1 = [\fvs \mapsto \mathtt{addLast}\ \integer] \cons \ostream_2
\\ \nkey\ \fresh
}{
   \tE[\data_1]
   \configreduce
   \tE[\data_2 \concat [\ND{\nkey; \integer; \KS{\emptyseq}}]] \writeresult [\fvs \xmapsto{\KS{\emptyseq}} \nkey]
}

\inferrule*[left=\plaintitle{AddEmpty}]{ 
   \nkey\ \fresh
}{
   \tuple{\emptyseq; [\fvs \mapsto \mathtt{addLast}\ \integer] \cons \ostream; \rstream; \ep}
   \configreduce
   \tuple{\tuple{\ND{\nkey; \integer; \KS{\emptyseq}}; \emptyseq}; \ostream; \{\fvs \xmapsto{\KS{\emptyseq}} \nkey\} \cup \rstream; \ep}
}
\end{mathpar}
\noindent
with an extension to the $\target$ definition with one extra case, defining $\target(\mathtt{addLast}\ \integer)$ as $\emptyset$. With this design, when an $\mathtt{addLast}$ operation flows to the head of the top-level operation stream, \plaintitle{AddEmpty} will add the new node if the graph currently consists of no nodes. Otherwise, the operation will be propagated through \plaintitle{Prop}, until reaching the last station, where it is realized by \plaintitle{AddAtLast}. 

% Non-determinisic add

% (1) remove $\neq \addn\ \integer$ side condition from \plaintitle{Start};
% (2) add case for $\target(\addn\ \integer) \df \KS{\emptyseq}$;
% (3) remove the existing \plaintitle{Add} rule;
% (4) add the following two rules:

% \begin{mathpar}
% \inferrule*[left=\plaintitle{AddEmpty}]{ 
%   \nkey\ \fresh
% }{
%   \tuple{\emptyseq; [\fvs \mapsto \addn\ \integer] \cons \ostream; \rstream; \ep}
%   \configreduce
%   \tuple{\tuple{\tuple{\nkey; \ep; \KS{\emptyseq}}; \emptyseq}; \ostream; \{\fvs \xmapsto{\KS{\emptyseq}} \nkey\} \cup \rstream; \ep}
% }

% \inferrule*[left=\plaintitle{NDAdd}]{
%   \data_i = \tuple{\inode; \ostream_i}\ \textrm{for}\ i = 1,2
% \\ \ostream_1 = [\fvs \mapsto \addn\ \integer] \cons \ostream_2
% \\ \nkey\ \fresh
% }{
%   \tE[\data_1]
%   \configreduce
%   \tE[\data_2 \concat [\tuple{\tuple{\nkey; \integer; \KS{\emptyseq}}; \emptyseq}]] \writeresult [\fvs \xmapsto{\KS{\emptyseq}} \nkey]
% }
% \end{mathpar}
%\dnote{commented out non-deterministic add. It would have been nice if it could serve as some kind of intermediate generalization in the core calculus, but given this is an extension, it no longer makes much sense.} 

It is also simple to extend our calculus to support node addition at a fix possible. For example, let us introduce the expression $\mathtt{addAt}\ \ep\ \ep'$ where it says a new node with payload value $\ep$ should be placed after the node with key $\ep'$.
We can extend the operation semantics with the following rule:
\begin{mathpar}
\inferrule*[left=\plaintitle{FPAdd}]{
   \data = \tuple{\ND{\nkey; \integer; \K}; [\fvs \mapsto \addn\ \integer'\ \nkey] \cons \ostream}
\\ \nkey'\ \fresh
}{
   \tE[\data]
   \configreduce
   \tE[\data \concat [\tuple{\ND{\nkey'; \integer'; \KS{\emptyseq}}; \emptyseq}]] \writeresult [\fvs \xmapsto{\KS{\emptyseq}} \nkey']
}
\end{mathpar}
with an extension to the $\target$ definition with one extra case, defining $\target(\mathtt{addAt}\ \ep\ \nkey)$ as $\{\nkey\}$.

\subsection{MapAll and FoldAll}

In our core calculus, both $\map$ and $\fold$ operations are target-based: the last argument in both operations are the list of keys indicating which nodes in the graph will be processed. It is simple to extend our calculus with variants where all nodes in the graph will be processed. The support for the $\starmap$ and $\starfold$ operations are defined below: 
\begin{mathpar}
\inferrule*[left=\plaintitle{MapAll}]{
   \inode = \ND{\nkey; \ep; \ep'}
\\ \tE[\tuple{\inode; [[\fvs \mapsto \map\ \func\ \KS{[\nkey]}]]}]
   \configreduce
   \tE[\tuple{\inode'; [[\fvs \mapsto \map\ \func\ \KS{\emptyseq}]]}]
}{
   \tE[\tuple{\inode; [\fvs \mapsto \starmap^+\ \func] \cons \ostream}]
   \configreduce
   \tE[\tuple{\inode'; [\fvs \mapsto \starmap^-\ \func] \cons \ostream}]
}

\inferrule*[left=\plaintitle{FoldAll}]{
   \inode = \ND{\nkey; \ep''; \ep'''}
\\ \tE[\tuple{\inode; [[\fvs \mapsto \fold\ \func\ \ep\ \KS{[\nkey]}]]}]
   \configreduce
   \tE[\tuple{\inode; [[\fvs \mapsto \fold\ \func\ \ep'\ \KS{\emptyseq}]]}]
}{
   \tE[\tuple{\inode; [\fvs \mapsto \starfold^+\ \func\ \ep] \cons \ostream}]
   \configreduce
   \tE[\tuple{\inode; [\fvs \mapsto \starfold^-\ \func\ \ep'] \cons \ostream}]
}
\end{mathpar}

Note that in this formal definition, both operations carry a polarity (\texttt{+} and \texttt{-}). The \texttt{+} version of each operation is the programmer syntax, whereas the \texttt{-} version plays a formal role: it means that the operation has been realized at the current station (and is ready to be propagated to the next station). In both rules, the pre-reduction of the operation carries a \texttt{+} polarity, and the post-reduction one carries a \texttt{-} polarity. Predictably, the reduction of $\starmap$ depends on that of $\map$ in our core calculus. The same applies to $\starfold$. As the polarity turns from \texttt{+} to \texttt{-} during the reduction, it needs to be reversed during propagation. This latter is supported by redefining the $\writeoperation$ and $\writeoptwo$ operations as the following:
\[
\begin{array}{r@{\ }c@{\ }l}
\tuple{\data; \ostream; \rstream; \ep} \writeoperation \batchgroup & \df &
  \tuple{\data; \ostream \concat [[ \fvs \mapsto \fundef{r}(\processor) \mid \fvs \mapsto \processor \in \batchgroup ]]; \rstream; \ep}
  \\
\tuple{\inode; \ostream} \cons \data \writeoptwo \batchgroup & \df &
  \tuple{\inode; \ostream \concat [[ \fvs \mapsto \fundef{r}(\processor) \mid \fvs \mapsto \processor \in \batchgroup ]]} \cons \data
  \\
\fundef{r}(\processor) & \df & \begin{cases}
\starmap^+\ \func & \fif \processor = \starmap^-\ \func
\\
\starfold^+\ \func\ \ep & \fif \processor = \starfold^-\ \func\ \ep
\\
\processor & \textrm{otherwise}
\end{cases}
\\
\end{array}
\]

To allow \plaintitle{Prop} to propagate
$\starmap$ and $\starfold$ with polarity equal to \texttt{-},
we extend $\target(\processor)$ by defining
$\target(\starmap^-\ \func)$
and
$\target(\starfold^-\ \func\ \ep)$
as $\emptyset$.

}%%%%%%%%%%%%%%%%%%%%%%%%% end inlong %%%%%%%%%%%%%%%%%%%%%%%%%

\section{Related Work}
\label{sec:related}

\paragraph{\bf Continuous Processing in Graph Databases}

The TLO semantics of \systemprefix{} provides a language-based foundation for optimization in the presence of multiple operations. The latter is a well-studied area in databases known as multi-query optimization (MQO). QUEL*~\cite{10.1145/318898.318993} is an early compiler optimization defined with a number of tactics for inter-query optimization, such as combining two REPLACE operations in a relational query language into one. This is analogous to fusing in the style of the \htitle{FuseM} rule in \systemprefix{}. Sellis~\cite{10.1145/42201.42203} focuses on how to optimize queries that share common tasks and how local and global query execution plans can be constructed based on them. Park and Segev~\cite{105474} formulates the problem as sub-expression identification among queries. The essence of exploring commonality among queries is embodied by the \htitle{Reuse} rule in \systemprefix{}.
%rule \htitle{Reuse}.
%\dnote{if you delete the latter and rename the former, make sure update here}
These pioneer efforts lead to a large body of research on MQO-style query optimization, both for graph databases (e.g., ~\cite{6228123,10.14778/3021924.3021929}) and non-graph databases (e.g.,~\cite{Ramachandra:2012:HOP:2213836.2213852,Sousa:2014:CQU:2594291.2594305,Scully:2017:POA:3009837.3009891}).

\inlong{%%%%%%%%%%%%%%%%%%%%%%%%% begin inlong %%%%%%%%%%%%%%%%%%%%%%%%%

The aforecited efforts pioneered the research direction on MQO-based query optimization, leading to fruitful results. For example, query prefetching~\cite{Ramachandra:2012:HOP:2213836.2213852} resorts to static
analysis and transformation to combine queries and reduce the number of
round-trips to access data.
Sousa et al{.}~\cite{Sousa:2014:CQU:2594291.2594305} describes a calculus and a system that allow multiple queries in the form of user-defined functions to be consolidated. 
As a form of reusing, query results may also be cached, such as through Sqlcache~\cite{Scully:2017:POA:3009837.3009891}. 

Often deployed as the backend of data-intensive applications today, graph databases~\cite{10.1145/3104031} inherit the need for continuous processing. MQO is a common optimization for graph databases. For example, Le et al{.}~\cite{6228123} optimizes multiple SPARQL queries over RDF data. Ren and Wang~\cite{10.14778/3021924.3021929} defines heuristic MQO algorithms for identifying subgraph isomorphism in large graphs. In industry, deploying graph databases in use scenarios that demand continuous processing give rise to their popularity as Online Analytical Processing (OLAP) databases. In this domain, the boundary between graph databases and graph analytics/mining engines is often blurred. We will review the latter in the next subsection.

}%%%%%%%%%%%%%%%%%%%%%%%%% end inlong %%%%%%%%%%%%%%%%%%%%%%%%%

%Batching is a standard technique used in continuous graph database processing~\cite{golab2003issues}.\dnote{is this correct? Is this paper about graph databases?} 

%Query synthesis~\cite{Cheung:2013:ODA:2491956.2462180} addresses inefficiencies
%in object-relational mapping libraries by synthesizing efficient SQL queries.

\inlong{%%%%%%%%%%%%%%%%%%%%%%%%% begin inlong %%%%%%%%%%%%%%%%%%%%%%%%%

\systemprefix{} shares the motivation with existing work that the ``one operation at a time'' view misses out on important opportunities of data processing optimization. 
}%%%%%%%%%%%%%%%%%%%%%%%%% end inlong %%%%%%%%%%%%%%%%%%%%%%%%%

%The experimental effectiveness evidenced by existing work serves as the motivation on why TLO --- an adaptation of MQO in the context of online dynamic graph processing --- is an essential component of our foundational study on continuous graph processing. 
As a foundational study, \systemprefix{} complements existing work in several dimensions. (1) \systemprefix{} generalizes the optimization 
%view of multi-queries and multi-updates into \emph{many-queries} and \emph{many-updates}, by taking a stream view on operations. This principled view allows different forms of optimizations to be unified 
as a stream rewriting problem, subsuming both optimizations commonly studied in MQO (such as reusing and fusing), and those that are not (such as reordering and batching). (2) \systemprefix{} rigorously defines TLO as a \emph{part} of the semantics of the database engine, with the goal of 
%, and \emph{part} of the semantics for the backend-frontend interaction of data-intensive applications. 
%This integrated language-based approach 
demonstrating the role of the former in the broader scope of the latter. \systemprefix{} is not only a semantics for studying what forms of TLOs can be supported, but also (more importantly) for illuminating the design space on \emph{when} and \emph{where} TLOs may happen (see ~\S~\ref{sec:tlo}). 
\inlong{%%%%%%%%%%%%%%%%%%%%%%%%% begin inlong %%%%%%%%%%%%%%%%%%%%%%%%%
(3) \systemprefix{} positions MQO in \emph{open-world} use scenarios, when compile-time optimization may not be applicable. In \systemprefix{}, all forms of TLOs can happen at runtime, and they can be performed in an opportunistic manner, including when multiple operations are encountered \emph{in-graph}. This flavor of MQO based on dynamic semantics expands the scope of applicability of MQO for scenarios where queries/updates are unknown \emph{a priori}. Indeed, if one takes the general view in programming languages that compiler optimization is a semantics-preserving transformation, our semantics-based definition on TLOs entails compile-time optimization when operations are known statically. 
}%%%%%%%%%%%%%%%%%%%%%%%%% end inlong %%%%%%%%%%%%%%%%%%%%%%%%%

\paragraph{\bf Continuous Graph Analytics Systems}

\inlong{%%%%%%%%%%%%%%%%%%%%%%%%% begin inlong %%%%%%%%%%%%%%%%%%%%%%%%%
The need for efficiently analyzing and mining large graphs gives rise to the rapid development of graph analytics/mining systems in the recent decade.
}%%%%%%%%%%%%%%%%%%%%%%%%% end inlong %%%%%%%%%%%%%%%%%%%%%%%%%
Anther source of motivation of our foundation is the active area of experimental continuous graph analytics systems,  with a large body of work on both \emph{online graph analytics systems}~\cite{ediger2012stinger,Cheng:2012:KTP:2168836.2168846,10.1145/2567948.2580051,Vora:2017:KFA:3037697.3037748,10.1145/3267809.3267811,dhulipala2019low,10.1145/3364180,han2014chronos,shi2016tornado,196274,10.1145/2168836.2168854}
and \emph{iterative graph analytics systems}~\cite{kang2009pegasus,Malewicz:2010:PSL:1807167.1807184,low2012distributed,kyrola2012graphchi,10.1145/2517349.2522739,180251,190488}. 
Both TLO and IOP play an important role in scalable system design of these systems. 
%Online graph processing has received significant attention in recent years for two reasons: scalability and real-time processing capability. 
\inlong{%%%%%%%%%%%%%%%%%%%%%%%%% begin inlong %%%%%%%%%%%%%%%%%%%%%%%%%
Stinger~\cite{ediger2012stinger} supports high performance dynamic graphs through a data structure involving linked lists of blocks. UNICORN~\cite{10.1145/2567948.2580051} incrementally processes the graph with matrix-vector operations. 
}%%%%%%%%%%%%%%%%%%%%%%%%% end inlong %%%%%%%%%%%%%%%%%%%%%%%%%
Some examples of online graph processing systems are as follows: GraPU~\cite{10.1145/3267809.3267811} allows updates to the graph to be buffered and pre-processed, similar to a TLO operation in our top-level operation stream; Kineograph~\cite{Cheng:2012:KTP:2168836.2168846} supports a commit protocol for incremental graph updates; DeltaGraph~\cite{dexter2016lazy} allows graph operations to be batched and fused within the graph through a Haskell datatype representation of an inductive graph; %Their experiments show improvements on response rates in online graph processing.
Kickstarter~\cite{Vora:2017:KFA:3037697.3037748} incrementally corrects the error in their approximation result;
Reflective consistency~\cite{dexterreflective} incrementally defers synchronization for relaxed consistency. 
\inlong{%%%%%%%%%%%%%%%%%%%%%%%%% begin inlong %%%%%%%%%%%%%%%%%%%%%%%%%
C-Trees~\cite{dhulipala2019low} are purely functional data structures to enable efficient concurrent processing in the presence of queries and updates. GraphOne~\cite{10.1145/3364180} proposes a new versioning-based data storage scheme to address rapid evolution of graphs in a real-time setting. Other examples include Chronos~\cite{han2014chronos}, Tornado~\cite{shi2016tornado}, Version Traveler~\cite{196274}, and LazyBase~\cite{10.1145/2168836.2168854}. 
}%%%%%%%%%%%%%%%%%%%%%%%%% end inlong %%%%%%%%%%%%%%%%%%%%%%%%%
Supporting iterative graph analytics and graph mining is a central goal of graph processing research in the past decade. For example, PEGASUS~\cite{kang2009pegasus} supports iterative matrix vector mutiplication built upon MapReduce, similar to our choice of graph operations. Their framework demonstrates the MapReduce model is sufficient to support common graph analytics such as random walk, connected components, diameter estimation, as well as pagerank. 
%As \systemprefix{} also supports MapReduce-based operations, \systemprefix{} can encode the graph analytics that PEGASUS supports. 
As another example, Pregel~\cite{Malewicz:2010:PSL:1807167.1807184} introduces the notion of supersteps, which our \textsc{CorePR} example is based on. 
GraphLab~\cite{low2012distributed} focuses on distributed computing support for iterative graph-based machine learning algorithms.
\inlong{%%%%%%%%%%%%%%%%%%%%%%%%% begin inlong %%%%%%%%%%%%%%%%%%%%%%%%%
 Other examples of graph processing systems that support iterative graph analytics and graph mining are GraphChi~\cite{kyrola2012graphchi}, Galois~\cite{10.1145/2517349.2522739}, Powergraph~\cite{180251} and GraphQ~\cite{190488}. 
}%%%%%%%%%%%%%%%%%%%%%%%%% end inlong %%%%%%%%%%%%%%%%%%%%%%%%%
TLO is common among these systems. For example, both Pregel and GraphLab allow low-level MPI messages across graph partitions on different clusters to be batched, which can be viewed as a lower-level semantics-oblivious implementation of batching we generalize. 
Our treatment of continuously processed operations as an operation stream may lead to adaptive optimization opportunities
%, e.g., shown in the superstep blending example 
(e.g., Example~\ref{ex:superstepblend}).
%\dnote{can add more if we cited others}   
%Much of this research is on graph processing operations, a topic independent of our calculus. Our operation can be extended to them. Our model is MapReduce-based, akin to PEGASUS. 
\systemprefix{} provides support for dynamic graphs without any restrictions. This is the most general setting beyond existing work on static graphs (e.g., ~\cite{180251,low2012distributed,10.1145/2442516.2442530}), or on dynamic graphs with restrictions on updates (e.g.,~\cite{190488,Vora:2017:KFA:3037697.3037748,10.1145/3267809.3267811}).

%systems. Both tackle the performance and scalability challenges in the presence of dynamic graph change, with OLTP focusing on simple queries and updates, whereas OLAP focusing on graph analytic workloads. 

%These two directions of experimental research provide the backdrop which \systemprefix{} complements with a foundational perspective. 

% Core features in \systemprefix{} capture the recurring themes in continuous graph processing systems. 
% (1) TLO is common in online analytics systems, such as buffer preprocessing~\cite{?}. For iterative systems, both Pregel and GraphLab allow low-level MPI messages across graph partitions on different clusters to be batched, which can be viewed as a lower-level semantics-oblivious implementation of batching we generalize. (2) The treatment of continuously processed operations as an operation stream allows iterative analytics systems to be unified with online analytics systems. For the former, the stream view may lead to optimization opportunities
% %, e.g., shown in the superstep blending example 
% (e.g, Example~\ref{ex:superstepblend}) and principled support parallelism (see \S~\ref{sec:parallel}). (3) In-graph operation streams enable lazy and incremental operation processing, a common theme in online analytics systems; (4) 

Overall, \systemprefix{} complements existing work with a correctness-driven approach, in order to elucidate the invariants and principles in building continuous graph processing systems. Despite a minimal core, the calculus spans the scope of the graph processing engine, its optimization, and its interaction with a frontend programming model. 

%By taking a correctness-driven approach, \systemprefix{} elucidates the invariants that must be preserved for building continuous graph processing systems. 
%As demonstrated in existing work, a real-world graph processing system involves complex interactions between the graph engine design, optimization and parallelism support, and frontend programming model. 

\paragraph{\bf Data Streaming Systems and Foundations}

%At the heart of \systemprefix{} is an insight to model continuous graph processing as flowing the \emph{operation stream} through structured \emph{data} (i.e., graph). Its dual view is to flow the \emph{data stream} through structured \emph{operations}. This latter view is often known as 

\emph{Data streaming systems} have a model where a stream of \emph{data} flow through data processing operations (often called stream processors) composed together through framework-defined combinators. This is a well established area, including data flow and data streaming languages~\cite{lucid,lustre,thies2002streamit,10.1145/1297027.1297043,Meyerovich:2009:FPL:1640089.1640091,10.1007/978-3-662-44202-9_15}, data flow processing frameworks~\cite{murray2011ciel,Zaharia:2016:ASU:3013530.2934664,10.1145/2517349.2522737,murray2013naiad}, and foundations~\cite{1458143,10.1145/1111037.1111054,10.1145/2660193.2660225}. Graph nodes may also form a stream~\cite{186216,roy2013x,10.1145/1065167.1065201,mcgregor2005finding,10.1145/1376616.1376634,10.14778/1920841.1920964,10.1145/1970392.1970397,10.1145/2745754.2745763}. 
Thanks to the fundamental difference between operation streams and data streams, \systemprefix{} explores a different design space. For example, for operation streams, TLO is an essential design component, %e.g., batching, reordering, fusing, or reusing. It 
which does not find natural parallels in data streaming systems. For example, the order in the data stream often does not matter, so it makes little sense to reorder data in the stream, or in general, rewrite them through a rule-based system as we did in Fig.~\ref{fig:fusion}. 

\paragraph{\bf Foundations on Incremental Processing}

%We now briefly summarize founational support. 
%Sloth~\cite{Cheung:2014:SLV:2588555.2593672} is a database system that adopts a notion of extended lazy evaluation to dynamically batch SQL queries.

Haller et al{.}~\cite{haller2018programming}
describes a formal semantics and lineage-based programming model for
%Spark~\cite{Zaharia:2016:ASU:3013530.2934664}-inspired
 distributed data processing. In their model, deferred evaluation is supported 
 at the boundary of distributed nodes to promote opportunities for operation 
 fusion and improve 
the efficiency of network communications. 
%Their calculus has a different design space to explore (distribution and fault tolerance), but 
%we share the philosophy of lazy operation processing. 
More broadly, incremental computing
systems~\cite{Pugh:1989:ICV:75277.75305,Acar:2006:AFP:1186632.1186634,Hammer:2014:ACD:2594291.2594324,harkesicedust,harkes2017icedust}
maintain the propagation latent in the control flow of a function, and
efficiently perform re-computation along the propagation path only when
necessary.
If we view a graph-processing operation as a function and the graph it
operates on as the argument, \systemprefix{} at its essence calls for a dual propagation design: the function propagates within the argument.
\inlong{
Our general relationship with classic lazy evaluation was described in \S~\ref{sec:dimenlazy}.
}
We reviewed incremental processing support in experimental graph processing systems earlier while discussing online analytics systems. 

% Lazy data processing is broadly related to lazy evaluation, a classic semantic
% feature in programming languages~\cite{Wadsworth1971SemanticsAP,Henderson:1976:LE:800168.811543,Friedman1976CONSSN}.
% In the call-by-name $\lambda$ calculus, the
% argument of a function may not be evaluated upon function application;
% with call-by-need the repeated evaluations are also avoided.
% We draw high-level inspiration from lazy evaluation,
% where some informal parallels exist.
% The delay of evaluation is reflected in our design
% where an operation is not processed in a monolithic step
% that involves the traversal of the entire graph, as in eager processing.
% In addition, if multiple operations depend on the same operation,
% our semantics guarantees that the processing of the common operation
% only happens once, whose results are spliced in-graph to all operations
% that depend on it.
% %

% \S~\ref{sec:dimenlazy}

\paragraph{\bf Graph Calculi}

Our calculus has a more distant relationship with languages and calculi designed specifically for graph construction and verification, with several examples as follows. Oliveira and Cook~\cite{oliveira:fps} studies the functional representation of graphs and demonstrates the benefit of embedding inductive graphs through higher-order abstract syntax. Fregel~\cite{10.1145/2951913.2951938} is a functional programming language where Pregel-style vertex computations are abstracted as higher-order functions. Hobor and Villard~\cite{10.1145/2429069.2429131} reasons about ramifications for shared data structures, including graphs, through separation logic. Raad et al{.}~\cite{10.1007/978-3-319-47958-3_17} describes a reasoning technique for verifying the correctness of concurrent graph-manipulating programs.

\section{Conclusion}

% This paper introduces \systemprefix{}, a foundation
% for lazy streaming graph processing.
% %
% \systemprefix{} features fine-grained in-data lazy processing
% which allows for optimizations such as batching, fusion,
% and splicing.
%

Designing correct continuous graph processing systems with expressive optimization support of temporal locality optimization and incremental operation processing is a challenging
problem in data management.
%The foundational route we
%take is less explored in % an exciting domain rich in experimental research.
%a domain rich in experimental research.
\systemprefix{} complements existing experimental research
by rigorously establishing \emph{correctness} while promoting \emph{expressiveness}.
%
%The soundness of \systemprefix{} is established
%through a bisimulation to eager continuous graph processing.

% In the future,
% we would like to develop an analytical model
% to understand the performance impact of lazy graph processing.
% %
% We would also like to build a more general framework
% that allows a family of lazy data processing designs to be validated,
% beyond lazy processing as an optimization and beyond graphs
% as a data structure.

\bibliography{references}

% \clearpage
% \appendix
% \input{supmat}

\end{document}